\documentclass[12pt]{article}
\usepackage[vmargin={2.5cm,2.5cm},hmargin={2.5cm}]{geometry}

\usepackage{natbib}
\usepackage[onehalfspacing]{setspace}
\usepackage{subcaption}
\input{definitions}

\newcommand{\mitSloan}{Massachusetts Institute of Technology, Sloan School of Management, \texttt{bonatti@mit.edu}}
\newcommand{\mitIDSSd}{Massachusetts Institute of Technology, Institute for Data, Systems, and Society, \texttt{dahleh@mit.edu}}
\newcommand{\mitIDSSh}{Massachusetts Institute of Technology, Institute for Data, Systems, and Society, \texttt{thibauth@mit.edu}}

\begin{document}


\title{Information Design with Elicitation\\and Strategic Coordination\thanks{We would like to acknowledge the initial contributions that Amir Nouripour made to the paper in its early stages. We are appreciative of some of the calculations and simulations he provided using the model. We thank seminar audiences at ASSA 2025, Bristol, LSE, MIT, Oxford, Toronto, TSE, and Yale, as well as Ian Ball, Dirk Bergemann, Gabriel Carroll, Roberto Corrao, Bob Gibbons, Stephen Morris, Jacopo Perego, Jean Tirole, Mike Whinston, Alex Wolitzky, as well as our discussant, Ilya Segal, for helpful comments and suggestions, and Tianyi Zhang for excellent research assistance. Bonatti acknowledges financial support from NSF grant SES 2519401. Dahleh and Horel acknowledge support from OCP Group and NIH grant R01AG058063.}}

\author{
Alessandro Bonatti\thanks{\mitSloan} \and
Munther A.\ Dahleh\thanks{\mitIDSSd} \and
Thibaut Horel\thanks{\mitIDSSh}
}

\maketitle


\begin{abstract}
We study linear–quadratic games of incomplete information with Gaussian uncertainty, where each player's payoff depends on a privately observed type and a common state. The designer observes the state, elicits types, and sells action recommendations. We characterize all implementable mechanisms with Gaussian joint distributions of actions and fundamentals, and identify the players-optimal, consumer-optimal, and revenue-maximizing designs. In games of strategic complements (substitutes), these optimal mechanisms maximally correlate (anticorrelate) players’ actions. When type uncertainty is large, recommendations become deterministic linear functions of the state and reports, but remain only partially revealing.
\bigskip

    \bigskip

\bigskip

\end{abstract}\newpage



\section{Introduction}\label{sec:intro}

In a  growing class of competitive environments, centralized digital platforms pool data from market participants, observe the prevailing fundamentals, and sell back action recommendations against subscription or per-transaction fees. The recent antitrust action against RealPage in the multifamily-rental yield-management market is a salient instance, with parallel arrangements operating across markets for property insurance, B2B procurement, and inter-dealer brokerage.  Each platform performs two roles. First, it improves individual firms' decisions with superior information about market fundamentals. Second, it facilitates coordination among competing firms by issuing (paid) action recommendations that steer the industry toward more profitable (though potentially less competitive) equilibria.\footnote{A similar problem is faced by industry information aggregators such as Nielsen and Circana, who sell recommendations based on consumer-provided data rather than proprietary signals. We discuss all these examples at greater length in the \emph{Applications} subsection below.}

However, conveying aggregate demand information alone is insufficient to issue optimal action recommendations; even if all participants share a common prior over the underlying market fundamentals, successfully implementing coordinated  strategies also requires insights into each firm's private characteristics, such as quality or marginal cost. Thus, each of these platforms faces the same core problem: how to elicit the participants' private information by jointly designing the action recommendations and the appropriate fees that ensure truth-telling and compliance, in addition to raising revenue.

In this paper, we study platforms that act as information sellers  in markets where coordinated actions are critical. We consider a joint problem of information and mechanism design in games with quadratic utilities, Gaussian uncertainty, and privately observed idiosyncratic preference types. 
The designer’s tasks are  threefold: (i) to elicit the private types, (ii) to design signals about an unknown state, and (iii) to collect payments. We  provide a flexible framework to answer  key questions such as: (a) Which joint distributions of actions, states, and types can the designer implement with and without monetary transfers? (b) What are the optimal mechanisms for the data buyers' profits and for the platform's revenue? (c) How does the nature of competition influence the solution? (d) What are the implications for downstream consumers?


\paragraph*{Summary of the Results.} 
Our first contribution is to characterize the joint distributions of state, types, and actions that are implementable with and without transfers. We show that the designer's type-elicitation and action-recommendation problems are inseparable: 
incentive compatibility requires preventing double deviations, where a player misreports their type and deviates from the recommended action (\cref{prop:obedience-char-sym} and \cref{prop:ic-char}). In order to deter such deviations,  the joint distribution of state, types, and actions must (in our examples) exhibit a non-negative covariance between a player's reported type and their competitors' actions. Furthermore, implementing any distribution with non-zero covariance between reported types and competitors' actions requires discriminatory transfers; without transfers, only information about the state  can be provided (\cref{prop:ic-char}).


Our second contribution is to  exhibit the structural properties of the optimal mechanisms for arbitrary objective functions that are linear in the covariance parameters of the action-state-type distribution. We show that for any such objective, the players' actions are maximally correlated, positively or negatively, conditional on the fundamentals (\cref{prop:linear}).

The ex ante welfare of the participating players is an example of such an objective. We show that a designer who maximizes the players' profits and has access to discriminatory transfers can elicit the players' types without distorting the equilibrium action distribution---the players-optimal mechanism maximizes their expected profits under the obedience constraints only (\cref{prop:opt-welfare}).  
 Furthermore, if  the prior variance of the private types is large,
the players-optimal action recommendations are linear combinations of the fundamentals, yet they are not fully revealing (\cref{prop:welfare-random}).\footnote{The slackness of the incentive constraints at the players-optimal mechanism extends to several environments but fails in games of strategic complements when the designer wishes to anticorrelate the players' actions.}

The players-optimal mechanism has immediate implications for our motivating examples: under the consortium structure that characterizes industry-benchmarking and pooled-data platforms (insurance ratemaking organizations, compensation-survey providers, source-to-pay vendors), the joint owners can implement the players-optimal action distribution by combining a Gaussian recommendation rule with discriminatory data-submission fees, without distorting actions away from the obedient optimum.\footnote{Since the IC-implementing transfers are pinned down only up to a player-specific constant, the consortium can balance the budget in expectation by rebating the average collected payment to each participating member, leaving the implemented action distribution unchanged.}

 Within the Gaussian class, we also characterize the revenue maximizing mechanism (\cref{prop:opt-revenue}). Both the players-  and the revenue-optimal action distributions systematically differ from that of a complete information game (\cref{prop:comparison}). For both strategic complements and strategic substitutes, the players-optimal design places greater emphasis on  private types, relative to the common state. Under strategic complements, this reliance on private types enhances coordination; under strategic substitutes, it fosters beneficial anti-coordination (e.g., in Cournot, it amplifies the differences in  output levels for firms with different costs).

 The comparison between the  players- and revenue-optimal mechanisms  is more nuanced, because a monopolist information seller wants to limit the players' rents, in addition to generating surplus. When the market profitability (as measured by an exogenous function of the prior means of state and types) is low, the monopolist understates the reliance of the mechanism on the private types, but qualitatively modifies the complete-information signal in the same direction as the players-optimum. When market  profitability is high, however, the opposite result holds---the revenue-maximizing mechanism is even more reliant on private types, relative to the common state.

 We contrast these results with the case of a designer aligned with downstream consumers. For Bertrand competition with private demand shifters (\cref{prop:bertrand-opt}), the consumer-optimal mechanism discloses no information about the common state or about competitors' types: the IC monotonicity binds and the optimum collapses to the prior Bayes--Nash equilibrium, implementable without discriminatory transfers. We trace out the Pareto frontier (i.e., the set of mechanisms that maximize a convex combination of players and consumer welfare) and use it to derive the implications of platforms' market power on the welfare of downstream consumers. Leveraging the comparison of the players- and revenue-optimal mechanisms, we obtain conditions under which, surprisingly, a revenue-maximizing monopolist can benefit consumers, relative to the case of (perfect) competition among platforms.
 
 In the Supplementary Appendix, we examine variations of our model that reflect the range of business models adopted by digital  platforms. In particular, \Cref{sec:participation} extends the model to settings where platforms act as pure 
information intermediaries and cannot exclude non-participating players, and 
\Cref{sec:delegation} examines the case in which players fully delegate their actions
to the platform. 

\paragraph*{Applications.} 
Our model captures several real-world information markets in which a platform collects granular participant data and returns decision guidance. The examples below differ in institutional detail, but share the same economic structure: recommendations are useful because they combine aggregate market information with firm-level inputs that remain privately observed absent data submission.

For example, in industrial benchmarking and pricing intelligence, a statistical agent or trade intermediary pools transaction-level data and sells back industry-specific benchmarks or recommended inputs. The canonical case is property and casualty insurance: Verisk's ISO produces advisory loss costs, rules, and rating content used in carrier ratemaking.\footnote{See National Association of Insurance Commissioners, ``NAIC Loss Cost Bulletins--Rates and Forms Filing,'' \url{https://content.naic.org/industry_rates_forms_loss_cost.htm}.} Closely related benchmarking products appear in consumer packaged goods, where NielsenIQ and Circana aggregate scanner and panel data. These are price-setting environments with strategic complements: higher recommended prices or benchmarked terms for one firm raise the attractiveness of higher prices or terms for its rivals.

In yield management and dynamic pricing, revenue-management vendors such as RealPage observe market-level demand indicators and property-level inputs, then sell recurring price recommendations through subscription or performance-linked contracts.\footnote{See U.S. Department of Justice, Antitrust Division, ``U.S. and Plaintiff States v. RealPage, Inc.,'' \url{https://www.justice.gov/atr/case/us-and-plaintiff-states-v-realpage-inc}.} The downstream game is differentiated price competition with strategic complements.

In wholesale procurement and B2B sourcing, source-to-pay platforms such as SAP Ariba and Coupa collect bids, supplier attributes, contract terms, and buyer requirements, then generate sourcing, bid-comparison, and award recommendations.\footnote{See Gartner Peer Insights, ``Source-to-Pay Suites Reviews and Ratings,'' \url{https://www.gartner.com/reviews/market/source-to-pay-suites}.} Related data intermediaries operate in freight and commodity procurement, where capacity and cost inputs also matter. The downstream game is split-award contracting with strategic substitutes: assigning more volume or capacity to one participant reduces the marginal value of assigning the same volume or capacity to another.

Finally, in inter-dealer brokerage and fixed-income data services, MarketAxess aggregates dealer quotes, trades, and inventory information to produce composite reference prices and execution analytics.\footnote{See LSEG, ``MarketAxess,'' \url{https://www.lseg.com/en/data-analytics/financial-data/pricing-and-market-data/fixed-income-pricing-data/government-and-corporate-bonds/marketaxess}.} The downstream game is inventory unwinding under price impact, again with strategic substitutes.

Each class shares our four main modeling assumptions. A third-party designer with central access to pooled data commits to a recommendation rule and a payment schedule. Each player's best response depends on three variables: the common fundamental, the participant's own private type, and the actions expected from other participants. Each player has a private type--cost, demand-side position, inventory imbalance, or capacity--that the platform recovers only through data submission. Discriminatory transfers take the form of subscription tiers, data-submission fees, transaction charges, or rebates. Both regimes of strategic interaction recur across the examples: complements in benchmarking and yield-management environments, substitutes in procurement and inter-dealer brokerage.

\paragraph*{Related Literature.}

Our paper is closely related to several strands of literature. First, we build on the work of  \citet{BM13}, \citet{bhm15}, \citet{AB19}, and \citet{mapeta20} on information design in games.\footnote{See also \citet{BM19} and \citet{KA19} for surveys of work in this area.} While a special case of our framework (without monetary transfers and private information) coincides with \cite{BM13}, the introduction of privately known types and discriminatory transfers substantially broadens both the scope of the model and its range of applications.

In recent contributions, \cite{smya22} adopt a duality approach to study  information design for competing players with a continuum of actions in linear-quadratic games when the designer's objective is also  quadratic. \cite{mu2024} show that information design reduces to a finite-dimensional semidefinite programming and also characterize optimality via dual certificates.

Second,  our paper contributes to the literature on mechanism design with externalities pioneered by  \citet{segal99} and \citet{jemo00,jemo06}. Relative to all these papers, designing information for competing buyers involves both private types and hidden actions.

Third, viewed as a model of designing and pricing information for a privately informed receiver, our paper adds competing buyers to the mechanism design approach of \citet{BA12}, \citet{komzl17}, \citet{kolo18}, \citet{BBM15}, \citet{BBS18}, \citet{lishxu21}, and \citet{yang22}, with \citet{BBS18} being the single-buyer analog of our setting (albeit with finite actions and states); privately known types to the approach of \cite{adpf86,adpf90}, \citet{BCT19}, and \cite{elgkw24}; and a coordination motive to the setting of \citet{ADHR20}. Relative to \citet{RO21} and \citet{BDHN22}, who consider dominant-strategy games with binary states and actions, our model introduces a  coordination motive for selling information. Our setting is also one of partial and mediated information sharing, in contrast to the complete and voluntary information sharing in the literature of \citet{kirby88}, \citet{vives90}, and \citet{raith96}. 

Finally, our paper contributes to the literature that studies the role of public and private information in determining a firm's ability to exercise market power \citep{vives2002,vives2011,mywa15,behm21,royo19,royo25}, the social value of information \citep{mosh02,anpa07,heve09,yang15}, and the trade-off between adaptation and coordination in multi-division organizations \citep{aldema08,rant08}. The focus on information design with elicitation of agents' private information differentiates our setting from these important papers.


\paragraph*{Notation.} For a vector space $V$ and a subset of vectors
$S\subseteq V$, $\spn(S)$ denotes the linear span of $S$. For $n,m\geq 1$,
$\M_{n,m}(\R)$ denotes the vector space of $n\times m$ matrices with real
entries. For convenience, we write $\M_n(\R)$ when $m=n$ and implicitly
identify $\M_{n,1}(\R)$ with $\R^n$. The identity and all-ones matrices of
$\M_n(\R)$ are denoted respectively by $I_n$ and $J_n$. Finally, $1_n$ denotes
the all-ones vector in $\R^n$, $\S_n^+(\R)$ denotes the cone of positive semidefinite matrices,  and we write $[n]\eqdef\set{1,\dots,n}$.

Unless stated otherwise, all random variables in this paper are assumed to be
defined on the same sample space $(\Omega, \mathcal{F}, \P)$. For random
variables $X\in\R$ and $Y\in\R$,
$\cov(X,Y)\eqdef\E\big[(X-\E[X])(Y-\E[Y])\big]$ denotes the covariance between
$X$ and $Y$ and $\var(X)\eqdef\cov(X,X)$ is the variance of $X$. We also
alternatively write $\sigma_X^2$ for $\var(X)$ and $\sigma_{XY}$ for
$\cov(X,Y)$. By extension, for random vectors $X\in\R^n$ and $Y\in\R^m$,
$\cov(X,Y)\eqdef\E\big[(X-\E[X])\tr{(Y-E[Y])}\big]$ denotes the
cross-covariance matrix of $X$ and $Y$, that is, the matrix in $\M_{n,m}(\R)$
whose entry $(i,j)$ is $\cov(X_i,Y_j)$. Finally, $\var(X)\eqdef \cov(X,X)$
is the covariance matrix of $X\in\R^n$.

\section{Model}\label{sec:model}

\subsection{Basic Game}\label{sec:game}

\paragraph*{Actions and Payoffs.} We consider $n$ players who compete in a game of
incomplete information. In this game, each player $i\in[n]$ has a private-value type $\theta_i$ and faces an unknown (common) payoff-relevant state $\omega$. We write $u_i(a; \theta_i, \omega)$ for the payoff of player $i$ given action profile $a\in\R^n$, type $\theta_i\in\R$ and state $\omega\in\R$.

We restrict ourselves to symmetric games with quadratic
payoffs. 
As in \citet{BM13}, we assume that each player $i$ has a linear best response to $a_{-i}$  given $\omega$ and $\theta_i$:
\begin{equation}\label{eq:best-response}
	a_i = r\sum_{j\neq i}a_j + s\omega + t\theta_i,
\end{equation}
where  $t\neq 0$ and $r\in(-1,\frac{1}{n-1})$.\footnote{This open-interval restriction ensures two properties used throughout: the matrix $J_n(1,-r)$ governing the complete-information best-reply system is invertible (\cref{prop:matrix}), so the Nash equilibrium of the complete-information game exists and is unique (cf.~\textsection\ref{Sec:bench}); and the set of symmetric obedient mechanisms is compact (\cref{prop:obedience-char-sym}), hence the firms-, revenue-, and consumer-optimal problems are bounded. In the exterior of this interval, these optimization problems become unbounded (\cref{rem:compactness}), and at the endpoints, obedient mechanisms (including the Nash equilibrium) fail to exist generically.} The sign of the coefficient $r$  determines whether actions are strategic complements or substitutes. 

To compute payoffs, we assume that the  best response function \eqref{eq:best-response} for each player $i$ is generated by the following utility function:
\begin{equation}\label{eq:utility}
	u_i(a;\theta_i,\omega) = -\frac 1 2 a_i^2 + ra_i\sum_{j\neq i}a_j
	+(s\omega + t\theta_i)a_i.
\end{equation}
Equation~\eqref{eq:utility} is a canonical quadratic payoff that generates best response~\eqref{eq:best-response}. The restriction we impose lies in the term $r a_i\sum_{j\neq i}a_j$: cross-effects enter only through the bilinear interaction between own and others' actions, with no further quadratic terms.

We now provide two classic examples of this framework.
\begin{example}[Bertrand Competition]\label{example:bertrand}
	Firms produce differentiated goods and compete in prices $p_i$. The demand curve of good $i$ is
	\begin{displaymath}
		Q_i(p;\theta_i,\omega) = \omega + r\sum_{j\neq i} p_j - p_i/2 + \theta_i,
	\end{displaymath}
	with $r>0$, where $\omega$ is the common demand intercept and $\theta_i$ is firm $i$'s idiosyncratic demand-side parameter (e.g., resulting from hidden investment in quality, advertising, or marketing).
    Marginal costs are commonly known and normalized to zero. 
    Hence, firm $i$'s profit is 
	\begin{displaymath}
		u_i(p;\theta_i,\omega) = p_i\, Q_i(p;\theta_i,\omega),
	\end{displaymath}
	which is the canonical payoff \eqref{eq:utility} exactly, and the best response is $p_i = \omega + \theta_i + r\sum_{j\neq i} p_j$, satisfying \eqref{eq:best-response} with $s=t=1$. 
\end{example}

\begin{example}[Cournot Competition]\label{example:cournot}
	Firms produce goods that are (imperfect) substitutes. Let $q_i$ denote the
	quantity of good $i$ produced by firm $i$ and $\theta_i$ its marginal cost. Assuming a linear demand curve with
	symmetric substitution patterns,  $P_i(q)=\omega + r\sum_{j\neq i} q_j- q_i/2$, with $r<0$, denotes good $i$'s  inverse demand curve. The profit of firm $i$ is then given by
	\begin{displaymath}
		u_i(q) = q_i P_i(q) - \theta_i q_i\,,
	\end{displaymath}
	and its best response is $q_i  = \omega + r\sum_{j\neq i}q_j - \theta_i$,
	which is of the form \eqref{eq:best-response} with $s=-t=1$. 
\end{example}

\subsection{Information Structure and Mechanism Design}

We assume that the vector $(\theta,\omega)$ is drawn from an independent Gaussian prior distribution with means $\mu_\theta\eqdef \E[\theta] = (\mu_{\theta_i})_{i\in[n]}\in\R^n$ and $\mu_\omega\eqdef \E[\omega]\in\R$, and variances $\var(\theta)=\diag(\sigma_{\theta_1}^2,\dots,\sigma_{\theta_n}^2)$ and $\sigma_\omega^2$, respectively. Each player $i$ observes their type $\theta_i$, while the designer observes the state of nature $\omega$.

We allow bidirectional preplay communication between the players and the designer after they respectively observe their type and the state of nature. By the revelation principle for communication games \citep{M82}, given any communication system and any Bayesian equilibrium of the induced communication game, there is an equivalent \emph{direct} and \emph{incentive-compatible} mechanism in which each player gets the same utility as in the given Bayesian equilibrium at every type. Thus, it is without loss to assume that the designer first asks the players to report their type and then issues (possibly correlated) action recommendations to each player that are potentially informative of both the state and the competitors' types.




Thus, a direct mechanism consists of two functions. 
\begin{itemize}
    \item An information policy $\tau:\R^{n}\times\R\to\Delta(\R^{n})$ that maps each player's type report and state of nature to a distribution over the actions of the $n$ players. We assume that for each measurable set $B\subseteq\R^n$, the map $(\theta,\omega)\mapsto\tau(\theta,\omega)(B)$ from $\R^n\times\R$ to $[0,1]$ is measurable. Thus, $\tau$ defines a Markov kernel from $\R^n\times\R$ to $\R^n$, or equivalently, a conditional distribution of action profiles given $(\theta,\omega)$.

     	
	 \item For each  player $i\in [n]$, a function $p_i:\R\to\R$ mapping their reported type to the payment they are being charged in exchange for an action recommendation.\footnote{Payments occur at the interim stage and hence need not condition on the other players'  reports or on the state. \Cref{lem:payment-interim} in \cref{sec:app-payments} justifies that this restriction is without loss of generality.}
\end{itemize}

Each player $i\in[n]$ chooses a pair of measurable functions $(\delta_{i,1},\delta_{i,2})$ with $\delta_{i,1}:\R\to\R$ and $\delta_{i,2}:\R\times\R\to\R$, such that $\delta_{i,1}(\theta_i)$ is player $i$'s type report when their true type is $\theta_i$, and $\delta_{i,2}(a_i,\theta_i)$ is player $i$'s final action in the game after receiving action recommendation $a_i$, when their true type is $\theta_i$. Thus, the ex ante expected utility of player $i\in[n]$ given an $n$-tuple $\big((\delta_{1,1},\delta_{1,2}),\dots,(\delta_{n,1},\delta_{n,2})\big)$ of such pairs is
\begin{displaymath}
    \E\left[u_i\big(\delta_{1,2}(a_1,\theta_1),\dots,\delta_{n,2}(a_n,\theta_n);\theta_i,\omega\big)-p_i\big(\delta_{i,1}(\theta_i)\big)\right]
\end{displaymath}
where the expectation is with respect to the prior distribution of $(\theta,\omega)$ and $(a_1,\dots,a_n)$ distributed as $\tau\big(\delta_{1,1}(\theta_1),\dots,\delta_{n,1}(\theta_n),\omega\big)$.
The direct mechanism $(\tau,p)$ is \emph{incentive-compatible} if it is an equilibrium for each player to choose the pair $(\delta_1^\star,\delta_2^\star)$ where $\delta_1^\star$ is the identity function (truthful reporting) and $\delta_2^\star:(a_i,\theta_i)\mapsto a_i$. Following \citet{M82}, the following definition gives an equivalent formulation of incentive-compatibility at the interim stage.

\begin{definition}[Incentive Compatibility]\label{def:ic}
    A mechanism $(\tau, p)$ is \emph{incentive compatible} if for each $i\in[n]$, $(\theta_i,\theta'_i)\in\R^2$, and all deviation functions $\delta:\R\to\R$,
\begin{align*}
    \E\big[u_i(a_i,a_{-i};\theta_i,\omega)\given \theta_i\big]-p_i(\theta_i)
 \geq
 \E\big[u_i(\delta(a'_i), a'_{-i};\theta_i,\omega)\given \theta_i\big]-p_i(\theta'_i),
\end{align*}
    where $a$ is distributed as $\tau(\theta,\omega)$ and $a'$ as $\tau(\theta_i',\theta_{-i},\omega)$.
\end{definition}


\paragraph*{Gaussian mechanisms.}
Throughout this paper, we restrict ourselves to Gaussian mechanisms.\footnote{In Sections \ref{sec:char} and  \ref{sec:opt}, we  highlight  which of our results  hold for arbitrary signal distributions.} In these mechanisms,  the information policy induces a joint Gaussian distribution of $(a,\theta,\omega)$. 
This joint distribution is characterized by the mean vector $\mu=\E[a,\theta,\omega]\in\R^{2n+1}$ and the covariance matrix $\K=\var(a,\theta,\omega)\in\M_{2n+1}(\R)$. Note that $\mu$ and $\K$ have the following block structure
\begin{equation}\label{eq:params}
	\mu = \begin{bmatrix}\mu_{a}\\\mu_{\theta}\\\mu_\omega\end{bmatrix}
	\quad\text{and}\quad
	\K = \begin{bmatrix}
		\K_{aa}&\K_{a\theta}&\K_{a\omega}\\[0.5ex]
		\tr{\K_{a\theta}}&\K_{\theta\theta}&0\\[0.5ex]
		\tr{\K_{a\omega}}&0&\sigma_{\omega}^2\\
	\end{bmatrix}\,.
\end{equation}
The means $\mu_\theta\eqdef \E[\theta] = (\mu_{\theta_i})_{i\in[n]}\in\R^n$ and $\mu_\omega\eqdef \E[\omega]\in\R$ are given by the prior distribution and the vector $\mu_a\eqdef\E[a]=(\mu_{a_i})_{i\in[n]}\in\R^n$ is chosen by the designer. Similarly, $\K_{\theta\theta} = \var(\theta)=\diag(\sigma_{\theta_1}^2,\dots,\sigma_{\theta_n}^2)$ and $\sigma_\omega^2$ are given by the prior distribution, whereas $\K_{aa} \eqdef \var(a)\in\M_n(\R)$, $\K_{a\theta}\eqdef\cov(a,\theta)\in\M_n(\R)$ and $\K_{a\omega}\eqdef\cov(a,\omega)\in\R^n$ are chosen by the designer.

A standard property of multivariate normals is that their conditional expectations are linear, hence the information policy $\tau$ is a linear Gaussian kernel: action recommendations are affine functions of the fundamentals $(\theta,\omega)$ to which zero-mean (but possibly correlated) noise is added:
\begin{equation}\label{eq:params-bis}
	a_i = \alpha_i + \beta_i(\omega - \mu_\omega)
	+ \sum_{j\in[n]}\gamma_{ij}(\theta_j-\mu_{\theta_j}) + \eps_i,
\end{equation}
for all $i\in[n]$ and where $\eps=(\eps_i)_{i\in[n]}$ is a zero-mean multivariate normal $\mathcal{N}(0,\K_\eps)$ independent of $(\theta,\omega)$. Equations~\eqref{eq:params} and \eqref{eq:params-bis} give two different parametrizations of the joint distribution of $(a,\theta,\omega)$. Writing $\alpha=(\alpha_i)_{i\in[n]}$, $\beta=(\beta_i)_{i\in[n]}$ and $\Gamma=(\gamma_{ij})_{(i,j)\in[n]^2}$, the two parametrizations identify the same distribution iff the following hold:
\begin{equation}\label{eq:params-dict}
	\mu_{a}=\alpha,
	\quad\K_{a\omega}=\sigma_{\omega}^2\beta,
	\quad \K_{a\theta}=\Gamma\K_{\theta\theta},
	\quad \K_{aa} = \sigma_\omega^2 \beta\tr{\beta}
	+ \Gamma\K_{\theta\theta}\tr{\Gamma} + \K_\eps.
\end{equation}

\paragraph*{Participation.} In our baseline model, we assume that all players participate in the mechanism (participation is ``forced''). In lieu of explicitly specifying the players' outside options, we only impose a nonnegativity constraint on each player's interim utility. In \cref{sec:participation}, we let players simultaneously choose whether to participate in the designer's mechanism, and any nonparticipating player still chooses an action in the downstream game. We consequently expand the definition of a mechanism to specify a distribution of action recommendations for each subset of participating players. We show that the set of symmetric, obedient, incentive-compatible Gaussian mechanisms coincides under the two interpretations: for any mechanism in our baseline setting, the designer can choose off-path action recommendations under endogenous participation that sustain full participation as an equilibrium, implementing the same joint distribution of $(a, \theta, \omega)$ and the same payments.

\section{Benchmarks}\label{Sec:bench}
We begin by covering three different benchmark settings that illustrate our designer's problem. We discuss these settings informally, and defer all formal results to  \cref{sec:char,sec:opt}.

\paragraph*{Single Agent Benchmark.} With a single agent $(n=1)$, the data buyer only wishes to learn the realization of the state $\omega$. In this case, the designer can achieve the first-best allocation of information and extract the entire surplus. The welfare- and revenue-optimal mechanism consists of  recommending action $a_i=s\omega+t\theta_i$ with probability one in each state, and charging a  price $p_i(\theta_i)$. Note, however, that the buyer's willingness for this information policy is given by $s^2\sigma^2_{\omega}/2$ for all types $\theta_i$. Therefore,  a constant price for the complete information structure yields truthful reporting of the agent's type and extracts the entire social surplus. 

\paragraph*{Complete Information Benchmark.} Now consider the case of $n>1$ and assume players have complete information about the state and about all types $\theta\in\mathbb{R}^n$. This game admits a unique Nash equilibrium. Collecting the best responses
\eqref{eq:best-response} for $i\in[n]$ yields the linear system
\begin{displaymath}
	J_n(1,-r)a = s\omega 1_n + t\theta,
\end{displaymath}
where $\theta=(\theta_1,\dots,\theta_n)\in\R^n$ is the vector of types and $J_n(1,-r)\eqdef I_n -r(J_n-I_n)$ is the $n\times n$ matrix with $1$ on the
diagonal and $-r$ off the diagonal. This matrix is invertible whenever
$r\notin\set{-1,\frac 1 {n-1}}$, in which case the
solution to the linear system is the unique Nash equilibrium of the complete information game:\footnote{We establish this property in \cref{prop:matrix} in \cref{sec:app-la}.}
\begin{equation}\label{eq:ne}
	a_i = \frac{s\omega+t\theta_i}{1-(n-1)r} + \frac{r\cdot t\sum_{j\neq
	i}(\theta_j-\theta_i)}{(1+r)(1-(n-1)r)}.
\end{equation}
\looseness=-1
This action profile plays an important role in the analysis that follows. 
First,  revealing a linear combination of the state and of the other players' types is sufficient for each player to play the complete-information equilibrium action. Second, if the designer knew the players' true types, it could recommend the actions in \eqref{eq:ne} and players would follow these recommendations by definition of an equilibrium. Third, each player could have an incentive to misreport their type, so to influence the action profile the designer recommends to their competitors. (For example, in a Cournot game, each firm wants to understate their cost and induce the designer to recommend lower quantities to its competitors.) In \cref{sec:ic}, we confirm that  recommending the complete information Nash equilibrium actions does, in fact, induce obedience but requires discriminatory transfers to induce  truthful reporting of the players' types.

\paragraph*{First-Best Benchmark.} If  $n>1$ players could coordinate on the jointly optimal action profile for every state and type vector, they would take  the first-best actions\footnote{This benchmark requires a stronger restriction on the range of $r$, namely $r\in\big(-\frac12,\frac 1{2(n-1)}\big)$. Indeed, outside this range the objective ceases to be concave and the first-best outcome is unbounded.}
\begin{equation}\label{eq:collusive_a}
    a_i^{\rm FB}:=\frac{s\omega+t\theta_i}{1-(n-1)2r}
	+\frac {2rt\sum_{j\neq i}(\theta_j-\theta_i)}{(1+2r)(1-(n-1)2r)}.
    \end{equation}
However, these actions are not best replies to one another. Thus, even if the designer could obtain truthful reports from the players, each one of them would have an incentive to deviate from the recommended course of action.

\section{Implementability}\label{sec:char}

In this section, we provide a complete characterization of incentive-compatible mechanisms (\cref{def:ic}). Any such mechanisms must, in particular, induce obedience at the second stage: conditional on truthful reporting in the first stage, each player must find it optimal to follow their recommendation. 
\Cref{prop:obedience-char-sym} characterizes the joint distributions of $(a,\theta,\omega)$ that arise from obedient recommendations; \cref{prop:ic-char} then characterizes the additional restrictions imposed by incentive compatibility, where types remain private and truth-telling is an equilibrium property rather than a primitive.

\subsection{Obedience Constraints}\label{sec:obedience}

A mechanism incentivizes \emph{obedience} if the recommended action is a best response for each player, conditioned on their type and their recommendation, i.e., if
\begin{displaymath}
    a_i\in\argmax_{a_i'\in\R}
    \E[u_i(a_i', a_{-i};\theta_i,\omega)\given a_i, \theta_i]
\end{displaymath}
almost surely for each player $i\in[n]$. For the  game in \cref{sec:game}, we have
\begin{displaymath}
    \E[u_i(a_i', a_{-i}; \theta_i,\omega)\given a_i, \theta_i]
    = -\frac 1 2 (a_i')^2 + r a_i'\sum_{j\neq i}\E[a_j\given a_i,\theta_i]
    + (s\E[\omega\given a_i,\theta_i]+t\theta_i)a_i',
\end{displaymath}
which is concave in $a_i'$. Under any mechanism, obedience of player $i\in[n]$ is thus equivalent to the first-order condition
\begin{equation}\label{eq:obedience}
    a_i = r\sum_{j\neq i}\E[a_j\given a_i, \theta_i]
    + s\E[\omega\given a_i,\theta_i] + t \theta_i\,,
\end{equation}
which yields a useful expression for player $i$'s expected payoff:
\begin{equation}\label{eq:obedience-utility-cond}
    \E[u_i(a_i; \theta_i, \omega)\given a_i, \theta_i] = \frac 1 2 a_i^2.
\end{equation}


We now focus on \emph{Gaussian} mechanisms, which are fully determined by the mean vector $\mu\in\R^{2n+1}$ and covariance matrix $\K\in\M_{2n+1}(\R)$ as described in \eqref{eq:params}, and we further restrict attention to symmetric mechanisms under a symmetric prior.\footnote{\cref{sec:symmetry-supp} shows that when the prior is symmetric, this restriction is without loss of generality for our optimality criteria of \cref{sec:opt}, because  the downstream game (\cref{sec:game}) is symmetric.} In this case,  there are only 6 degrees of freedom: any of the coordinates of $\mu_a$ and $\K_{a\omega}$, and the on- and off-diagonal entries of $\K_{aa}$ and $\K_{a\theta}$ (see \cref{lemma:sym-char} in \cref{sec:char-app} for details). In what follows, we choose an arbitrary $i\in[n]$ and $j\neq i$ and write these parameters as $\mu_{a_i}$, $\sigma_{a_i\omega}$, $\sigma_{a_i}^2$, $\sigma_{a_ia_j}$, $\sigma_{a_i\theta_i}$ and $\sigma_{a_i\theta_j}$. 

\begin{proposition}\label{prop:obedience-char-sym}
    Assume that $r\in\big(-1,\frac 1 {n-1}\big)$. Then, $\mu$ and $\K$ are the mean vector and covariance matrix of a \emph{symmetric} and \emph{obedient} mechanism iff
    \begin{enumerate}
	\item The mean action $\mu_{a_i}$ of each player $i\in[n]$ is determined by the prior's mean:
	    \begin{equation}\label{eq:obedience-mean-sym}
    \mu_{a_i} = \frac{s\mu_\omega+t\mu_{\theta_i}}{1-(n-1)r}.
	    \end{equation}
	\item The covariance matrix $\K$ satisfies the following linear constraints for each $i\in[n]$
	    \begin{equation}\label{eq:obedience-var-sym}
    \begin{cases}
	\sigma_{a_i}^2 = (n-1)r\sigma_{a_ia_j} + s\sigma_{a_i\omega}
	+ t\sigma_{a_i\theta_i},\\
	\sigma_{a_i\theta_i} = (n-1)r\sigma_{a_i\theta_j}
	+ t \sigma_{\theta_i}^2.
    \end{cases}
	    \end{equation}
    \end{enumerate}
\end{proposition}

The proof of \cref{prop:obedience-char-sym} in \cref{sec:obedience-app} 
 shows that the mean action recommendation $\mu_{a_i}$ is entirely determined by the mean of the prior \eqref{eq:obedience-mean-sym} and that the covariance coefficients must satisfy the two linear equality constraints \eqref{eq:obedience-var-sym}, effectively reducing the degrees of freedom of the mechanism to only 3 covariance parameters: $\sigma_{a_i\omega}$, $\sigma_{a_i\theta_j}$, and $\sigma_{a_ia_j}$. The remaining parameters $\sigma_{a_i}^2$ and $\sigma_{a_i\theta_i}$ are determined according to \eqref{eq:obedience-var-sym}.\footnote{The proof also shows that conditions \eqref{eq:obedience-mean-sym} and~\eqref{eq:obedience-var-sym} are still \emph{necessary} for obedience under general (not necessarily Gaussian) mechanisms. For such mechanisms, the conditional distribution $a\given \theta,\omega$ is not entirely determined by $\mu$ and $\K$. Hence, obedience imposes further constraints on the higher order moments of the conditional action distribution.}

Finally, as obedience constraints pin down the mean action for each player, equation \eqref{eq:obedience-utility-cond} implies that a player's expected payoff is measured by the \emph{variance} of their actions.

\Cref{prop:obedience-char-sym} assumed that $\K$ is a proper covariance matrix, that is, it must be positive semi-definite, which translates into two inequality constraints on the covariance parameters (see \cref{lemma:psd} in  \cref{sec:symmetry}). We can however obtain a joint characterization of obedience and positive semi-definiteness by using the alternative parametrization provided in \eqref{eq:params-bis}. Using \eqref{eq:params-dict}, we can write the action recommendation of a symmetric mechanism as
\begin{equation}\label{eq:params-sym}
    a_i = \mu_{a_i}
    + \frac{\sigma_{a_i\omega}}{\sigma_{\omega}^2}(\omega-\mu_\omega)
    + \frac{\sigma_{a_i\theta_i}}{\sigma_{\theta_i}^2}(\theta_i-\mu_{\theta_i})
    + \frac{\sigma_{a_i\theta_j}}{\sigma_{\theta_i}^2}\sum_{j\neq i}(\theta_j-\mu_{\theta_j})
    + \delta\eps_i,
\end{equation}
for some $\delta\in\R$ and where the variables $(\eps_i)_{i\in[n]}$ are
$\N(0,1)$ variables with $\cov(\eps_i,\eps_j) = \rho$ for some $\rho\in\R$.
We can then  describe the set of all symmetric and obedient mechanisms in terms of the three parameters $(\sigma_{a_i\omega},\sigma_{a_i\theta_j},\rho)$, for which the constraints take a  simple form.

\begin{proposition}\label{cor:obedience-char-bis}
    For each $r\in(-1,\frac 1 {n-1})$, the action recommendations of symmetric obedient mechanisms are of the form \eqref{eq:params-sym} and are exhaustively described by parameters $\sigma_{a_i\omega}$,
    $\sigma_{a_i\theta_j}$, and $\rho\eqdef\cov(\eps_i,\eps_j)$ subject to:
    \begin{enumerate}
	\item  the covariance $\rho$ satisfies $-\frac 1{n-1}\leq
	    \rho\leq 1$.
	\item $\sigma_{a_i\theta_j}/rt\sigma_{\theta_i}^2$ and
	    $\sigma_{a_i\omega}/s\sigma_{\omega}^2$ satisfy
	    the ellipse constraint
	    \begin{equation}\label{eq:ellipse}\tag{$\mathcal{E}$}
		t^2\sigma_{\theta_i}^2\frac{r^2y_0}{fx_0}\big(x-x_0\big)^2
		+s^2\sigma_{\omega}^2\big(y - y_0\big)^2
		\leq t^2\sigma_{\theta_i}^2\frac{r^2y_0}{f}x_0
		+s^2\sigma_\omega^2y_0^2
	    \end{equation}
	    in the variables $(x,y)\in\R^2$, with $f\eqdef\frac 1 {n-1}$ and center
	    \begin{displaymath}
		x_0 = \frac{f}{2(f-r)(1+r)}
		\quad\mathrm{and}\quad
		y_0 = \frac f {2(f-r)}.
	    \end{displaymath}
    \end{enumerate}
    Moreover, $a_i$ is deterministic conditioned on $(\theta,\omega)$, that is, $\delta=0$ iff the constraint \eqref{eq:ellipse} is binding. Denoting by $\xi$ the (positive) slack of the constraint, we have
    \begin{equation}\label{eq:var-from-slack}
	\delta^2 = \xi\frac{f-r}{f-\rho r}.
    \end{equation}
\end{proposition}

By \cref{cor:obedience-char-bis}, a symmetric obedient mechanism is fully determined by the choice of a point inside the ellipse \eqref{eq:ellipse}, and of the covariance $\rho\in\big[-\frac 1{n-1},1\big]$. Given  these choices, the remaining parameters in \eqref{eq:params-sym} are obtained as follows: the variance of the noise $\delta$ is given by \eqref{eq:var-from-slack}, the mean action $\mu_{a_i}$ by \eqref{eq:obedience-mean-sym} and the covariance $\sigma_{a_i\theta_i}$ by the second equality in \eqref{eq:obedience-var-sym}.

\begin{figure}[t]
\begin{center}
    \includegraphics[scale=1.1]{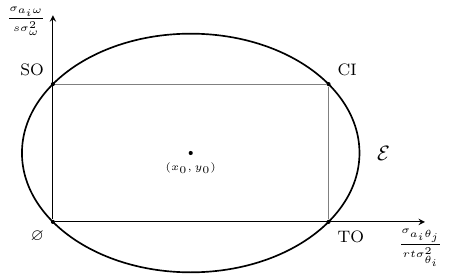}
\end{center}
\caption{Symmetric mechanisms satisfying obedience are parametrized by points inside the ellipse \ref{eq:ellipse} and by the noise correlation $\rho\in\big[-\frac 1{n-1},1\big]$. The mechanism is deterministic conditioned on $(\theta,\omega)$ iff it lies on \ref{eq:ellipse}. The ellipse circumscribes a rectangle formed by four deterministic mechanisms corresponding to “extremal” information structures (see the text).}
\label{fig:obedience}
\end{figure}

\looseness=-1
\Cref{fig:obedience} illustrates the  ellipse \eqref{eq:ellipse} in the coordinate system $(\sigma_{a_i\theta_j}/rt\sigma_{\theta_i}^2,\sigma_{a_i\omega}/s\sigma_{\omega}^2)$ where moving to the right (resp.\ up) corresponds to increasing the weight on the competitors' types (resp.\ state) in the action recommendation.\footnote{The reader is invited to explore the interactive visualization available at \url{http://thibaut.horel.org/info-coord-viz/}, which shows how varying the environment parameters changes the constraints and optimal mechanisms depicted in this figure and other figures in the paper.} Importantly, the origin $\varnothing=(0,0)$ lies on the ellipse. Since $\sigma_{a_i\theta_j}=0$ implies $\sigma_{a_i\theta_i}=t\sigma_{\theta_i}^2$, the action recommendation at $\varnothing$ is
\begin{equation}\label{eq:bne}
    a_i = \mu_{a_i} + t(\theta_i-\mu_{\theta_i}),
\end{equation}
and induces the Bayes–Nash equilibrium of the interim game in which each player observes only their own type. In other words, sending this recommendation is equivalent to revealing no information to the players. Then, starting from $\varnothing$ and moving parallel to the coordinate axes, we encounter three deterministic obedient mechanisms that form a rectangle inscribed in \eqref{eq:ellipse}:
\begin{itemize}
    \item the point $\text{SO}=(0,2y_0)$, \emph{State Only}, for which
    the players (who know their own type) learn $\omega$ exactly upon receiving their recommendation. Hence it induces the Bayes–Nash equilibrium of the game in which the players know their own type as well as the state.
    \item the point $\text{TO}=(2x_0,0)$, \emph{Types Only}, 
    which induces the \emph{information sharing} equilibrium, where players know each other's private types but have no information about $\omega$.
     \item the point $\text{CI}=(2x_0,2y_0)$, \emph{Complete Information}, 
	which is outcome-equivalent to the Nash equilibrium of the complete information game seen in \eqref{eq:ne}.

\end{itemize}

The segment $[\varnothing,\mathrm{SO}]$ characterized by $\sigma_{a_i\theta_j}=0$ is the set of obedient mechanisms whose recommendation to player $i$ conditions only on $\theta_i$ and $\omega$, that is, individual-level algorithms that do not reveal information about competitors' submissions. By \cref{prop:ic-char}, this segment coincides with the set of mechanisms implementable without discriminatory transfers.

Because the rectangle \(\{\varnothing,\mathrm{SO},\mathrm{TO},\mathrm{CI}\}\) is inscribed in the ellipse \eqref{eq:ellipse}, Proposition \ref{cor:obedience-char-bis} establishes that  it is possible to \emph{obediently} induce stronger covariance between a player's action and their competitors' types than under complete information (i.e., \(x>2x_0\)). This, however, requires reducing the covariance with the state (i.e., choose \(y<2y_0\)). Likewise, \emph{negative} covariance with the state (\(y<0\), e.g., recommending prices that fall in high-demand states) can be obedient, but only to the extent that the joint choice \((x,y)\) remains inside the ellipse---intuitively, one can tilt recommendations against the state \emph{unless} they also tilt ``the wrong way'' with competitors’ types so much that \((x,y)\) exits the ellipse.

\subsection{Incentive Compatibility}\label{sec:ic}

Incentive compatibility (\cref{def:ic}) requires preventing both first-stage deviations (misreports of the type) and double deviations, in which a player misreports their type and then deviates from the recommended action.

Because an incentive compatible mechanism must in particular incentivize obedience, we fix an obedient mechanism and wish to characterize when it additionally rules out double deviations.
We begin by computing the optimal second-stage action after a misreport. Suppose player $i$ has true type $\theta_i$, reports $\theta_i'$ in the first stage, and observes recommendation $a_i$ in the second. The first-order condition for the second-stage action is
\begin{displaymath}
a_i' = s\,\E[\omega\given a_i,\theta_i,\theta_i'] + r\sum_{j\neq i}\E[a_j\given a_i,\theta_i,\theta_i'] + t\theta_i.
\end{displaymath}
Conditional on $(\theta',\omega)$, the recommendations' distribution depends on $\theta_i$ only through the report $\theta_i'$, hence $\omega$ and $a_{-i}$ are independent of $\theta_i$ given $(a_i,\theta_i')$. The conditional expectations therefore simplify to $\E[\omega\given a_i,\theta_i']$ and $\E[a_j\given a_i,\theta_i']$. Moreover, obedience at the reported type $\theta_i'$ gives, by \eqref{eq:obedience},
\begin{displaymath}
    s\E[\omega\given a_i,\theta_i'] + r\sum_{j\neq i}\E[a_j\given a_i,\theta_i'] = a_i - t\theta_i'.
\end{displaymath}
    Substituting this into the first-order condition yields a simple formula for the optimal deviation,
\begin{equation}\label{eq:deviation}
a_i' = a_i + t(\theta_i - \theta_i').
\end{equation}
Note that the effect of a misreport (from $\theta_i$ to $\theta_i'$) on the distribution of opponents' actions is fully reflected in the recommended action $a_i$, and does not affect the difference between the player's optimal and recommended actions $a_i' - a_i$.

The interim utility of player $i$ with true type $\theta_i$ who reports $\theta_i'$ and  best responds at the action stage is then given by
\begin{equation}\label{eq:tu-misreport}
\tilde u_i(\theta_i';\theta_i) = \frac{1}{2}\E[a_i\given\theta_i']^2 + t(\theta_i - \theta_i')\E[a_i\given\theta_i'] + \frac{t^2}{2}(\theta_i - \theta_i')^2 + \frac{1}{2}\var(a_i\given\theta_i').
\end{equation}
Setting $\theta_i'=\theta_i$ in \eqref{eq:tu-misreport} yields the on-path interim utility of an obedient mechanism,
\begin{equation}\label{eq:interim-utility-ob}
\tilde u_i(\theta_i) = \frac{1}{2}\E[a_i\given\theta_i]^2 + \frac{1}{2}\var(a_i\given\theta_i).
\end{equation}

\begin{proposition}\label{prop:ic-char}
    The symmetric mechanism $(\mu,\K,p)$ is incentive compatible if and only if it is obedient and for each player $i\in[n]$:
    \begin{enumerate}
	\item The derivative of the payment function is given by
	    \begin{displaymath}
		p_i'(\theta_i)
		=\prn[\bigg]{\frac{\sigma_{a_i\theta_i}}{\sigma_{\theta_i}^2} -t}
		\E[a_i\given \theta_i]
		= \prn[\bigg]{\frac{\sigma_{a_i\theta_i}}{\sigma_{\theta_i}^2} -t}
		\prn[\Big]{\mu_{a_i} + \frac{\sigma_{a_i\theta_i}}{\sigma_{\theta_i}^2}
		(\theta_i-\mu_{\theta_i})}.
	    \end{displaymath}
	\item The covariance satisfies $t\sigma_{a_i\theta_i}\geq t^2\sigma_{\theta_i}^2$; equivalently, by obedience \eqref{eq:obedience-var-sym}, $rt\sigma_{a_j\theta_i}\geq 0$.\footnote{The second-order condition $t\sigma_{a_i\theta_i}\geq t^2\sigma_{\theta_i}^2$ is strictly stronger than the monotonicity $t\sigma_{a_i\theta_i}\geq 0$ that would suffice if the player were committed to following the recommendation at the action stage. The additional $t^2\sigma_{\theta_i}^2$ comes from the quadratic correction $(t^2/2)(\theta_i-\theta_i')^2$ in \eqref{eq:tu-misreport}, which is generated by the optimal action-stage deviation $a_i'=a_i+t(\theta_i-\theta_i')$. Without this correction, only the trivial allocation monotonicity would be required. In contrast, our earlier work \citep{BDHN22} showed that, for multiplicatively decomposable utilities of the form $u_i(a;\theta_i,\omega) = \theta_i\cdot\pi(a;\omega)$, incentive compatibility is equivalent to requiring truthfulness and obedience separately.}
    \end{enumerate}
\end{proposition}

The proof, given in \cref{sec:app-tic}, shows that for an arbitrary mechanism (not necessarily Gaussian), the requirement that truth-telling $\theta_i'=\theta_i$ maximize $\tilde u_i(\theta_i';\theta_i)-p_i(\theta_i')$ is equivalent to requiring the rent $\tilde u_i(\theta_i)-p_i(\theta_i)$ in \eqref{eq:interim-utility-ob} to be $t^2$-strongly convex with subderivative $t\E[a_i\given\theta_i]$. For a Gaussian mechanism, $t\E[a_i\given\theta_i]$ is affine in $\theta_i$ with slope $t\sigma_{a_i\theta_i}/\sigma_{\theta_i}^2$, which allows us to write the first-order condition $p_i'=\tu_i'-t\E[a_i\given\theta_i]$ as in part~1. The $t^2$-strong convexity condition is equivalent to $t^2$-strong monotonicity of $t\E[a_i\given\theta_i]$. This reduces to $t\sigma_{a_i\theta_i}\geq t^2\sigma_{\theta_i}^2$ in the Gaussian case.

For example, in the Bertrand setting  of \cref{example:bertrand} (where $t>0$ and $r>0$), a firm that understates its own demand-shifter ($\theta_i'<\theta_i$) receives a lower recommended price and then deviates upward to capture its true demand. Incentive compatibility must discourage such double deviations; \cref{prop:ic-char} shows that this requires others' actions to move in the right direction with each player's report, $rt\sigma_{a_j\theta_i}\geq 0$, which under Bertrand parameters means that understating one's demand also leads to lower prices by competitors and, combined with the firm's own lower recommendation, makes the upward deviation unprofitable. In Cournot competition the signs of $t$ and $r$ both flip, so $rt>0$ is preserved; the directions of the deviation and of the disciplining response reverse, but the same anti-misreport logic applies.

\begin{figure}[t]
\begin{center}
    \includegraphics[scale=1.1]{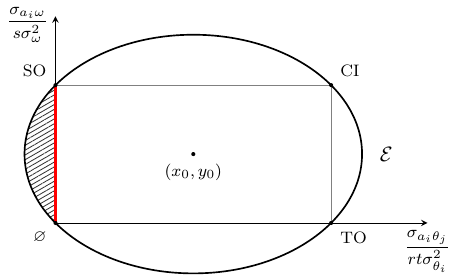}
\end{center}
    \caption{Symmetric and incentive compatible mechanisms must satisfy $rt\sigma_{a_i\theta_j}\geq 0$. The hatched area contains the obedient mechanisms that are not incentive compatible: these admit profitable double deviations for the players. At the boundary $\sigma_{a_i\theta_j}=0$, the red line segment $[\varnothing,\text{SO}]$ contains the obedient mechanisms that are implementable without transfers.\protect\footnotemark}
    \label{fig:ic}
\end{figure}
\footnotetext{See also the interactive visualization at \url{http://thibaut.horel.org/info-coord-viz}.}

Finally, \cref{prop:ic-char} shows that $p_i'\equiv 0$ if and only if $\sigma_{a_i\theta_i}=t\sigma_{\theta_i}^2$. The obedient mechanisms for which $\sigma_{a_i\theta_i}=t\sigma_{\theta_i}^2$ are precisely the mechanisms implementable without discriminatory transfers. By the second covariance constraint required for obedience \eqref{eq:obedience-var-sym}, this condition is equivalent to $\sigma_{a_j\theta_i}=0$, which means agent $i$'s report does not influence agent $j$'s actions. Thus, the mechanisms implementable without transfers are those that provide information about the common state only.\footnote{This result is consistent with Theorem~1 in \cite{komzl17} on the equivalence between private and public persuasion. Relative to their work, our setting introduces competing receivers but specializes to quadratic preferences.} 

\Cref{fig:ic} illustrates how incentive compatibility refines the set of obedient mechanisms.\footnote{In \cref{sec:participation}, we also show that it is possible to construct off-path action recommendations that implement the same set of mechanisms, with and without transfers, when the designer cannot control the players' access to the downstream game.} Intuitively, a mechanism that only sends signals about $\omega$ does not require transfers: information about $\omega$ does not change the agent's expected action, so an agent's preferences over information about the state are independent of their type $\theta_i$. It is more surprising that the designer must resort to transfers to implement any distribution that correlates one player's actions to another player's report. When $\sigma_{a_j\theta_i}\neq 0$, agent $i$'s report influences agent $j$'s expected action, and hence agent $i$'s own expected action. Because different types $\theta_i$ bear heterogeneous costs of taking a higher or lower average action, the agent's preferences over type reports depend on the agent's true type. It then follows that the designer must use discriminatory payments to align the different types' incentives.

\section{Optimal Obedient Mechanisms}\label{sec:opt}

We begin this section by deriving the structural properties of mechanisms that maximize an arbitrary linear objective. We then move to the firms-, revenue-, and consumer-optimal mechanisms. For each of  these problems, our approach consists of characterizing the optimal \emph{obedient} mechanism, constructing the associated transfers (\cref{prop:ic-char}), and then verifying that the mechanism satisfies the incentive compatibility constraint $r t\sigma_{a_j\theta_i}\geq 0$.


\subsection{Structural Properties}\label{sec:structure}

We have so far restricted attention to symmetric mechanisms. As it turns out,  such restriction is without loss of generality when maximizing a symmetric and concave objective function over the class of Gaussian and obedient mechanisms. \Cref{sec:symmetry-supp} shows that this follows from a symmetrization argument, which uses the fact that the obedience constraints on the moments of a Gaussian obedient mechanism are stable under a relabeling of the players.

Because all the objective functions we consider below are symmetric, this result allows us to focus on optimization problems of the form
\begin{displaymath}
    \sup_{\K\in\cO_s} F(\K),
\end{displaymath}
where $F$ is concave and $\cO_s$ is the set of covariance matrices of symmetric and obedient mechanisms characterized by the two linear covariance constraints (\cref{prop:obedience-char-sym}) and positive semi-definiteness (see \cref{lemma:psd} in  \cref{sec:symmetry}). Note that $\K$ is the only optimization variable since the mean $\mu$ of the action recommendations is determined by obedience.

\looseness=-1
As discussed below \cref{prop:obedience-char-sym}, the linear covariance constraints reduce the degrees of freedom to three variables, e.g.\ $(\sigma_{a_ia_j},\sigma_{a_i\theta_j},\sigma_{a_i\omega})$. Let $\tilde F$ be the “reduced” objective function obtained by substituting $\sigma_{a_i}^2$ and $\sigma_{a_i\theta_i}$ in the definition of $F$ using the obedience constraints \eqref{eq:obedience-var-sym}. Since the substitutions are linear, $\tilde F$ remains concave and the optimization problem becomes
\begin{equation}\label{eq:reduced}
    \begin{aligned}
	\max&\;\; \tilde F(\sigma_{a_ia_j}, \sigma_{a_i\omega},\sigma_{a_i\theta_j})\\
	\text{s.t.}&\;\;
	\frac{1}{\sigma_{\theta_i}^2}\bigg[\frac{f-r}{f}\sigma_{a_i\theta_j}
	-t\sigma_{\theta_i}^2\bigg]^2
	\leq s\sigma_{a_i\omega} + \frac{rt}{f}\sigma_{a_i\theta_j}
	-\frac{f-r}{f}\sigma_{a_ia_j}+t^2\sigma_{\theta_i}^2,\\
	&\;\;\frac{1}{\sigma_{\theta_i}^2}\bigg[\frac{1+r}{f}\sigma_{a_i\theta_j}
	+t\sigma_{\theta_i}^2\bigg]^2
	+\frac{n}{\sigma_\omega^2}\sigma_{a_i\omega}^2
	\leq s\sigma_{a_i\omega} + \frac{rt}{f}\sigma_{a_i\theta_j} +
	\frac{1+r}{f}\sigma_{a_ia_j}
	+t^2\sigma_{\theta_i}^2.
    \end{aligned}
\end{equation}
The two inequality constraints come from the positive semi-definiteness of the covariance matrix $\K$ (cf.\ \cref{lemma:psd}), after substituting $\sigma_{a_i}^2$ and $\sigma_{a_i\theta_i}$ using \eqref{eq:obedience-var-sym}.

For the firms- and consumer-optimal mechanisms considered below, the reduced objective $\tilde F$ is linear. In this case, the following proposition, holding for arbitrary linear objectives, describes the optimal mechanism.

\begin{proposition}\label{prop:linear}
    For $(\alpha,\beta,\gamma)\in\R^3$, consider problem \eqref{eq:reduced} with
    \begin{displaymath}
	\tilde F(\sigma_{a_ia_j}, \sigma_{a_i\omega},\sigma_{a_i\theta_j}) = (n-1)\alpha r \sigma_{a_ia_j} + \beta s\sigma_{a_i\omega} + (n-1)\gamma rt\sigma_{a_i\theta_j}.
    \end{displaymath}
    For $r\in\big(-1,\frac 1 {n-1}\big)$, there exists a unique optimal mechanism. Its action recommendations take the form \eqref{eq:params-sym} with noise correlation $\rho=1$  when $\alpha r>0$ and $\rho=-\frac 1{n-1}$ when $\alpha r<0$. The optimal values for $\sigma_{a_i\theta_j}$ and $\sigma_{a_i\omega}$ are given by 
	$\sigma_{a_i\theta_j} = x(\lambda^\star)rt\sigma_{\theta_i}^2$ and
	$\sigma_{a_i\omega} = y(\lambda^\star)s\sigma_\omega^2$ with
    \begin{displaymath}
	x(\lambda)\eqdef
	\frac{f}{2(1+r)}
	\frac{\lambda nf + \alpha(r+2)+\gamma(r+1)}
	{\big[\lambda nf(f-r) -\alpha r(r+1)\big]},
	\quad y(\lambda)\eqdef
	\frac{1}{2n}\frac{\lambda nf-\alpha r + \beta(r+1)}
	{\lambda(f-r)-\alpha r},
    \end{displaymath}
    $f\eqdef \frac 1{n-1}$ and $\lambda^\star\eqdef \min\set[\big]{\lambda\geq \max\set[\big]{0, \frac {\alpha r} {f-r}} \;\big|\; \big(x(\lambda), y(\lambda)\big) \text{ satisfies \eqref{eq:ellipse}}}$. The remaining mechanism parameters $(\mu_{a_i}, \sigma_{a_i\theta_i}, \delta)$ are given by \eqref{eq:obedience-mean-sym}, \eqref{eq:obedience-var-sym} and \eqref{eq:var-from-slack} respectively.
\end{proposition}

As seen in \cref{prop:linear}, proved in \cref{sec:structure-app}, the optimal action recommendations are always maximally correlated or anticorrelated conditioned on $(\theta,\omega)$ depending on the sign of $\alpha r$.\footnote{This fact holds more generally for an arbitrary concave objective $\tilde F$ that is strictly monotone in $\sigma_{a_ia_j}$ (see \cref{prop:structural-main} in \Cref{sec:opt-app}).} As we will see, when maximizing the firms' profits (\cref{sec:welfare}), we have $\alpha=1$, and the correlation is positive or negative according to whether actions are strategic complements or substitutes. In contrast, when maximizing \emph{consumer} welfare (\cref{sec:consumer}), we have $\alpha=-1$ and the optimal correlation structure is reversed.

The following geometric interpretation provides additional intuition about the structure of the optimal mechanism.
Since the noise correlation $\rho$ is determined by the sign of $\alpha r$, there are only two remaining degrees of freedom: $\sigma_{a_i\theta_j}$ and $\sigma_{a_i\omega}$. By stationarity, these two variables are constrained to lie on the curve defined by $\big(x(\lambda),y(\lambda)\big)$, namely, a portion of hyperbola. This curve is parametrized by the Lagrange multiplier $\lambda$ associated with the ellipse constraint \eqref{eq:ellipse}, constrained to $\lambda\geq \lambda_{\rm min} \eqdef\max\set{0, \alpha r/(f-r)}$ by dual feasibility. Moreover, it is easy to verify that the functions $x$ and $y$ are monotone and converge respectively to $x_0$ and $y_0$, the coordinates of \ref{eq:ellipse}'s center.
Complementary slackness implies that either $\lambda=0$ or \eqref{eq:ellipse} is binding (in particular, \eqref{eq:ellipse} always binds when $\lambda_{\rm min}>0$). We thus have two mutually exclusive cases:
\begin{enumerate}
    \item The point $\big(x(\lambda_{\rm min}),y(\lambda_{\rm min})\big)$  satisfies \eqref{eq:ellipse} strictly, in which case we have an interior solution—that is, randomized action recommendations—determined by $\lambda^\star=\lambda_{\rm min}$.
    \item The point $\big(x(\lambda_{\rm min}), y(\lambda_{\rm min})\big)$ violates \eqref{eq:ellipse}. In this case, complementary slackness requires the ellipse constraint to bind at the optimum. Thus, the optimal action recommendation is obtained by moving along the hyperbola (increasing the value of $\lambda$ from $\lambda_{\rm min}$), until we find the (unique) intersection of the hyperbola with the ellipse's boundary. In this case, the mechanism is deterministic conditioned on $(\theta,\omega)$.
\end{enumerate}
These two cases are illustrated in \Cref{fig:welfare} below for the firms-optimal mechanism.

\subsection{Firms-Optimal Mechanism}\label{sec:welfare}


For $n>1$, define \emph{producer surplus} as $W_F\eqdef\sum_{i\in[n]}\E[u_i(a;\theta_i,\omega)]$ and \emph{consumer surplus} as $W_C\eqdef$ the expected utility of the representative consumer under the Cournot or Bertrand interpretation (derived in \cref{sec:consumer} and \cref{app-consumer}).

In \eqref{eq:obedience-utility-cond}, we wrote  the expected utility of player $i$ in any obedient mechanism as
\begin{equation}\label{eq:expected-utility}
    \E[u_i(a;\theta_i,\omega)] = \frac 1 2 \E[a_i^2]
    = \frac {\mu_{a_i}^2+\sigma_{a_i}^2} 2,
\end{equation}
\looseness=-1
where the first equality uses the law of total expectation. Since the value of $\mu_{a_i}$ is pinned down by obedience \eqref{eq:obedience-mean-sym}, maximizing producer surplus over the set of obedient mechanisms is  equivalent to maximizing $\sum_{i=1}^n\sigma_{a_i}^2$. In particular, this objective function is invariant under permutation of the players, hence 
we can restrict to symmetric mechanisms without loss of generality. Thus, after substituting $\sigma_{a_i}^2$ and $\sigma_{a_i\theta_i}$ using the obedience constraints \eqref{eq:obedience-var-sym}, maximizing surplus over the class of symmetric obedient mechanisms takes the form \eqref{eq:reduced} with
\begin{displaymath}
	\tilde F(\sigma_{a_ia_j}, \sigma_{a_i\omega},\sigma_{a_i\theta_j})
	= (n-1)r\sigma_{a_ia_j}+s\sigma_{a_i\omega}+(n-1)rt\sigma_{a_i\theta_j}.
\end{displaymath}
The reduced objective function is linear and strictly monotone in $\sigma_{a_ia_j}$ and we can instantiate \cref{prop:linear} with $\alpha=\beta=\gamma=1$. The resulting mechanism is described in the following proposition; the restriction to Gaussian mechanisms is without loss here, with scope discussed in \cref{rem:gaussian-scope}.

\begin{proposition}\label{prop:opt-welfare}
    For $r\in\big(-1,\frac 1 {n-1}\big)$, there exists a unique symmetric mechanism maximizing producer surplus subject to obedience. In this mechanism, the action recommendations take the form \eqref{eq:params-sym} with noise correlation $\rho=1$  when $r>0$ and $\rho=-\frac 1{n-1}$ when $r<0$. The optimal values for $\sigma_{a_i\theta_j}$ and $\sigma_{a_i\omega}$ have the form
	$\sigma_{a_i\theta_j} = x(\lambda^\star)rt\sigma_{\theta_i}^2$ and 
	$\sigma_{a_i\omega} = y(\lambda^\star)s\sigma_\omega^2$ with
    \begin{equation}\label{eq:stationarity-welfare}
	x(\lambda)= \frac{f}{2(1+r)}\cdot \frac{\lambda nf + 2r+3}{\lambda nf(f-r) -(r+1)r},
	\quad y(\lambda) = \frac{1}{2n}\cdot\frac{\lambda nf+1}{\lambda(f-r)-r},
    \end{equation}
    where $f\eqdef \frac 1 {n-1}$ and $\lambda^\star = \min\set[\big]{\lambda\geq \max\set[\big]{0, \frac r {f-r}} \;\big|\; \big(x(\lambda), y(\lambda)\big) \text{ satisfies \eqref{eq:ellipse}}}$. The remaining  parameters $(\mu_{a_i}, \sigma_{a_i\theta_i}, \delta)$ are given by \eqref{eq:obedience-mean-sym}, \eqref{eq:obedience-var-sym} and \eqref{eq:var-from-slack} respectively.
\end{proposition}

\begin{figure}[!t]
    \begin{center}
	\hfill{}
	\includegraphics[scale=.9]{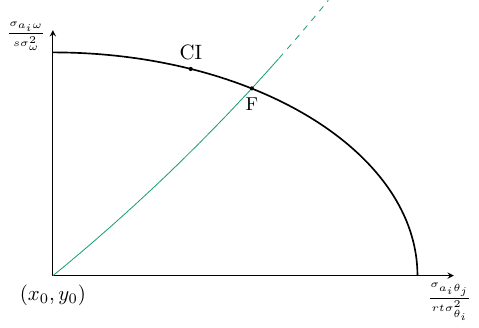}
	\hfill{}
	\includegraphics[scale=.9]{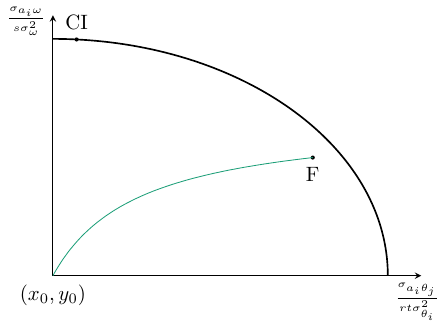}
	\hfill{}
    \end{center}
    \vspace{-1em}
    \caption{Obedience ellipse \ref{eq:ellipse} and firms-optimal mechanisms (F) in the coordinate system $x=\sigma_{a_i\theta_j}/rt\sigma_{\theta_i}^2$, $y=\sigma_{a_i\omega}/s\sigma_{\omega}^2$. We focus on the quadrant $x\geq x_0$, $y\geq y_0$ containing \ref{eq:ellipse}'s top-right quarter as well as the stationarity curve $\big(x(\lambda),y(\lambda)\big)$ parametrized by $\lambda\geq \lambda_{\rm min}$ (in green). On the left, the curve originates outside \ref{eq:ellipse}, and the firms-optimal mechanism is attained at the intersection with \ref{eq:ellipse}. On the right, the curve originates inside \ref{eq:ellipse} and the optimal mechanism is obtained for $\lambda=\lambda_{\rm min}$.}
    \label{fig:welfare}
\end{figure}

The geometric interpretation of this mechanism is the same as the one described below \cref{prop:linear}. In particular, we have an interior or boundary solution depending on whether the stationarity curve $\big(x(\lambda),y(\lambda)\big)$ originates inside or outside \ref{eq:ellipse}. \Cref{fig:welfare} illustrates these two situations\footnote{See also the interactive visualization at \url{http://thibaut.horel.org/info-coord-viz}.} and \cref{prop:welfare-random} below gives a complete characterization of the condition under which each case arises.

The functions $x(\lambda)$ and $y(\lambda)$ in \cref{prop:opt-welfare} are easily seen to be decreasing and since they converge to \ref{eq:ellipse}'s center, the stationarity curve $\big(x(\lambda),y(\lambda)\big)$ lies entirely in the quadrant $\set{(x,y)\in\R^2\given x\geq x_0, y\geq y_0}$ (cf.~\cref{fig:welfare}). This means the covariance condition of \cref{prop:ic-char} (incentive-compatibility) 
is satisfied, and the mechanism is implementable as a result. We collect these results in the following corollary.

\begin{corollary}\label{cor:welfare-ic}
    The mechanism of \cref{prop:opt-welfare} satisfies
    \begin{displaymath}
	\frac{\sigma_{a_i\omega}}{s\sigma_{\omega}^2}\geq y_0\geq 0,
	\quad
	\frac{\sigma_{a_i\theta_j}}{rt\sigma_{\theta_i}^2}\geq x_0\geq 0,
	\quad\text{and}\quad
	\frac{\sigma_{a_i\theta_i}}{t\sigma_{\theta_i}^2}\geq 1+\frac{r^2}f x_0\geq 1.
    \end{displaymath}
\end{corollary}

\Cref{cor:welfare-ic} establishes that the IC monotonicity is slack at the firms-optimal point:  optimization  subject to obedience only delivers a mechanism that is also incentive compatible. The intuition is the alignment between the designer's and the firms' objectives: with access to discriminatory transfers, the designer elicits the type while leaving the action distribution undistorted. 
More broadly, \cref{prop:ic-slack-regime} in \cref{sec:opt-app} establishes a closed-form necessary and sufficient condition on the linear-objective parameters $(\alpha,\beta,\gamma)$ under which the obedient optimum satisfies the IC monotonicity with slack; the firms-optimal case $(\alpha,\beta,\gamma)=(1,1,1)$ is one instance  in this class.

However, the slackness of the IC monotonicity is not a universal property of our framework. It fails for games of strategic complements $(r>0)$ when the designer wishes to anticorrelate the firms' actions: in such cases the obedient-optimal point violates the IC monotonicity of \cref{prop:ic-char}, and the optimal mechanism collapses to the no-transfer segment $[\varnothing,\mathrm{SO}]$ of \cref{fig:ic} where $\sigma_{a_j\theta_i}=0$. Two concrete instances of this phenomenon are the consumer-optimal mechanism in the Bertrand case (\cref{sec:consumer}) and in the Cournot case with complement goods ($r>0$) in \cref{app-consumer}.

Not only does the firms-optimal mechanism  in \cref{fig:welfare} lie in the positive quadrant: observe that it also places more weight on private types and less weight on the common state than the complete-information outcome. This is especially noteworthy in games of strategic complements, where one might expect heavier reliance on the common state to facilitate coordination. We now offer some intuition for this result, which we  state and prove formally in \cref{sec:comparison} below, where we also provide a comparison with the revenue-optimal mechanism (\cref{prop:comparison}).

Consider Bertrand competition with private demand shifters (\cref{example:bertrand}). If demand-side information is not shared, one firm's high own demand has no direct effect on its rivals' pricing. But if the platform observes and discloses these signals, a firm learning that its competitor's demand is high anticipates higher prices from that competitor and raises its own price. With strategic complements, this reaction feeds back: both firms' prices move together. Selective disclosure therefore amplifies correlation in strategies.

Full revelation of the common state would also correlate actions, but residual heterogeneity in private demand prevents fully optimal co-movement. Moreover, complete-information actions underweight the social value of coordination. The firms-optimal mechanism thus reduces the precision of actions' responses to the state in order to strengthen coordination, ensuring that each firm's action over-reacts not only to its own demand but also to its competitors' demand.  

Finally, as discussed below \cref{prop:opt-welfare} and seen in \cref{fig:welfare}, the firms-optimal recommendations are deterministic or randomized (conditioned on $\theta,\omega$) according to whether the stationarity curve $\big(x(\lambda),y(\lambda)\big)$ originates outside or inside \ref{eq:ellipse}. This can be characterized in terms of the game parameters $n,r,s,t$ and of the prior variances. More specifically, observe that the stationarity curve \eqref{eq:stationarity-welfare} depends only on $r$ and $n$, while \ref{eq:ellipse}'s equation depends on $n$, $r$ and the ratio $t^2\sigma_{\theta_i}^2/s^2\sigma_{\omega}^2$. With a careful analysis, we then obtain the following result.

\begin{proposition}\label{prop:welfare-random}
Consider the optimal mechanism in \cref{prop:opt-welfare}.
    \begin{enumerate}
    \item For $-\frac 1{n+1}\leq r < \frac 1{n-1}$ the action recommendations are always deterministic conditioned on $\theta,\omega$ (boundary solution).
    \item For $-1< r < -\frac 1 {n+1}$, there exists a threshold $\tau_{\rm F}(n,r)$ such that the action recommendations are randomized conditioned on $\theta,\omega$ (interior solution) iff $t^2\sigma_{\theta_i}^2/s^2\sigma_{\omega}^2< \tau_{\rm F}(n,r)$.
    \end{enumerate}
\end{proposition}


This can be intuitively understood as follows.  In a game of strategic substitutes $r<0$, the mechanism designer wishes to maximally anticorrelate the players' actions, even if obedience requires making actions less responsive to the common state. This is optimally achieved by placing more weight on the players' private types, relative to the complete-information Nash equilibrium (cf.\ \cref{prop:comparison} below). 
However, when $\sigma_{\theta_i}^2$ is small (relative to $\sigma_{\omega}^2$), the variation in types is not sufficient to generate sufficiently strong anti-correlation in actions.\footnote{The proof of this proposition (\cref{sec:welfare-app}) gives an expression for the threshold function $\tau_{\rm F}$.} The mechanism therefore supplements this force with additional negatively correlated noise, again sacrificing precision to improve the correlation structure.

In contrast, when the firms-optimal mechanism issues deterministic action recommendations, Proposition \ref{prop:opt-welfare} shows that it induces actions that are linear in each player's type, in the state, and in the other players' average type. Therefore, a signal structure that fully reveals the corresponding linear combination of state and competitors' types implements the desired Bayes Correlated Equilibrium action distribution as the unique Bayes Nash Equilibrium \citep{BM13}.

\subsection{Revenue-Optimal Mechanism}\label{sec:revenue}

 The revenue-optimal mechanism maximizes the firms' expected welfare net of their information rents. We therefore begin this section by deriving an expression for the players' (interim) information rent by deriving the largest payments that implement an incentive-compatible mechanism. Thus, we begin with the on-path payment formula from \cref{prop:ic-char}, and we choose the largest constant term that guarantees non-negative utility to all types. 
 

\begin{corollary}\label{cor:opt-pay-ob}
   In any incentive-compatible mechanism, player $i$'s  minimal information rent  takes the form
    \begin{equation}\label{eq:payment-ob}
	U_i(\theta_i)\equiv   \tu_i(\theta_i)-p_i(\theta_i) =
	t\int_{\theta_i^*}^{\theta_i} \E[a_i\given\theta_i=s]ds=
	\frac{t\sigma_{\theta_i}^2}{2\sigma_{a_i\theta_i}}
	\Big(\mu_{a_i}+\frac{\sigma_{a_i\theta_i}}{\sigma_{\theta_i}^2}(\theta_i-\mu_{\theta_i})\Big)^2,
    \end{equation}
   where $\mu_{a_i}$ is player $i$'s  average action and $\theta_i^*$ is the unique solution to $\E[a_i\given\theta_i^*]=0.$
\end{corollary}


This expression  lends itself to a natural interpretation. The mechanism needs to compensate an agent for being responsive to their type (as measured by the covariance $\sigma_{a_i\theta_i}$) in order to prevent profitable misreporting. However, the covariance term $\sigma_{a_i\theta_i}$ enters the information rent twice and has an a priori ambiguous effect. To gain intuition, notice that  the \emph{marginal} rent of any  type $\theta_i$ is proportional to their  expected action,
$$
U'_i(\theta_i)=t\E[a_i\given\theta_i]=t
	\Big(\mu_{a_i}+\frac{\sigma_{a_i\theta_i}}{\sigma_{\theta_i}^2}(\theta_i-\mu_{\theta_i})\Big).
$$
 If the recommendation rule is obedient, the expected action $(\mu_{a_i})$ of the average type $(\theta_i=\mu_{\theta_i})$  is pinned down by the obedience constraints. The  parameter $\sigma_{a_i\theta_i}$ then rotates the expected action of type $\theta_i$ (and hence the marginal rent function) around its ex-ante mean.

\begin{figure}[!t]
\centering
\begin{subfigure}[t]{0.55\textwidth}
\centering
\begin{tikzpicture}[
  scale=.8,   
  >=stealth,   
  every node/.style={font=\small}  
]

  \draw[->] (-0.5,0) -- (5,0) node[right] {$\theta_i$};
  \draw[->] (0,-0.5) -- (0,6.5) node[left] {$U'(\theta_i)$};

  \coordinate (theta_star) at (1,0);
  \coordinate (mu_bi) at (3,0);
  \coordinate (theta) at (4,0);
  \coordinate (mu_ai) at (0,3); 
  \coordinate (line_end) at (5,6); 
  \coordinate (U_end) at (4,4.5); 

  \draw[thick] (theta_star) -- (line_end) 
      node[pos=0.62, above, rotate=56] {$\E[a_i\given\theta_i]=\mu_{a_i} + \sigma_{a_i\theta_i} (\theta_i - \mu_{\theta_i})$};

  \draw[dashed] (mu_ai) --++ (3,0); 
  \draw[dashed] (mu_bi) --++ (0,3); 
  \draw[dashed] (theta) -- (U_end);  

  \fill (theta_star) circle(1.5pt) node[below] {$\theta_i^*$};
  \fill (mu_bi) circle(1.5pt) node[below] {$\mu_{\theta_i}$};
  \fill (theta) circle(1.5pt) node[below] {$\theta$};

  \fill (mu_ai) circle(1.5pt) node[left] {$\mu_{a_i}$};

  \fill[pattern=north west lines, pattern color=black!50] 
      (theta_star) -- (U_end) -- (theta) -- cycle;

  \draw[->] (3.5,1.5) to[out=40,in=220] (4.25,1.5) 
      node[right] {$U_i(\theta_i)=\int_{\theta^*_i}^{\theta_i} \E[a_i\given\theta_i]d\theta_i$};

\end{tikzpicture}
\caption{\small{Higher $\sigma_{a_i\theta_i}\Rightarrow$ rents grow quicker}}
\end{subfigure}
\hfill
\begin{subfigure}[t]{0.44\textwidth}
\centering
\begin{tikzpicture}[
  scale=.8,    
  >=stealth,    
  every node/.style={font=\small}  
]

  \draw[->] (-1.5,0) -- (5,0) node[right] {$\theta_i$};
  \draw[->] (0,-0.5) -- (0,6.5) node[left] {$U'(\theta_i)$};

  \coordinate (theta_star) at (-1,0); 
  \coordinate (mu_bi) at (3,0);
  \coordinate (theta) at (4,0);
  \coordinate (mu_ai) at (0,3);       
  \coordinate (line_end) at (5,4.5);
  \coordinate (U_end) at (4,3.75);     

  \draw[thick] (theta_star) -- (line_end) 
      node[pos=0.62, above, rotate=37] {$\E[a_i\given\theta_i]=\mu_{a_i} + \sigma_{a_i\theta_i} (\theta_i - \mu_{\theta_i})$};

  \draw[dashed] (mu_ai) --++ (3,0); 
  \draw[dashed] (mu_bi) --++ (0,3); 
  \draw[dashed] (theta) -- (U_end);  

  \fill (theta_star) circle(1.5pt) node[below] {$\theta_i^*$};
  \fill (mu_bi) circle(1.5pt) node[below] {$\mu_{\theta_i}$};
  \fill (theta) circle(1.5pt) node[below] {$\theta$};

  \fill (mu_ai) circle(1.5pt) node[left] {$\mu_{a_i}$};

  \fill[pattern=north west lines, pattern color=black!50] 
      (theta_star) -- (U_end) -- (theta) -- cycle;

\end{tikzpicture}
\caption{\small{Lower $\sigma_{a_i\theta_i}\Rightarrow$ higher rent $U_i(\mu_{\theta_i})$}}
\end{subfigure}
\caption{Dual effect of type-action covariance on information rents ($t=1$, $\sigma_{\theta_i}^2=1)$}
\label{fig:margrent}
\end{figure}

\Cref{fig:margrent} illustrates this effect: both a very steep and a very flat $\E[a_i\given\theta_i]$ lead to large information rents. Indeed, a mechanism that is not very responsive to the player's type will necessarily induce a large measure of types to take a positive action in expectation, which increases the marginal rent at every $\theta_i<\mu_{\theta_i}$. This effect is relatively stronger the higher the average action $\mu_{a_i}$. Conversely, a mechanism that is too responsive to types will require the information rent to grow very quickly with $\theta_i$, which also reduces the maximal  payments. These forces carry over to the ex ante information rent, which is equal to
\begin{displaymath}
   \mathbb{E}[U_i(\theta_i)]= \frac t 2\sigma_{a_i\theta_i} +\frac t 2
    \frac{\mu_{a_i}^2\sigma_{\theta_i}^2}{\sigma_{a_i\theta_i}}\geq 0.
\end{displaymath}
Thus, in both the interim and the ex ante expressions for the information rent, the nonlinear term makes the choice of $\sigma_{a_i\theta_i}$ dependent on $\mu_{a_i}^2$. 

This dependence is manifest in the seller's revenue objective. Using \cref{cor:opt-pay-ob} (with the maximal-payment characterization of \cref{prop:max-payment} in \cref{sec:participation}), the highest expected payment of any implementable mechanism that guarantees participation is
\begin{equation}\label{eq:expected-payment}
	    \E[p_i(\theta_i)]
	    = \E[u_i(a;\theta,\omega)] - \frac t 2 \sigma_{a_i\theta_i} -
	    \frac{t\mu_{a_i}^2\sigma_{\theta_i}^2}{2\sigma_{a_i\theta_i}}=
	     \frac{\mu_{a_i}^2} 2 + \frac{\sigma_{a_i}^2} 2 -\frac t 2\sigma_{a_i\theta_i}
	    - \frac{t\mu_{a_i}^2\sigma_{\theta_i}^2}{2\sigma_{a_i\theta_i}},
\end{equation}
where the second equality uses the expected utility of a player in obedient
mechanisms \eqref{eq:expected-utility}.  Because $\mu_{a_i}$ is pinned down by the
obedience constraint, the objective function \eqref{eq:expected-payment} is concave in the  covariance coefficients over the set of incentive compatible
mechanisms.

After substituting $\sigma_{a_i}^2$ using the first linear obedience
constraint \eqref{eq:obedience-var-sym}, maximizing revenue over the class of
symmetric obedient mechanisms takes the form \eqref{eq:reduced} with
\begin{displaymath}
	\tilde F(\sigma_{a_ia_j}, \sigma_{a_i\omega},\sigma_{a_i\theta_j})
	= (n-1)r\sigma_{a_ia_j}+s\sigma_{a_i\omega}
		- \frac{\mu_{a_i}^2}{1+r\sigma_{a_i\theta_j}/ft\sigma_{\theta_i}^2}.
\end{displaymath}

\begin{proposition}\label{prop:opt-revenue}
    For $r\in\big(-1,\frac 1 {n-1}\big)$, there exists a unique Gaussian and symmetric mechanism maximizing revenue subject to obedience. The revenue-optimal recommendations follow the structure of  \cref{prop:opt-welfare}
    with $x(\lambda)$ the unique solution in $(-f/r^2,+\infty)$ of
\begin{equation}\label{eq:cubic}
	\frac{2x}{f} \big[\lambda nf(f-r)-r(1+r)\big]
	- \frac{\lambda nf + 2 + r}{1+r}
	=\frac{\mu_{a_i}^2}{t^2\sigma_{\theta_i}^2}
    \frac 1{(1+r^2x/f)^2},
\end{equation}
	and $\displaystyle y(\lambda)= \frac{1}{2n}\frac{\lambda nf+1}{\lambda(f-r)-r}$.
    \end{proposition}

The proof is in \cref{sec:rev-app}. The geometric interpretation presented below \cref{prop:opt-welfare} applies similarly. The only difference is that the curve $\big(x(\lambda),y(\lambda)\big)$ is no longer a portion of hyperbola, because $x(\lambda)$ now solves the cubic equation \eqref{eq:cubic}.

Crucially, the mean average action $\mu_{a_i}=f(s\mu_{\omega}+t\mu_{\theta_i})/(f-r)$ now enters the expression for the stationarity curve  \eqref{eq:cubic}. This is due to the coupling of the variables $\mu_{a_i}$ and $\sigma_{a_i\theta_i}$ in the players' information rents discussed above, which in turn implies that $\mu_{a_i}$ and $\sigma_{a_i\theta_i}$ are not additively separable in the expected payment objective \eqref{eq:expected-payment}, contrary to what happened for the producer-surplus objective. Hence, the covariance matrix of the optimal mechanism now also depends on the prior means (through $\mu_{a_i}$). While the revenue-optimal recommendations are then deterministic or randomized in a similar way to \cref{prop:welfare-random},  the threshold is now a function $\tau_{\rm R}(n,r,\mu_{a_i}^2,\sigma_{\theta_i}^2)$ which is decreasing in $\mu_{a_i}^2$. This dependency of the revenue-optimal mechanism on $\mu_{a_i}^2$ also impacts the comparison  with the firms-optimal mechanism, which we explore in the following section.

\begin{remark}[Non-Gaussian Mechanisms]\label{rem:gaussian-scope}
The restriction to Gaussian mechanisms is without loss for the firms- and consumer-optimal problems (\cref{sec:welfare} and \cref{sec:consumer}; see also \citet{mu2024} for the corresponding result without elicitation), but is with loss for the revenue-optimal problem. 
A non-Gaussian recommendation rule can separate the dependence of $\E[a_i\given\theta_i]$ from the dependence of $\var(a_i\given\theta_i)$ on type and thereby reduce the rent below the Gaussian minimum. The Gaussian optimum nonetheless attains the first-best revenue when $n=1$ (as discussed in \cref{Sec:bench}), when $n\to\infty$, and when $r\to 0$, so the revenue under the Gaussian restriction is not uniformly bounded away from the optimal revenue. 
\end{remark}

\subsection{Comparison of Obedient Mechanisms}\label{sec:comparison}

In this section, we wish to compare three obedient mechanisms: the complete-information Nash benchmark (denoted by CI, \cref{Sec:bench}); the firms-optimal mechanism (denoted by F, \cref{sec:welfare}); and the revenue-optimal mechanism
(denoted by R, \cref{sec:revenue}). Recall that the mean action recommendation is pinned down by the obedience
constraint in all three mechanisms. 
 Thus,  comparing obedient mechanisms amounts to comparing the covariance
coefficients $(\sigma_{a_i\theta_i}, \sigma_{a_i\theta_j}, \sigma_{a_i\omega})$
and any noise that may be added to the  recommendations. 

\begin{figure}[!t]
    \begin{center}
	\includegraphics{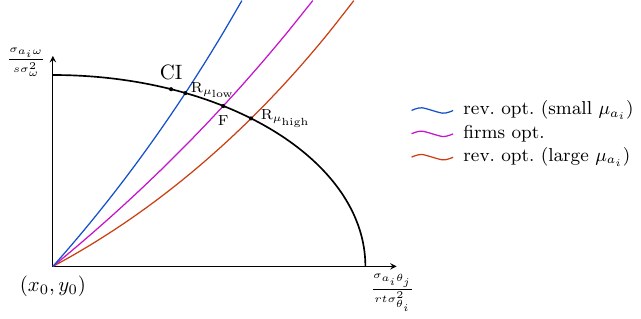}
    \end{center}
    \caption{Comparison of obedient mechanisms.\protect\footnotemark}
\end{figure}

\footnotetext{The reader is invited to vary the parameter $\mu_{a_i}$ in the interactive visualization available at \url{https://thibaut.horel.org/info-coord-viz/} to explore its effect on the revenue-optimal mechanism and on the relative ordering of mechanisms.
}

\begin{proposition}\label{prop:comparison}
	Fix all the exogenous variables except $\mu_\omega$ and
	$\mu_{\theta_i}$. Then:
	\begin{enumerate}
		\item In the firms-optimal mechanism and in the complete-information
			Nash equilibrium, the value of $(\sigma_{a_i\theta_i},
			\sigma_{a_i\theta_j}, \sigma_{a_i\omega})$ is independent of
			$\mu_{a_i}$.
		\item In the revenue-optimal mechanism,
			$\sigma_{a_i\theta_i}^{\rm R}/{t\sigma_{\theta_i}^2}$ and
			$\sigma_{a_i\theta_j}^{\rm R}/{rt\sigma_{\theta_i}^2}$ are increasing
			functions of $\mu_{a_i}^2$, whereas
			$\sigma_{a_i\omega}^{\rm R}/s\sigma_{\omega}^2$ is decreasing in this
			parameter. We write $\sigma_{a_i\theta_i}^{\rm R}(\mu_{a_i}^2)$ to
			make the dependency explicit, and similarly for
			$\sigma_{a_i\theta_j}$, $\sigma_{a_i\omega}$.
		    \item For both $\sigma_{a_i\theta_j}/rt\sigma_{\theta_i}^2$ and $\sigma_{a_i\theta_i}/t\sigma_{\theta_i}^2$, we have
	\begin{displaymath}
		0\leq \frac{\sigma_{a_i\theta_j}^{\rm CI}}{rt\sigma_{\theta_i}^2}
		\leq \frac{\sigma_{a_i\theta_j}^{\rm R}(0)}{rt\sigma_{\theta_i}^2}
		\leq \frac{\sigma_{a_i\theta_j}^{\rm F}}{rt\sigma_{\theta_i}^2}
		\leq \frac{\sigma_{a_i\theta_j}^{\rm R}(\infty)}{rt\sigma_{\theta_i}^2}.
	\end{displaymath}
	 
\item For $\sigma_{a_i\omega}$, the order is reversed,
	\begin{displaymath}
		0\leq \frac{\sigma_{a_i\omega}^{\rm R}(\infty)}{s\sigma_{\omega}^2}
		\leq \frac{\sigma_{a_i\omega}^{\rm F}}{s\sigma_{\omega}^2}
		\leq \frac{\sigma_{a_i\omega}^{\rm R}(0)}{s\sigma_{\omega}^2}
		\leq\frac{\sigma_{a_i\omega}^{\rm CI}}{s\sigma_{\omega}^2}.
	\end{displaymath}
	\end{enumerate}
\end{proposition}

We already discussed below \cref{prop:opt-welfare} how the firms-optimal mechanism places more weight on private types and less weight on the common state than the complete-information outcome. A natural next question is whether this qualitative feature is also present in the seller's revenue-maximizing mechanism. The monopolist information seller wants to limit the players' rents, relative to their welfare-maximizing level. An incomplete intuition would then suggest that the  monopolist will understate the reliance of the mechanism on the private types, but will still qualitatively modify the complete-information signal in the same direction.

Instead, the seller may sometimes overstate the reliance of the mechanism on the private types. To gain more intuition, recall \Cref{fig:margrent} and type $\theta$'s information rent.  
When the market is very profitable (as measured by the exogenous level $\mu_{a_i}$ of average action under obedience), a non-responsive mechanism (i.e., a low $\sigma_{a_i\theta_i}$) results in worst-off  type $\theta^*_i$ that is far from the mean type $\mu_\theta$. This means the most frequent types in the distribution receive larger rents.  Thus, \Cref{fig:margrent}  suggests that  reducing $\sigma_{a_i\theta_i}$ is costlier when $\mu_{a_i}$ is large, and beneficial when $\mu_{a_i}$ is small, so to limit the rate of growth of information rents to the right of the mean.

\subsection{Consumer-Optimal Mechanism}\label{sec:consumer}

We now consider a representative consumer with a linear-quadratic utility that microfounds the demand-shifter Bertrand interpretation of our model (\cref{example:bertrand}).\footnote{An analysis of the corresponding Cournot model is in \cref{app-consumer}.} A simple derivation (see the proof of \cref{prop:bertrand-opt} in \cref{sec:app-consumer}) shows that for the demand curve of \cref{example:bertrand}, when the firms choose a price vector $p\in\R^n$, the consumer surplus becomes
\begin{equation}\label{eq:bertrand-utility}
    W_C(p;\omega,\theta) = \frac{1}{4}\sum_{i=1}^n p_i^2 - \frac{r}{2}\sum_{i\neq j} p_i p_j - s\omega\sum_{i=1}^n p_i - t\sum_{i=1}^n \theta_i\,p_i + C(\omega,\theta),
\end{equation}
where $C(\omega,\theta)$ is a $p$-independent constant. 

Maximizing the consumer's ex-ante surplus is thus equivalent to maximizing
\begin{displaymath}
    \sigma_{a_i}^2 - 2(n-1)r\sigma_{a_ia_j} - 4s\sigma_{a_i\omega} - 4t\sigma_{a_i\theta_i}.
\end{displaymath}
Rewriting this objective under obedience in terms of $\sigma_{a_i\theta_j}$, $\sigma_{a_ia_j}$, and $\sigma_{a_i\omega}$ using \eqref{eq:obedience-var-sym}, the resulting optimization problem takes the form \eqref{eq:reduced} with
\begin{displaymath}
	\tilde F(\sigma_{a_ia_j}, \sigma_{a_i\omega},\sigma_{a_i\theta_j})
	=-(n-1)r\,\sigma_{a_ia_j} - 3s\,\sigma_{a_i\omega} - 3(n-1)rt\,\sigma_{a_i\theta_j},
\end{displaymath}
\looseness=-1
which corresponds to the objective in \cref{prop:linear} with $(\alpha,\beta,\gamma)=(-1,-3,-3)$. Note that the monotonicity in $\sigma_{a_ia_j}$, $\sigma_{a_i\omega}$, and $\sigma_{a_i\theta_j}$ is all reversed relative to firms-welfare and revenue objectives in \cref{sec:welfare} and \cref{sec:revenue}.

\begin{figure}[!t]
\centering
\begin{tikzpicture}[scale=2.0]
    \def\xo{1.5}      
    \def\yo{0.6}      
    \def\xa{1.8}      

    \pgfmathsetmacro{\xb}{\xa*\yo/sqrt(\xa*\xa-\xo*\xo)};

    \draw[->] (\xo-\xa, 0) -- (\xo+1.1*\xa, 0) node[right, font=\small] {$\sigma_{a_i\theta_j}/(rt\sigma_{\theta_i}^2)$};
    \draw[->] (0, \yo-\xb) -- (0, \yo+1.2*\xb) node[above, font=\small] {$\sigma_{a_i\omega}/(s\sigma_{\omega}^2)$};

    \draw[thick] (\xo, \yo) circle [x radius=\xa, y radius=\xb];

    \draw[help lines, gray] (0, 0) rectangle (2*\xo, 2*\yo);

    \draw[red, very thick]
        ({\xo + \xa*cos(20)}, {\yo + \xb*sin(20)})
	arc[start angle=20, end angle=-acos(-\xo/\xa), x radius=\xa, y radius=\xb];

    \fill (0, 0) circle (1.2pt) node[below left, font=\small] {$\varnothing$};
    \fill (2*\xo, 0) circle (1.2pt) node[below right, font=\small] {TO};
    \fill (0, 2*\yo) circle (1.2pt) node[above left, font=\small] {SO};
    \fill (2*\xo, 2*\yo) circle (1.2pt) node[above right, font=\small] {CI};

    \fill (\xo, \yo) circle (0.4pt);

    \fill ({\xo+\xa*cos(20)}, {\yo+\xb*sin(20)}) circle (1.5pt) node[right,xshift=3pt]{F};

    \fill[blue!80!black] (0, 0) circle (1.5pt) node[below right,yshift=-8pt]{B};

\end{tikzpicture}
\caption{Pareto frontier between firms-optimal F and consumer-optimal B in the Bertrand formulation (\cref{example:bertrand}). The red arc on the obedience ellipse covers the IC-slack range $\eta\in[0,1/2)$ of the convex combination of objectives: the optimum traces the boundary clockwise from F through TO toward $\varnothing$ as the consumer's weight rises. At $\eta=1/2$ the IC monotonicity starts to bind, and the optimum then collapses to $\varnothing$ for every $\eta\in(1/2,1]$, which induces the prior Bayes--Nash equilibrium with no information disclosure.}
\label{fig:consumer-opt}
\end{figure}

\begin{proposition}\label{prop:bertrand-opt}
    For $r\in(0,\frac 1{n-1})$, the symmetric, incentive compatible mechanism maximizing the consumer's expected utility under the demand-shifter Bertrand formulation coincides with the prior Bayes--Nash equilibrium point $\varnothing$ of \cref{fig:ic}: it has $\sigma_{a_i\theta_j}=0$ and $\sigma_{a_i\omega}=0$, the IC monotonicity binds, and the mechanism discloses no information about the state $\omega$ or competitor types $\theta_{-i}$.
\end{proposition}

The IC-binding result follows from \cref{prop:linear} with $(\alpha,\beta,\gamma)=(-1,-3,-3)$: the proof of \cref{prop:bertrand-opt} in \cref{sec:app-consumer} shows the unconstrained obedient optimum is a boundary solution with $\sigma_{a_i\theta_j}^*<0$, so the IC monotonicity of \cref{prop:ic-char} binds. The IC-constrained optimum then lies on the segment $[\varnothing,\mathrm{SO}]$ of \cref{fig:ic}; on that segment the consumer's residual objective is strictly decreasing in $\sigma_{a_i\omega}$, so the optimum lies at $\sigma_{a_i\omega}=0$, that is, $\varnothing$.

\paragraph{Firm vs consumer welfare.} We conclude this section with a discussion of the trade-off between the firms and consumer utilities. Specifically, we consider the frontier of incentive-compatible mechanisms that are Pareto efficient for the firms and the consumer.

Because we have convex constraints and linear objectives, a standard duality argument implies that this frontier is exactly the locus of maximizers for a convex combination of our two objectives. Using the expressions for $\tilde F$ above \cref{prop:opt-welfare} and \cref{prop:bertrand-opt}, this combined objective takes the following form under the demand-shifter Bertrand formulation:
\begin{displaymath}
    \tilde F_\eta = (n-1)r(1-2\eta)\,\sigma_{a_ia_j} \,+\, (1-4\eta)\,s\sigma_{a_i\omega} \,+\, (n-1)rt\,(1-4\eta)\,\sigma_{a_i\theta_j},
\end{displaymath}
\looseness=-1
where $\eta\in[0,1]$ is the weight of the combination, with $\eta=0$ (resp.\ $\eta=1$) corresponding to the firms (resp.\ consumer) welfare. 
The objective is linear in the covariance parameters, so \cref{prop:linear} still applies pointwise in $\eta$ to characterize the optimal mechanism. 

The structure of the frontier is illustrated in \Cref{fig:consumer-opt}. A designer aligned with downstream consumers (or with total welfare) issues qualitatively different recommendations. Along the IC-slack portion of the Pareto frontier of \cref{fig:consumer-opt}, raising the consumer weight pulls the optimum clockwise from F through TO into the region where the action-state covariance is negative (low prices when demand is high) while the cross-firm action correlation remains positive. Once the consumer weight exceeds $1/2$, the IC monotonicity binds and the optimum collapses to the prior Bayes--Nash equilibrium $\varnothing$, where it remains for every $\eta\in(1/2,1]$. This is stated formally and proved in \cref{prop:pareto-opt} in \cref{sec:app-consumer}.

\section{Conclusions}\label{sec:conclusion}

We model data provision as a joint problem of information and mechanism design with 
externalities. 
A benevolent designer can  use correlated signals to  facilitate coordination or anti-coordination among competing agents. Achieving this requires both commitment power and the ability  to impose transfers. In the firms-optimal case, the alignment between the designer's and the firms' objectives means that the welfare-maximizing mechanism elicits private information without distorting the distribution of actions. This property fails in the revenue-maximizing case, where the designer must distort allocations and extract information through screening, making the specification of off-path actions central to the mechanism's profitability. 

Our framework can help clarify the distinct roles of asymmetric information and hidden actions in 
optimal mechanism design. In some applications, players fully delegate their actions to the 
platform---for example, through auto-bidding in digital advertising or ``auto-accept'' 
functionality in rental pricing platforms such as RealPage---eliminating obedience constraints. 
We analyze this case in \cref{sec:delegation}. When players retain private information, the 
welfare-optimal mechanism implements the first-best outcome, while the revenue-optimal 
mechanism reduces to a screening problem with a distinct solution that yields higher 
platform profits.  

Finally, our framework extends tractably to asymmetric environments. These include vertically 
integrated designers who maximize the payoff of a single player rather than aggregate welfare, 
as in Amazon’s dual role \citep{kamu23,rodr25}, and regulatory interventions where 
platform-mediated actions affect downstream agents such as consumers. Both extensions relax 
the symmetry assumptions of the benchmark model but remain analyzable within the same 
general framework.
\newpage

\appendix

\section{Interim Payments}\label{sec:app-payments}

The class of mechanisms introduced in \cref{sec:model} restricts the payment to the interim form $p_i:\R\to\R$. The following lemma justifies the restriction: any equilibrium of a direct mechanism with a more general payment $\tilde p_i:\R^n\times\R\to\R$ remains an equilibrium of the same mechanism with the interim-equivalent payment $\bar p_i$ obtained by integrating out $(\theta'_{-i},\omega)$.

\begin{lemma}\label{lem:payment-interim}
Fix a Perfect Bayesian Equilibrium $\sigma$ of a direct mechanism $(\tau,\tilde p)$ with $\tilde p_i:\R^n\times\R\to\R$. Define the interim-equivalent payment $\bar p_i:\R\to\R$ by
\[
\bar p_i(\theta'_i)\eqdef\E_\sigma\big[\tilde p_i(\theta'_i,\theta'_{-i},\omega)\given\theta'_i\big],
\]
where the expectation is under the conditional law of $(\theta'_{-i},\omega)$ given $\theta'_i$ induced by $\sigma$ and the prior. Then $\sigma$ is also an equilibrium of $(\tau,\bar p)$, and yields the same on-path interim expected utility for every player and every type.
\end{lemma}
\begin{proof}
At every on-path information set of player $i$, the interim distribution of $(\theta'_{-i},\omega)$ given $\theta'_i$ is fixed by $\sigma$ and the prior. The law of iterated expectations gives $\E_\sigma[\tilde p_i\given\theta'_i]=\bar p_i(\theta'_i)$ identically. Moreover, types are independent under the prior, and $\sigma_{-i}$ depends only on $(\theta_{-i},\omega)$, so the conditional distribution of $(\theta'_{-i},\omega)$ given player $i$'s own type $\theta_i$ equals its marginal. Hence, for any report $\theta'_i$ on or off the equilibrium path,
\[
\E\big[\tilde p_i(\theta_i',\theta'_{-i},\omega)\given\theta_i\big]=\bar p_i(\theta_i').
\]
The interim expected utility of every report therefore coincides under the two payments, and $\sigma$ remains an equilibrium under $\bar p_i$.
\end{proof}

Within the incentive compatible class produced by the revelation principle, a payment $p_i:\R\to\R$ independent of the other players' reports and state is therefore without loss.

\section{Additional Results and Proofs for Section~\ref{sec:char}}\label{sec:char-app}

\subsection{Symmetry and positive semi-definiteness}\label{sec:symmetry}

Let $\sg_n$ be the symmetric group on $[n]$ and for each permutation
$\pi\in\sg_n$, denote by $P_\pi\in\M_n(\R)$ the permutation matrix with
$(P_\pi)_{i,j} = \mathbf{1}\set{i=\pi(j)}$ for $(i,j)\in[n]^2$. In particular,
for $x\in\R^n$, $P_\pi x$ is the permuted vector whose $i$th coordinate is
$(P_\pi x)_i = x_{\pi^{-1}(i)}$ for $i\in [n]$.

\begin{definition}\label{def:symmetric}
    A mechanism is \emph{symmetric} if $(a,\theta,\omega)$ and $(P_\pi a,P_\pi
    \theta,\omega)$ are identically distributed for each permutation $\pi\in\sg_n$.
\end{definition}

The mean vector and covariance matrix of symmetric mechanisms have a simple structure presented in the next lemma.

\begin{lemma}\label{lemma:sym-char}
	Let $\mu\in\R^{2n+1}$ and $\K\in\M_{2n+1}(\R)$ be respectively the mean
	vector and covariance matrix of a Gaussian mechanism, with the block
	structure indicated in \eqref{eq:params}. Then the mechanism is symmetric
	iff $\mu_a,\mu_\theta,\K_{a,\omega}\in\R 1_n$ and
	$\K_{aa},\K_{a\theta}\in\R I_n+\R J_n$ and $\K_{\theta\theta}\in\R I_n$.
\end{lemma}

\begin{proof}
	For a permutation $\pi\in\sg_n$, let us denote by $\mu^\pi$ and
	$\K^\pi$ the mean vector and covariance matrix of $(P_\pi a, P_\pi
	\theta, \omega)$. An immediate derivation gives
	\begin{equation}\label{eq:perm-params}
		\mu^\pi = \begin{bmatrix}P_\pi\mu_a\\P_\pi\mu_\theta\\\mu_\omega\end{bmatrix}
		\quad\text{and}\quad
		\K^\pi = \begin{bmatrix}
			P_\pi\K_{aa}\tr{P_\pi}&P_\pi\K_{a\theta}\tr{P_\pi}&P_\pi\K_{a\omega}\\[0.5ex]
			P_\pi\tr\K_{a\theta}\tr{P_\pi}&P_\pi\K_{\theta\theta}\tr{P_\pi}&0\\[0.5ex]
			\tr{\K_{a\omega}}\tr{P_\pi}&0&\sigma_{\omega}^2\\
		\end{bmatrix}\,.
	\end{equation}
	Because the mechanism is Gaussian, it is fully determined by its means and
	covariance matrix, hence the mechanism is symmetric iff $\K^\pi=\K$ and
	$\mu^\pi=\mu$ for all $\pi\in\sg_n$. From \cref{prop:matrix}, we
	immediately obtain that $\mu_a,\mu_\theta,\K_{a\omega}\in\spn(1_n)$ and
	$\K_{aa},\K_{a\theta},\K_{\theta\theta}\in\spn(I_n,J_n)$. Since the prior on $\theta$ is independent, $\K_{\theta\theta}$ is diagonal,
	implying that $\K_{\theta\theta}\in\spn(I_n)$.
\end{proof}

\begin{lemma}\label{lemma:psd}
	The covariance matrix $\K$ of a symmetric mechanism is positive semidefinite iff
	\begin{enumerate}
		\item $\sigma_{a_i\theta_i}=\sigma_{a_i\theta_j}=0$ whenever $\sigma_{\theta_i}^2=0$, and
			$\sigma_{a_i\omega}=0$ whenever $\sigma_\omega^2=0$.
		\item The following inequality constraints hold
			with the convention $0/0=0$.
	\begin{displaymath}
		\begin{cases}
			\frac{1}{\sigma_{\theta_i}^2}(\sigma_{a_i\theta_i}-\sigma_{a_i\theta_j})^2
			\leq \sigma_{a_i}^2 - \sigma_{a_ia_j}\\
			\frac{1}{\sigma_{\theta_i}^2}(\sigma_{a_i\theta_i}+(n-1)\sigma_{a_i\theta_j})^2+\frac{n}{\sigma_\omega^2}\sigma_{a_i\omega}^2
			\leq \sigma_{a_i}^2 + (n-1)\sigma_{a_ia_j}
		\end{cases},
	\end{displaymath}
	Furthermore, there is equality in the first inequality iff
	$\cov(a_i,a_j\given \theta,\omega) = \var(a_i\given \theta,\omega)$ and in
	the second inequality iff $\cov(a_i,a_j\given \theta,\omega) =
	-\var(a_i\given \theta,\omega)/(n-1)$.
	\end{enumerate}
\end{lemma}

\begin{proof}
	Consider a symmetric mechanism with covariance matrix $\K$. We use the
	alternative pa\-ra\-metrization of the mechanism provided by
	\eqref{eq:params-bis}. From $\K_{a\omega} = \sigma_\omega^2\beta$, we
	obtain that $\sigma_\omega^2=0$ implies $\K_{a\omega}=0$. Since
	$\K_{a\omega}\in\spn(1_n)$ by \cref{lemma:sym-char}, $\sigma_\omega^2\neq
	0$ implies that $\beta\in\spn(1_n)$ with
		$\beta_i = \frac{\sigma_{a_i\omega}}{\sigma_\omega^2}$.
	Similarly, from $\K_{a\theta} = \Gamma\K_{\theta\theta}$ and $\K_{\theta\theta}=\sigma_{\theta_i}^2I_n$ we get that $\sigma_{\theta_i}^2 = 0$ implies $\K_{a\theta}=0$. Since furthermore,
	$\K_{a\theta}\in\spn(I_n,J_n)$ by \cref{lemma:sym-char}, $\sigma_{\theta_i}^2\neq 0$
	implies that $\Gamma\in\spn(I_n,J_n)$ and for all $(i,j)\in[n]^2$
		$\gamma_{i,j} = \frac{\sigma_{a_i\theta_j}}{\sigma_{\theta_i}^2}$.
	The only constraint on parametrization
	\eqref{eq:params-bis} is that $\K_\eps$ be positive semidefinite. Using
	that $\K_\eps = \K_{aa} - \sigma_\omega^2\beta\tr\beta
	- \Gamma\tr\Gamma\K_{\theta\theta}$, we see that $\K_\eps$ is also in
	$\spn(I_n,J_n)$ and we compute the on- and off-diagonal entries of $\K_\eps$.
	\begin{equation}\label{eq:ke}
		\begin{gathered}
		(\K_\eps)_{ii}=\sigma_{a_i}^2 - \sigma_\omega^2\beta_i^2 - \sum_{k=1}^n
		\sigma_{\theta_k}^2\gamma_{ik}^2
		=\sigma_{a_i}^2 - \frac{\sigma_{a_i\omega}^2}{\sigma_\omega^2}
		- \frac{\sigma_{a_i\theta_i}^2}{\sigma_{\theta_i}^2}
		- (n-1)\frac{\sigma_{a_i\theta_j}^2}{\sigma_{\theta_i}^2}\\
		(\K_\eps)_{ij}=\sigma_{a_ia_j} - \sigma_\omega^2\beta_i\beta_j - \sum_{k=1}^n
		\sigma_{\theta_k}^2\gamma_{ik}\gamma_{jk}
		=\sigma_{a_ia_j} - \frac{\sigma_{a_i\omega}^2}{\sigma_\omega^2}
		-2\frac{\sigma_{a_i\theta_i}\sigma_{a_i\theta_j}}{\sigma_{\theta_i}^2}
		-(n-2)\frac{\sigma_{a_i\theta_j}^2}{\sigma_{\theta_i}^2}
	\end{gathered}
\end{equation}
Note, that the previous expressions remain valid when $\sigma_\omega^2=0$
(implying $\K_{a\omega}=0$), or $\sigma_{\theta_i}^2=0$ (implying $\K_{a\theta}=0)$, by
adopting the convention $0/0=0$. \Cref{prop:matrix} states that $\K_\eps$
is positive semidefinite iff
	$(\K_\eps)_{ii}\geq (\K_\eps)_{ij}$, equivalently with \eqref{eq:ke},
	\begin{displaymath}
		\sigma_{a_i}^2-\sigma_{a_ia_j}\geq
		\frac{1}{\sigma_{\theta_i}^2}(\sigma_{a_i\theta_i}-\sigma_{a_i\theta_j})^2,
	\end{displaymath}
	and $(\K_\eps)_{ii}\geq -(n-1)(\K_\eps)_{ij}$, or equivalently,
	\begin{align*}
		\sigma_{a_i}^2+(n-1)\sigma_{a_ia_j}
		&\geq n\frac{\sigma_{a_i\omega}^2}{\sigma_\omega^2}
		+\frac{\sigma_{a_i\theta_i}^2}{\sigma_{\theta_i}^2}
		+2(n-1)\frac{\sigma_{a_i\theta_i}\sigma_{a_i\theta_j}}{\sigma_{\theta_i}^2}
		+(n-1)^2\frac{\sigma_{a_i\theta_j}^2}{\sigma_{\theta_i}^2}\\
		&= n\frac{\sigma_{a_i\omega}^2}{\sigma_\omega^2}
		+\frac{1}{\sigma_{\theta_i}^2}\big(\sigma_{a_i\theta_i}
		+ (n-1)\sigma_{a_i\theta_j}\big)^2.
	\end{align*}
	Finally, since $\K_\eps$ is the covariance matrix of $a\given
	\theta,\omega$, the previous two inequalities are equivalent
	to $\var(a_i\given \theta,\omega)\geq \cov(a_i,a_j\given \theta,\omega)$ and
	$(n-1)\cov(a_i,a_j\given \theta,\omega)\geq -\var(a_i\given \theta,\omega)$
	respectively.
\end{proof}

\subsection{Obedience Constraints}\label{sec:obedience-app}

\begin{proposition}\label{prop:obedience-char}
    Assume that $r\notin\set{-1,\frac 1 {n-1}}$, then $\mu$ and $\K$ are respectively the mean vector and covariance matrix of an obedient mechanism iff
    \begin{enumerate}
	\item The mean action $\mu_{a_i}$ of each player $i\in[n]$ is determined by the prior's mean:
	    \begin{equation}\label{eq:obedience-mean}
		\mu_{a_i} = \frac{s\mu_\omega+t\mu_{\theta_i}}{1-(n-1)r}
		+ \frac{r\cdot t\sum_{j\neq i}(\mu_{\theta_j}-\mu_{\theta_i})}
		{(1+r)(1-(n-1)r)}.
	    \end{equation}
	\item The covariance matrix $\K$ satisfies the following linear constraints for each $i\in[n]$
	    \begin{equation}\label{eq:obedience-var}
		\begin{cases}
		    \sigma_{a_i}^2 = r\sum_{j\neq i}\sigma_{a_ia_j}
		    + s\sigma_{a_i\omega} + t\sigma_{a_i\theta_i}\\
		    \sigma_{a_i\theta_i} = r\sum_{j\neq i}\sigma_{a_j\theta_i}
		    + t\sigma_{\theta_i}^2
		\end{cases}.
	    \end{equation}
    \end{enumerate}
\end{proposition}

\begin{proof}
Because the mechanism is Gaussian, the random variable on the right-hand side
of \eqref{eq:obedience} is normal and $(a_i,\theta_i)$-measurable. It is thus
fully determined by its mean and its covariances with $a_i$ and $\theta_i$.
Consequently \eqref{eq:obedience} is equivalent to
\begin{equation}\label{eq:obedience-bis}
	\begin{cases}
		\mu_{a_i} = r\sum_{j\neq i}\mu_{a_j} + s\mu_\omega + t\mu_{\theta_i}\\
		\sigma_{a_i}^2 = r\sum_{j\neq i}\sigma_{a_ia_j} + s\sigma_{a_i\omega}
		+ t\sigma_{a_i\theta_i}\\
		\sigma_{a_i\theta_i} = r\sum_{j\neq i}\sigma_{a_j\theta_i} +t \sigma_{\theta_i}^2
	\end{cases},
\end{equation}
where the first equation expresses the equality of means in
\eqref{eq:obedience}, and the second (resp.\ third)  equation expresses the
equality of the covariance with $a_i$ (resp.\ $\theta_i$) in \eqref{eq:obedience}.
We used the identity $\cov(\E[X|Y], Z)=\cov(X, Z)$
for random variables $(X,Y,Z)$ such that $Z$ is $Y$-measurable.

Collecting the first equation in \eqref{eq:obedience-bis} for each
$i\in[n]$ gives the linear system
	$J_n(1,-r)\mu_a = s\mu_\omega 1_n + t\mu_\theta$, whose solution 
    (using \cref{prop:matrix}) yields \eqref{eq:obedience-mean}. The second and third equations in
    \eqref{eq:obedience-bis} are precisely \eqref{eq:obedience-var}.
\end{proof}

\begin{proof}[Proof of \cref{cor:obedience-char-bis}]
    We first use \eqref{eq:params-sym} to express $\sigma_{a_i}^2$ and $\sigma_{a_ia_j}$ in terms of the remaining covariance parameters as well as $\delta$ and $\rho$:
    \begin{displaymath}
	\begin{cases}\sigma_{a_i}^2=\frac{\sigma_{a_i\omega}^2}{\sigma_{\omega}^2}
	+\frac{\sigma_{a_i\theta_i}^2}{\sigma_{\theta_i}^2}
	+\frac{\sigma_{a_i\theta_j}^2}{f\sigma_{\theta_i}^2}+\delta^2\\
	\sigma_{a_ia_j}=\frac{\sigma_{a_i\omega}^2}{\sigma_{\omega}^2}
	+2\frac{\sigma_{a_i\theta_i}\sigma_{a_i\theta_j}}{\sigma_{\theta_i}^2}
	+(n-2)\frac{\sigma_{a_i\theta_j}^2}{\sigma_{\theta_i}^2}+\rho\delta^2
	\end{cases}.
    \end{displaymath}
    Substituting in the first obedience constraint \eqref{eq:obedience-var-sym} and multiplying by $f$
    \begin{multline*}
f\frac{\sigma_{a_i\omega}^2}{\sigma_{\omega}^2}
	+f\frac{\sigma_{a_i\theta_i}^2}{\sigma_{\theta_i}^2}
	+\frac{\sigma_{a_i\theta_j}^2}{\sigma_{\theta_i}^2}+f\delta^2\\
=r\frac{\sigma_{a_i\omega}^2}{\sigma_{\omega}^2}
	+2r\frac{\sigma_{a_i\theta_i}\sigma_{a_i\theta_j}}{\sigma_{\theta_i}^2}
	+(n-2)r\frac{\sigma_{a_i\theta_j}^2}{\sigma_{\theta_i}^2}+\rho r\delta^2
	+sf\sigma_{a_i\omega}+tf\sigma_{a_i\theta_i},
\end{multline*}
and after reordering the terms and completing a square
\begin{displaymath}
    (f-r)\frac{\sigma_{a_i\omega}^2}{\sigma_{\omega}^2}
	+f\left(\frac{\sigma_{a_i\theta_i}}{\sigma_{\theta_i}}-\frac{r\sigma_{a_i\theta_j}}{f\sigma_{\theta_i}}\right)^2
	+\frac{\sigma_{a_i\theta_j}^2}{\sigma_{\theta_i}^2}\frac{(1+r)(f-r)}{f}
    =(\rho r-f)\delta^2
	+sf\sigma_{a_i\omega}+tf\sigma_{a_i\theta_i}.
\end{displaymath}
    We substitute out $\sigma_{a_i\theta_i}$ using the second obedience constraint in \eqref{eq:obedience-var-sym} and divide by $(f-r$)
\begin{displaymath}
    \frac{\sigma_{a_i\omega}^2}{\sigma_{\omega}^2}
    -\frac{sf}{f-r}\sigma_{a_i\omega}
	+\frac{\sigma_{a_i\theta_j}^2}{\sigma_{\theta_i}^2}\frac{(1+r)}{f}
    -\frac{rt}{f-r}\sigma_{a_i\theta_j}
    =\delta^2\frac{\rho r-f}{f-r},
\end{displaymath}
and completing the squares on the left-hand side finally yields
\begin{multline*}
	\frac{(1+r)r^2t^2\sigma_{\theta_i}^2}{f}\left[\frac{\sigma_{a_i\theta_j}}{rt\sigma_{\theta_i}^2}
    -\frac{f}{2(1+r)(f-r)}\right]^2
    -\frac{r^2t^2f\sigma_{\theta_i}^2}{4(1+r)(f-r)^2}\\
    +s^2\sigma_{\omega}^2\left[\frac{\sigma_{a_i\omega}}{s\sigma_{\omega}^2}
    -\frac{f}{2(f-r)}\right]^2 - \frac{s^2f^2\sigma_{\omega}^2}{4(f-r)^2}
    =\delta^2\frac{\rho r-f}{f-r}.
\end{multline*}
    Recognizing on the left-hand side the slack $\xi$ of the ellipse constraint \eqref{eq:ellipse}, we obtain
    \begin{equation}\label{eq:slack-foo}
	\delta^2 = \xi\frac{f-r}{f-\rho r}.
    \end{equation}
    As discussed in the proof of \cref{lemma:psd}, the covariance matrix of the noise is $\K_\eps=(1-\rho)I_n+\rho J_n$, and by \cref{prop:matrix}, positive semi-definiteness of $\K_\eps$ is equivalent to $-f\leq \rho\leq 1$. When this constraint is satisfied, and for $r\in(-1,f)$, the multiplicative factor $(f-r)/(f-\rho r)$ is positive, hence \eqref{eq:slack-foo} implies that $\xi\geq 0$, which is exactly the ellipse constraint \eqref{eq:ellipse}.

    Conversely, if the ellipse constraint \eqref{eq:ellipse} is satisfied, then \eqref{eq:slack-foo} uniquely defines $\delta$ in such a way that the first covariance obedience constraint is satisfied (since \eqref{eq:slack-foo} is equivalent to it by our derivation). Finally, defining $\sigma_{a_i\theta_i}$ and $\mu_{a_i}$ using \eqref{eq:obedience-mean-sym} and \eqref{eq:obedience-var-sym} respectively, guarantees that the remaining obedience constraints are satisfied.
\end{proof}

\subsection{Incentive Compatibility}\label{sec:app-tic}

\begin{proof}[Proof of \cref{prop:ic-char}]
The optimal second-stage action after a misreport is derived in the body: \eqref{eq:deviation} gives $a_i'=a_i+t(\theta_i-\theta_i')$, and the resulting misreport interim utility is \eqref{eq:tu-misreport}. Incentive compatibility requires truthful reporting $\theta_i'=\theta_i$ to maximize $\theta_i'\mapsto\tu_i(\theta_i';\theta_i)-p_i(\theta_i')$ for every type $\theta_i$. Using the on-path utility \eqref{eq:interim-utility-ob}, we can write the misreport interim utility as
\begin{displaymath}
    \tu_i(\theta_i';\theta_i)
    = \tu_i(\theta_i') + t\,\E[a_i\given\theta_i']\,(\theta_i-\theta_i')
    + \frac{t^2}{2}(\theta_i-\theta_i')^2.
\end{displaymath}
Hence, writing $V_i\eqdef \tu_i-p_i$ for the on-path rent, truthful reporting is optimal iff, for all $\theta_i,\theta_i'$,
\begin{equation}\label{eq:ic-sconvex}
    V_i(\theta_i)
    \geq V_i(\theta_i')
    + t\,\E[a_i\given\theta_i']\,(\theta_i-\theta_i')
    + \frac{t^2}{2}(\theta_i-\theta_i')^2
\end{equation}
This inequality states precisely that $V_i$ is $t^2$-strongly convex with subgradient $t\E[a_i\given\theta_i]$ for each $\theta_i$. Moreover, the $t^2$-strong convexity is equivalent to the $t^2$-strong monotonicity of the subgradient, that is, for all $\theta_i\neq\theta_i'$,
\begin{displaymath}
    t\,\frac{\E[a_i\given\theta_i]-\E[a_i\given\theta_i']}{\theta_i-\theta_i'}
\geq t^2
\end{displaymath}

\looseness=-1
    For a Gaussian mechanism (when $\sigma_{\theta_i}^2>0$), the conditional expectation is affine in the type,
    \begin{equation}\label{eq:cond-mean-affine}
    \E[a_i\given\theta_i]
    = \mu_{a_i} + \frac{\sigma_{a_i\theta_i}}{\sigma_{\theta_i}^2}(\theta_i-\mu_{\theta_i}),
\end{equation}
hence continuous in $\theta_i$. Since $V_i$ is convex, hence absolutely continuous, with a continuous subgradient, it is in fact differentiable everywhere: for each $\theta_i$,
    \begin{displaymath}
	p_i'(\theta_i) = \tu_i'(\theta_i)-t\E[a_i\given\theta_i].
    \end{displaymath}
    Because $\var(a_i\given\theta_i)$ does not depend on $\theta_i$ under a Gaussian mechanism, \eqref{eq:interim-utility-ob} and \eqref{eq:cond-mean-affine} give $\tu_i'(\theta_i)=\E[a_i\given\theta_i]\sigma_{a_i\theta_i}/\sigma_{\theta_i}^2$, hence
    \begin{displaymath}
	    p_i'(\theta_i)=
    \prn[\bigg]{\frac{\sigma_{a_i\theta_i}}{\sigma_{\theta_i}^2} -t}
    \E[a_i\given \theta_i]
    = \prn[\bigg]{\frac{\sigma_{a_i\theta_i}}{\sigma_{\theta_i}^2} -t}
    \prn[\big]{\mu_{a_i} + \frac{\sigma_{a_i\theta_i}}{\sigma_{\theta_i}^2}
    (\theta_i-\mu_{\theta_i})},
    \end{displaymath}
    where the second equality uses \eqref{eq:cond-mean-affine} again. This establishes part~1 of the characterization.

Furthermore, by \eqref{eq:cond-mean-affine}, $\theta_i\mapsto t\,\E[a_i\given\theta_i]$ has constant derivative $t\sigma_{a_i\theta_i}/\sigma_{\theta_i}^2$, so the $t^2$-strong monotonicity reduces to $t\sigma_{a_i\theta_i}/\sigma_{\theta_i}^2\geq t^2$, i.e.\ $t\sigma_{a_i\theta_i}\geq t^2\sigma_{\theta_i}^2$. Since $\sigma_{a_i\theta_i}-t\sigma_{\theta_i}^2=r\sum_{j\neq i}\sigma_{a_j\theta_i}$ by obedience \eqref{eq:obedience-var-sym}, this is equivalent to $rt\sum_{j\neq i}\sigma_{a_j\theta_i}\geq 0$, which is part~2.

Finally, when $\sigma_{\theta_i}^2=0$ the type distribution is degenerate, so $\sigma_{a_i\theta_i}=0$ and part~2 holds trivially, with $\E[a_i\given\theta_i']=\mu_{a_i}$ constant. In this case, $p_i'(\theta_i)=-t\mu_{a_i}$.
\end{proof}

\section{Additional Results and  Proofs for Section~\ref{sec:opt}}\label{sec:opt-app}

\subsection{Structural Properties}\label{sec:structure-app}

\begin{proposition}\label{prop:structural-main}
    Consider the problem \eqref{eq:reduced} with $r\in\big(-1,\frac 1 {n-1}\big)$ and assume that $\tilde F$ is concave, differentiable and strictly monotone in $\sigma_{a_ia_j}$. Then in an optimal mechanism, the action recommendations are maximally positively (resp.\ negatively) correlated conditioned on $(\theta,\omega)$ whenever $\tilde F$ is increasing (resp.\ decreasing).
\end{proposition}

\begin{proof}[Proof of \cref{prop:structural-main}]
    Since the problem \eqref{eq:reduced} is convex, the Karush–Kuhn–Tucker conditions are necessary for optimality. The stationarity condition associated with $\sigma_{a_ia_j}$ yields
    \begin{displaymath}
	-f\frac{\partial\tilde F}{\partial \sigma_{a_ia_j}} +\lambda(f-r)-\nu(1+r)=0,
    \end{displaymath}
    where $\lambda$ and $\nu$ are the non-negative multipliers associated with the first and second inequality constraint respectively.

    If $\frac{\partial\tilde F}{\partial \sigma_{a_ia_j}}>0$, we can solve the stationarity condition for $\lambda$ and obtain
    \begin{displaymath}
	\lambda = \frac{\nu(1+r) + f \partial\tilde F/\partial\sigma_{a_ia_j}}{f-r}>0,
    \end{displaymath}
    where the inequality uses that $-1\leq r\leq f$. Since $\lambda>0$, complementary slackness implies that the first inequality constraint is biding. Hence by \cref{lemma:psd}, the action recommendations are maximally positively correlated conditioned on $(\theta,\omega)$.

    Similarly, if $\frac{\partial\tilde F}{\partial \sigma_{a_ia_j}}<0$, we solve the stationarity condition for $\nu$ and obtain
    \begin{displaymath}
	\nu = \frac{\lambda(f-r) - f \partial\tilde F/\partial\sigma_{a_ia_j}}{r+1}>0.
    \end{displaymath}
    This time, complementary slackness implies that the second inequality constraint binds, hence the action recommendations are maximally negatively correlated conditioned on $(\theta,\omega)$.
\end{proof}

\begin{proof}[Proof of \cref{prop:linear}]
    The claim about action correlations is an immediate consequence of \cref{prop:structural-main}. We introduce non-negative Lagrange multipliers $\lambda$ and $\nu$ for the first and second inequality constraints in \eqref{eq:reduced} respectively and look for a solution to the KKT conditions. The stationarity condition yields the following system:
\begin{displaymath}
    \begin{cases}
	-\alpha r +\big(f-r\big)\lambda -(r+1)\nu = 0\\
	-\beta s -s\lambda +
	\nu\left(\frac{2n}{\sigma_\omega^2}\sigma_{a_i\omega} -s\right) = 0\\
	2\lambda(f-r)
	\Big[\frac{f-r}{f}\frac{\sigma_{a_i\theta_j}}{\sigma_{\theta_i}^2}-t\Big]
	+2\nu(1+r)
	\Big[\frac{1+r}{f}\frac{\sigma_{a_i\theta_j}}{\sigma_{\theta_i}^2}+t\Big]
	= rt(\lambda+\nu+\gamma).
    \end{cases}
\end{displaymath}
From this, we express $\nu$, $\sigma_{a_i\omega}$ and $\sigma_{a_i\theta_j}$ as a function of $\lambda$:
\begin{displaymath}
    \begin{gathered}
    \nu = \frac{\lambda(f-r)-\alpha r}{r+1},\quad
    \sigma_{a_i\omega}
	= \frac{s\sigma_\omega^2}{2n}\frac{\lambda nf-\alpha r + \beta(r+1)}
	{\lambda(f-r)-\alpha r},\\
    \sigma_{a_i\theta_j} =
	\frac{frt\sigma_{\theta_i}^2}{2(r+1)}
	\frac{\lambda nf + \alpha(r+2) +\gamma(r+1)}
	{\big[\lambda nf(f-r) -\alpha r(1+r)\big]}.
    \end{gathered}
\end{displaymath}
    Using that $\nu\geq 0$ by dual feasibility, we get that $\lambda\geq\lambda_{\rm min}\eqdef\max\set{0, \alpha r/(f-r)}$ and the denominators in the expressions for $\sigma_{a_i\omega}$ and $\sigma_{a_i\theta_j}$ cannot vanish when $r\in(-1,f)$.

    The variables $x\eqdef \sigma_{a_i\theta_j}/rt\sigma_{\theta_i}^2$ and $y\eqdef\sigma_{a_i\omega}/s\sigma_{\omega}^2$ are thus constrained to lie on the curve $\big(x(\lambda),y(\lambda)\big)$, where $x$ and $y$ are the two functions given in the statement of the proposition. Moreover $x(\lambda)$ and $y(\lambda)$ must satisfy \eqref{eq:ellipse} by \cref{prop:obedience-char-sym}. We thus have two cases:
\begin{enumerate}
    \item The constraint \eqref{eq:ellipse} is violated at $\lambda_{\rm min}$. In this case, since $x(\lambda)$ and $y(\lambda)$ are monotone and converge to $x_0$ and $y_0$ respectively, there exists a unique $\lambda^\star\in(\lambda_{\rm min},+\infty)$ such that $\big(x(\lambda^\star),y(\lambda^\star)\big)$ lies on the boundary of \eqref{eq:ellipse}.
    \item The constraint \eqref{eq:ellipse} is satisfied at $\lambda_{\rm min}$, in which case we define $\lambda^\star=\lambda_{\rm min}$.
\end{enumerate}
In either case, complementary slackness holds at $\lambda^\star$ and we have
the optimal solution.
\end{proof}

\begin{proposition}\label{prop:ic-slack-regime}
    Consider problem \eqref{eq:reduced} with the linear objective from \cref{prop:linear}, with parameters $(\alpha,\beta,\gamma)\in\R^3$ and $r\in(-1,\frac 1 {n-1})$. Let $\lambda_{\min} \eqdef \max\set[\big]{0,\,\alpha r/(f-r)}$, and define
\begin{displaymath}
N(\lambda) \eqdef \lambda nf+  \alpha(r+2) + \gamma(r+1)
    \quad\text{and}\quad
\lambda_0 \eqdef-\frac{\alpha(r+2)+\gamma(r+1)}{nf},
\end{displaymath}
\looseness=-1
    the numerator of $x(\lambda)$ in \cref{prop:linear} and its unique zero.
    Then, the monotonicity condition in \cref{prop:ic-char}, $rt\sigma_{a_j\theta_i}\geq 0$, is satisfied by the obedient-optimum mechanism of \cref{prop:linear} iff
    \begin{displaymath}
	N(\lambda_{\rm min})\geq 0 \quad\text{or}\quad\big(N(\lambda_{\rm min})<0 \text{ and } y(\lambda_0)\notin (0,2y_0)\big),
    \end{displaymath}
    in which case the mechanism is incentive compatible. Otherwise, the monotonicity condition is violated, and binds at the incentive-compatible optimum.
\end{proposition}

\begin{proof}
    Recall that the denominator of $x(\lambda)$ in \cref{prop:linear} is positive on $[\lambda_{\rm min},+\infty)$ and that $x$ is a monotone function with $\lim_{\lambda\to\infty} x(\lambda)=x_0>0$. Thus, if $N(\lambda_{\min})\geq 0$, then the stationarity curve $\big(x(\lambda),y(\lambda)\big)$ is entirely contained in the half-plane $x\geq 0$. In particular, $x(\lambda^\star)\geq 0$ at the optimum.

    If $N(\lambda_{\rm min})<0$, then the stationarity curve originates in the half-plane $x<0$ and crosses the vertical axis $x=0$ at $\lambda_0>\lambda_{\rm min}$. In this case:
    \begin{itemize}
	\item if $y(\lambda_0)\in(0,2 y_0)$, then the point $\big(x(\lambda_0),y(\lambda_0)\big)$ is in the interior of the obedience ellipse~\ref{eq:ellipse}. Therefore the obedience-optimum mechanism is attained at $\lambda^\star\in[\lambda_{\rm min},\lambda_0)$ for which $x(\lambda^\star)<0$.
	\item otherwise, $y(\lambda_0)\notin(0,2 y_0)$, hence the point $\big(x(\lambda_0), y(\lambda_0)\big)$ is either outside \ref{eq:ellipse} or on its boundary. The obedience-optimum mechanism is thus attained at $\lambda^\star\geq\lambda_0$ for which $x(\lambda^\star)\geq0$.\qedhere
    \end{itemize}
\end{proof}

\subsection{Firms-Optimal Mechanism}\label{sec:welfare-app}

The following lemma will be used repeatedly in this and subsequent sections to determine whether the optimal recommendations are deterministic or randomized conditioned on $\theta,\omega$.

\begin{lemma}\label{lem:ellipse}
    Consider a point $P=(x,y)$ with $x> 2x_0$. Then the ellipse constraint \eqref{eq:ellipse} at $P$ is equivalent to
    \begin{equation}\label{eq:rand-threshold}
	\frac{t^2\sigma_{\theta_i}^2}{s^2\sigma_{\omega}^2}
	\leq \tau_{n,r}(x,y) \eqdef \frac{fx_0}{r^2y_0} \frac{y(2y_0-y)}{x(x-2x_0)}.\qedhere
    \end{equation}
    In particular, (i) if $y> 2y_0$, \eqref{eq:ellipse} is always violated;
	(ii) if $0<y\leq2y_0$, \eqref{eq:ellipse} is satisfied or violated according to \eqref{eq:rand-threshold}. Moreover, $x\mapsto \tau_{n,r}(x,y)$ is decreasing.
\end{lemma}

\begin{proof}
    Note that $y_0/x_0 = 1+r$, hence we can rewrite the ellipse constraint as
    \begin{displaymath}
	\frac{r^2y_0}{fx_0} (x-x_0)^2 + \rho (y-y_0)^2\leq 
	\frac{r^2y_0}{fx_0} x_0^2 + \rho y_0^2,
    \end{displaymath}
    with $\rho\eqdef s^2\sigma_\omega^2/t^2\sigma_{\theta_i}^2$. This is easily seen to be equivalent to \eqref{eq:rand-threshold} and the lemma follows.
\end{proof}

\begin{proof}[Proof of \cref{prop:welfare-random}]
    Recall from the discussion below \cref{prop:linear} that the recommendations are deterministic or randomized conditioned on $\theta,\omega$ according to whether $P_{\rm min}=\big(x(\lambda_{\rm min}), y(\lambda_{\rm min})\big)$ lies outside \ref{eq:ellipse}. We will thus apply \cref{lem:ellipse} to $P_{\rm min}$.
    
    First, we verify that $x(\lambda_{\rm min})>2x_0$. Using the expression for $x(\lambda)$ in the proof of \cref{prop:opt-welfare}, these are respectively equivalent for $r<0$ and $r>0$ to 
    \begin{displaymath}
	(r+1)2f+(f-r)\geq 0
	\quad\text{and}\quad
	(r+1)f+2(f-r)\geq 0,
    \end{displaymath}
    which are easily seen to be true when $r\in(-1,f)$. Next, observe that when $r>0$, $y(\lambda_{\rm min}^+)=+\infty$ and when $r<0$, $y(\lambda_{\min})>2y_0$ is equivalent to $r> -1/(n+1)$. Hence the conclusion immediately follows from \cref{lem:ellipse} which also gives an expression for $\tau_{\rm F}(n,r)$:
\begin{displaymath}
    \tau_{\rm F}(n,r)=-\frac{(1+r)^3\big(1+(n+1)r\big)}{n^2r^2(2r+3)\big(f(2r+3)-r\big)}.\qedhere
\end{displaymath}
\end{proof}

\subsection{Revenue-Optimal Mechanism}\label{sec:rev-app}

\begin{lemma}\label{lemma:implicit-eq}
	Consider the following parametrized equation in the unknown $x$:
\begin{equation}\label{eq:main-inter}
    b(\lambda)\cdot x-c(\lambda) = \frac{a}{(1+dx)^2}
\end{equation}
Assume that $a,d\geq 0$, and $b:I\to\R$ is positive and increasing for some
interval $I$. For each $\lambda\in I$, define $\tilde x(\lambda)\eqdef
c(\lambda)/b(\lambda)$ and assume that $\tilde x$ is decreasing over $I$. Then,
\begin{enumerate}
	\item For each $\lambda\in I$, \eqref{eq:main-inter} has a unique solution
		$x(\lambda)$ in the interval $(-1/d,+\infty)$.
	    \item For each $\lambda\in I$, $x(\lambda)\geq\tilde x(\lambda)$ and $x(\lambda)$ is increasing with $a$.
	\item The function $x$ is decreasing over $I$.
	\item If $\lim_{\lambda\to\sup I}b(\lambda)=+\infty$ then
		$\lim_{\lambda\to\sup I} x(\lambda)=\max\set[\big]{-1/d,\lim_{\lambda\to\sup I}\tilde
		x(\lambda)}$.
\end{enumerate}
\end{lemma}

\begin{proof}
    \begin{enumerate}
	\item This follows immediately from the following two observations:
	    \begin{itemize}
		\item for each $\lambda\in I$, the function $f_\lambda(x)=b(\lambda)\cdot x - c(\lambda)$ is increasing with $\lim_{x\to -\frac 1 d} f_\lambda(x)<+\infty$ and $\lim_{x\to\infty} f_\lambda(x)=+\infty$, due to $b(\lambda)>0$;
    \item the function $g(x)=a/(1+dx)^2$ on the right-hand side of \eqref{eq:main-inter}
	is decreasing with $\lim_{x\to (-1/d)^+}g(x) =+\infty$ and $\lim_{x\to\infty}
	g(x)=0$ due to $a,d\geq 0$.
	    \end{itemize}
	\item Subtracting $b(\lambda)\tilde x(\lambda) - c(\lambda)=0$ from both sides	of \eqref{eq:main-inter}, we have
	    \begin{equation}\label{eq:x-foo}
		b(\lambda)\cdot\big[x(\lambda)-\tilde x(\lambda)\big] =
		\frac{a}{\big(1+dx(\lambda)\big)^2},
	    \end{equation}
	from which $x(\lambda)\geq \tilde x(\lambda)$ follows due to $a,d\geq 0$. That $x(\lambda)$ is increasing with $a$ follows by differentiating \eqref{eq:main-inter} with respect to $a$ (for fixed $\lambda)$ and the positivity of $a$, $b$ and $d$.

	\item Differentiating the previous equation we obtain
	\begin{displaymath}
	    \bigg[b(\lambda)+\frac{2ad}{\big(1+dx(\lambda)\big)^3}\bigg]x'(\lambda)
	    =b(\lambda)\cdot\tilde x'(\lambda)
	    -b'(\lambda)\cdot\big(x(\lambda)-\tilde x(\lambda)\big)\leq 0,
	\end{displaymath}
	\looseness=-1
	where the inequality uses that $b>0$, $b'>0$, $\tilde x'<0$ and $x\geq
	\tilde x$.
	We conclude using that the multiplicative factor on the left-hand side is positive since $b>0$, $a\geq 0$ and $x>-1/d$.
    \item The function $x$ and $\tilde x$ are decreasing so they converge to limits
	$l$ and $\tilde l$ respectively, with $\tilde l$ possibly equal to
	$-\infty$. Since $x>-1/d$, we have $l\geq -1/d$, and by item (2) we have $l\geq
	\tilde l$.  Either $l=\tilde l$ and we are done, or $\tilde l < l$, but then the left-hand side in \eqref{eq:x-foo}
			converges to $+\infty$ since $b(\lambda)\to+\infty$, hence the
			right-hand side must also converge to $+\infty$, which implies that
			$l=-1/d$. In this case, we necessarily have $\tilde l < -1/d$.
\end{enumerate}
\end{proof}

\begin{proof}[Proof of \cref{prop:opt-revenue}]
    We follow the same steps as in the proof of \cref{prop:linear}. The revenue objective only differs from welfare in its dependency on $\sigma_{a_i\theta_j}$. As a result, only the third KKT stationary condition changes and becomes
\begin{displaymath}
	2\lambda(f-r)
	\Big[\frac{f-r}{f}\frac{\sigma_{a_i\theta_j}}{\sigma_{\theta_i}^2}-t\Big]
	+2\nu(1+r)
	\Big[\frac{1+r}{f}\frac{\sigma_{a_i\theta_j}}{\sigma_{\theta_i}^2}+t\Big]
	    = rt(\lambda+\nu) + \frac{r\mu_{a_i}^2}{t\sigma_{\theta_i}^2} \frac 1{(1+r\sigma_{a_i\theta_j}/ft\sigma_{\theta_i}^2)^2}.
\end{displaymath}
    The expressions for $\nu$ and $\sigma_{a_i\omega}$ as a function of $\lambda$ thus also remain the same as in the welfare case, and after substituting out $\nu$, the above equation becomes
\begin{displaymath}
	\frac{2\sigma_{a_i\theta_j}}{frt\sigma_{\theta_i}^2}
	\big[\lambda nf(f-r)-r(1+r)\big] - 2 -\frac{\lambda nf - r}{1+r}
	=\frac{\mu_{a_i}^2}{t^2\sigma_{\theta_i}^2}
\frac 1{(1+r\sigma_{a_i\theta_j}/ft\sigma_{\theta_i}^2)^2}
\end{displaymath}
In other words, $\sigma_{a_i\theta_j}/rt\sigma_{\theta_i}^2$ is a solution of
    an equation of the form \eqref{eq:main-inter} with
\begin{displaymath}
	b(\lambda) = \frac 2 {f}
	\big[\lambda nf(f-r)-r(1+r)\big],
	\quad
	c(\lambda) =\frac{\lambda nf + r+2}{1+r},
	\quad a=\frac{\mu_{a_i}^2}{t^2\sigma_{\theta_i}^2},
	\quad d = \frac{r^2}{f}.
\end{displaymath}
It is easy to verify that the assumptions of \cref{lemma:implicit-eq} are
verified, hence there is a unique solution
$x(\lambda)$ in the interval $(-f/r^2,+\infty)$. Furthermore $x(\lambda)$ is
decreasing with
\begin{displaymath}
	\lim_{\lambda\to\infty} x(\lambda) = \lim_{\lambda\to\infty}
	\frac{c(\lambda)}{b(\lambda)} = \frac f{2(1+r)(f-r)} =x_0.
\end{displaymath}
We can now conclude the proof in exactly the same way as for 
the proof of \cref{prop:linear}.
\end{proof}

\subsection{Comparison of Mechanisms}

\begin{lemma}\label{lem:monotone}
	Let $I\subseteq\R$ be a right-unbounded interval and let $(f_1, f_2,
	g)\in(\R^I)^3$ be continuous non-increasing functions vanishing at infinity
	such that $f_1\geq f_2$ pointwise. Consider $r\geq 0$ and define $\lambda_i
	\eqdef \inf\set{\lambda\in I\given f_i(\lambda)+g(\lambda)\leq r}$ for
	$i\in\set{1,2}$. Then, 
	\begin{displaymath}
		\lambda_1\geq \lambda_2,\quad
		f_1(\lambda_1)\geq f_2(\lambda_2)
		\quad\mathrm{and}\quad
		g(\lambda_1)\leq g(\lambda_2).
	\end{displaymath}
\end{lemma}

\begin{proof}
	Define $I_i\eqdef\set{\lambda\in I\given f_i(\lambda)+g(\lambda)\leq
	r}$ so that $\lambda_i = \inf I_i$ for $i\in\set{1,2}$.
	\begin{enumerate}
		\item It follows from our assumptions that $f_1(\lambda)+g(\lambda)\leq
			r$ implies $f_2(\lambda)+g(\lambda)\leq r$. Hence, $I_1\subseteq
			I_2$ from which $\lambda_2\leq \lambda_1$ immediately follows.
		\item The inequality $g(\lambda_1)\leq g(\lambda_2)$ follows from 1.\
			and the fact that $g$ is non-increasing.
		\item We distinguish two cases:
			\begin{itemize}
				\item If $\lambda_1=\lambda_2$, we have
			$f_1(\lambda_1)=f_1(\lambda_2)\geq f_2(\lambda_2)$ by our
			assumption that $f_1\geq f_2$.
				\item Otherwise, by 1.\ we have that $\lambda_1>\lambda_2$.
					This implies that $\lambda_2\notin I_1$, or equivalently,
			\begin{displaymath}
				f_1(\lambda_2)+g(\lambda_2)>r,
			\end{displaymath}
			and the continuity of $f_1+g$ then implies that
			$f_1(\lambda_1)+g(\lambda_1)=r$. Thus,
			\begin{displaymath}
				f_1(\lambda_1)+g(\lambda_2)
				\geq f_1(\lambda_1)+g(\lambda_1)
				=r\geq f_2(\lambda_2)+g(\lambda_2),
			\end{displaymath}
			where the first inequality is by 2.\ and the last inequality is by
			definition of $\lambda_2$. We obtain $f_1(\lambda_1)\geq
			f_2(\lambda_2)$ after subtracting $g(\lambda_2)$ on both
			sides.\qedhere
	\end{itemize}
	\end{enumerate}
\end{proof}

\begin{proof}[Proof of \cref{prop:comparison}]
1. This is an immediate consequence of \cref{prop:opt-welfare} and \eqref{eq:ne}.

    2. The monotonicity of $\sigma_{a_i\theta_j}^{\rm R}/rt\sigma_{\theta_i}^2$ and $\sigma_{a_i\omega}^{\rm R}/s\sigma_{\omega}^2$ follows from the structure of the optimal mechanism described in \cref{prop:opt-revenue} and \cref{lem:monotone}. Indeed, consider two values of $\mu_{a_i}^2$, namely $\mu_1$ and $\mu_2$ with $\mu_1\geq\mu_2$. Define $I=[\max\set{0,r/(f-r)},+\infty)$ and for each $\lambda\in I$, let $x_1(\lambda)$ and $x_2(\lambda)$ be the solutions to \eqref{eq:cubic} in $(-f/r^2,+\infty)$ for $\mu_{a_i}^2=\mu_1$ and $\mu_{a_i}^2=\mu_2$ respectively. The existence of $x_1$ and $x_2$ is guaranteed by \cref{lemma:implicit-eq} which also implies that $x_1$ and $x_2$ are decreasing and converge to $x_0$ at $+\infty$. Furthermore, the same lemma (item 2) implies that $x_1\geq x_2$ pointwise. We can thus apply \cref{lem:monotone} with
    \begin{displaymath}
	f_1(\lambda) = a\big(x_1(\lambda) - x_0\big)^2,\; f_2(\lambda) = a\big(x_2(\lambda)-x_0\big)^2,\;\text{and}\; g(\lambda)= b\big(y(\lambda)-y_0\big)^2,
    \end{displaymath}
    where $a$, $b$, and $r$ are the (positive) coefficients defined such that the equation of the obedience ellipse \eqref{eq:ellipse} can be written as
	$a(x-x_0)^2+b(y-y_0)^2\leq r$.
    We then conclude that with $\lambda_1$ and $\lambda_2$ defined as in the statement of \cref{lem:monotone}, $f_1(\lambda_1)\geq f_2(\lambda_2)$ and $g(\lambda_1)\leq g(\lambda_2)$. But since $x_i\geq x_0$ pointwise for $i\in\set{1,2}$ and likewise $y\geq y_0$, this in turns imply that
    \begin{displaymath}
	x_1(\lambda_1)\geq x_2(\lambda_2)\quad\text{and}\quad
	y(\lambda_1)\leq y(\lambda_2).
    \end{displaymath}
    This concludes the proof of monotonicity since by \cref{prop:opt-revenue} $x_k(\lambda_k)=\sigma_{a_i\theta_j}^{\rm R}(\mu_k)/rt\sigma_{\theta_i}^2$ and $y(\lambda_k)=\sigma_{a_i\omega}^{\rm R}(\mu_k)/s\sigma_{\omega}^2$ for $k\in\set{1,2}$. Finally, the monotonicity of $\sigma_{a_i\theta_i}^{\rm R}/t\sigma_{\theta_i}^2$ immediately follows from the one of $\sigma_{a_i\theta_j}^{\rm R}/rt\sigma_{\theta_i}^2$ and the obedience equality constraint \eqref{eq:obedience-var-sym}.

    3. \& 4. We start by comparing the revenue-optimal mechanism (R-opt) for $\mu_{a_i}=0$ to the complete-information Nash equilibrium (CI). Let $x^{\rm R}$ and $y^{\rm R}$ be the functions defined in \cref{prop:opt-revenue} when $\mu_{a_i}=0$, that is,
    \begin{displaymath}
	x^{\rm R}(\lambda) = \frac{f(\lambda nf + 2 +r)}{2(1+r)[\lambda nf(f-r) -r(1+r)]}
	\quad\text{and}\quad
	y^{\rm R}(\lambda) = \frac{1}{2n}\frac{\lambda nf +1}{\lambda(f-r)-r}
    \end{displaymath}
    and let us compute $\lambda^{\rm R}$ solution to
    \begin{displaymath}
	x^{\rm R}(\lambda) = x^{\rm CI} \eqdef \frac{\sigma_{a_i\theta_j}^{\rm CI}}{rt\sigma_{\theta_i}^2}
	= \frac{f}{(f-r)(1+r)},
	\quad\text{namely,}\quad
	\lambda^{\rm R} = \frac{f(2+r)+r^2}{nf(f-r)}.
    \end{displaymath}
    We easily verify that $\lambda^{\rm R}$ satisfies dual feasibility $(\lambda^{\rm R}\geq\lambda_{\rm min})$ when $r<0$. And for $r>0$, this is equivalent to the inequality $r^2-r+2f\geq 0$, which is true since $r< f$.
    Then, observe that the inequality $y^{\rm R}(\lambda^{\rm R})\leq y^{\rm CI}$ is equivalent to $(f-r)(1-r)\geq 0$, and is satisfied for $r\in(-1, f)$. This shows that the point $\big(x^{\rm CI}, y^{\rm R}(\lambda^{\rm R})\big)$ belongs to the stationarity curve defining R-opt  and lies “below” the CI. This implies in particular that at the origin of the stationarity curve we always have $x^{\rm R}(\lambda_{\rm min})\geq x^{\rm CI}$. We thus distinguish two cases:
    \begin{itemize}
	\item The stationarity curve originates outside \ref{eq:ellipse}. In this case, R-opt is obtained at the intersection of the stationarity curve and \ref{eq:ellipse}. But by the above argument,  $(x^{\rm CI}, y^{\rm R}(\lambda^{\rm R}))$ lies on the stationarity curve below the CI (and in particular inside \ref{eq:ellipse}). So the sought for intersection must satisfy $x^{\rm R}(\lambda^\star)\geq x^{\rm CI}$ and by consequence $y^{\rm R}(\lambda^\star)\leq y^{\rm CI}$.
	\item The stationarity curve originates inside \ref{eq:ellipse}. Since $x^{\rm R}(\lambda_{\rm min})\geq x^{\rm CI}$, we necessarily have $y^{\rm R}(\lambda_{\rm min})\leq y^{\rm CI}$. But in this case, $\lambda_{\rm min}=\lambda^\star$ and we immediately obtain the desired comparison with the CI.
    \end{itemize}


Let us now compare R-opt to F-opt. Let $x^{\rm F}$ be the function defined in \cref{prop:opt-welfare}. Since $r+1\geq 0$, we immediately see that $x^{\rm F}\geq x^{\rm R}$ pointwise. We thus obtain the desired comparison using \cref{lem:monotone} and the same argument as in the proof of 2.

    Finally, the solution to \eqref{eq:cubic} converges to $+\infty$ when $\mu_{a_i}^2\to\infty$. This implies that when $\mu_{a_i}^2\to+\infty$, R-opt converges to the endpoint of the major axis satisfying $x\geq x_0$. Since the stationarity curve of W-opt is contained in the quadrant $\set{x,y\given x> x_0, y> y_0}$, we obtain the comparison between W-opt and  R-opt as $\mu_{a_i}^2=+\infty$.
\end{proof}

\subsection{Consumer-Optimal Mechanisms}\label{sec:app-consumer}

\begin{proof}[Proof of \cref{prop:bertrand-opt}]
The consumer's primitive utility consistent with the demand curve $Q_i(p;\theta_i)=\omega+r\sum_{j\neq i}p_j-p_i/2+\theta_i$ of \cref{example:bertrand} is
\begin{equation}\label{eq:cons-gen}
U_C(q;\omega,\theta) = \frac{1}{2}(q-\omega 1_n-\theta)^{\top} M^{-1} (q-\omega 1_n-\theta) + C(\omega,\theta),
\end{equation}
where $M\eqdef -\frac{1}{2}I_n + r(J_n-I_n)$ and $C(\omega,\theta)$ is a $q$-independent constant.\footnote{For $r\geq 1/[2(n-1)]$, $M$ ceases to be negative definite; on this range $U_C$ should be read as a formal quadratic utility index rather than a representative-consumer microfoundation. The mechanism characterization of \cref{prop:bertrand-opt}, which only uses the linear-objective machinery of \cref{prop:linear} with $(\alpha,\beta,\gamma)=(-1,-3,-3)$, holds throughout $r\in(0,1/(n-1))$.} Indeed, for this utility, the first-order condition $\nabla_q U_C=p$ gives
\begin{equation}\label{eq:d}
q^*(p;\omega,\theta) = M p + \omega 1_n + \theta,
\end{equation}
which reproduces the demand of \cref{example:bertrand}. The consumer's net surplus at posted prices $p$ is
\begin{align*}
W_C(p;\omega,\theta) &= U_C(q^*;\omega,\theta) - \ip{p,q^*}\\
&= -\frac{1}{2}\,p^{\top} M p - \omega\sum_i p_i - \sum_i \theta_i p_i + C(\omega,\theta).
\end{align*}
Expanding $-\frac{1}{2}p^{\top}M p = \frac{1}{4}\sum_i p_i^2 - \frac{r}{2}\sum_{i\neq j}p_ip_j$ yields \eqref{eq:bertrand-utility}.

Taking ex-ante expectation under symmetry, identifying $p_i=a_i$, and dropping the constant $\E[C(\omega,\theta)]$, the consumer's per-firm welfare is proportional to $\sigma_{a_i}^2 - 2(n-1)r\sigma_{a_ia_j} - 4s\sigma_{a_i\omega} - 4t\sigma_{a_i\theta_i}$. Substituting the obedience identities \eqref{eq:obedience-var-sym} for $\sigma_{a_i}^2$ and $\sigma_{a_i\theta_i}$, and dropping the constant $-3t\sigma_{\theta_i}^2$, the reduced objective is $-(n-1)r\sigma_{a_ia_j} - 3s\sigma_{a_i\omega} - 3(n-1)rt\sigma_{a_i\theta_j}$. This matches the linear-objective form of \cref{prop:linear} with $(\alpha,\beta,\gamma)=(-1,-3,-3)$.

    \textit{IC binds.} For $(\alpha,\beta,\gamma)=(-1,-3,-3)$ and $r\in(0,\frac 1{n-1})$, \cref{prop:linear} gives $\lambda_{\min}=0$ and
\begin{displaymath}
x(\lambda) = \frac{f}{2(1+r)}\cdot\frac{\lambda nf - 4r - 5}{\lambda nf(f-r) + (1+r)r},
\quad
    y(\lambda) = \frac 1{2n} \frac{\lambda nf -3-2r}{\lambda(f-r)+r}.
\end{displaymath}
    The monotonicity of $x$ is given by the sign of the determinant $\Delta=r(1+r)+(4r+5)(f-r)$ which is positive for $r\in(-1,f)$, so $x(\lambda)$ is strictly increasing from $x(0) = -f(4r+5)/(2(1+r)^2 r) < 0$ to $x_0>0$ as $\lambda\to\infty$, crossing zero at $\lambda_0\eqdef (4r+5)/(nf)$. Similarly the monotonicity of $y$ is given by the sign of $(1+r)(3f-2r)$ which is positive for $r\in(-1,f)$.

The point $\big(x(0),y(0)\big)$ has both coordinates negative ($y(0) = -(3+2r)/(2nr)<0$); since the ellipse \eqref{eq:ellipse} passes through $\varnothing=(0,0)$ and is centered at $(x_0,y_0)$ with $x_0,y_0>0$, $\big(x(0),y(0)\big)$ lies strictly outside $\mathcal{E}$.
    At $\lambda=\lambda_0$, $\big(x(\lambda_0),y(\lambda_0)\big)=\big(0,y(\lambda_0)\big)$ with $y(\lambda_0) = (r+1)f/[(4r+5)(f-r) + nfr]>0$,
placing $\big(0,y(\lambda_0)\big)$ strictly inside the chord $[\varnothing,\mathrm{SO}]$.

By continuity, the stationarity curve $\big(x(\lambda),y(\lambda)\big)$ crosses $\partial\mathcal{E}$ at some $\lambda^\star\in(0,\lambda_0)$; hence $x(\lambda^\star)<0$, the unconstrained obedient optimum has $\sigma_{a_i\theta_j}^\star<0$, the IC monotonicity $\sigma_{a_i\theta_j}\geq 0$ of \cref{prop:ic-char} binds, and the IC-constrained optimum lies on the segment $[\varnothing,\mathrm{SO}]$ of \cref{fig:ic}.

The result is implementable without discriminatory transfers (\cref{prop:ic-char}).

\textit{Location on the segment.} Since the stationarity condition for $y$ is unaffected by the IC constraint $x\geq 0$, the IC-constrained optimum lies on the curve $\big(0,y(\lambda)\big)$. Since $y(0)<0$, the curve starts outside $\mathcal{E}$ and, since $y(\lambda)$ is increasing, first enters $\mathcal{E}$ at $y=0$, i.e., at $\varnothing=(0,0)\in\partial\mathcal{E}$. Hence the consumer-optimal mechanism is at $\varnothing$: $\sigma_{a_i\omega}=0$, $\sigma_{a_i\theta_j}=0$, and the recommendation reduces to the prior Bayes--Nash equilibrium $a_i=\mu_{a_i}+t(\theta_i-\mu_{\theta_i})$.
\end{proof}

\begin{proposition}\label{prop:pareto-opt}
    For $r\in(0,\frac 1 {n-1})$ and $\eta\in[0,1]$, consider the linear objective
\begin{displaymath}
    \tilde F_\eta = (n-1)r(1-2\eta)\,\sigma_{a_ia_j} \,+\, (1-4\eta)\,s\sigma_{a_i\omega} \,+\, (n-1)rt\,(1-4\eta)\,\sigma_{a_i\theta_j},
\end{displaymath}
    interpolating between the producer expected surplus $(\eta=0)$ and the consumer's expected surplus $(\eta=1)$. Then the symmetric, incentive compatible mechanism maximizing $\tilde F_\eta$ satisfies:
    \begin{enumerate}
	\item for $\eta<\frac 1 2$, it is a boundary solution with $\sigma_{a_i\theta_j}>0$, hence the IC monotonicity condition is slack;
	\item for $\eta\geq \frac 1 2$, it collapses to the consumer-optimal mechanism, which is also the prior Bayes--Nash equilibrium for which $\sigma_{a_i\theta_j}=\sigma_{a_i\omega}=0$ and where the IC condition binds.
    \end{enumerate}
\end{proposition}

\begin{proof}
    This objective is the special case of the general linear objective in \cref{prop:linear} with $\alpha=1-2\eta$ and $\beta=\gamma=1-4\eta=2\alpha-1$. In particular, the structure of the optimal mechanism is governed by the stationarity curve $\big(x(\lambda),y(\lambda)\big)$ with the expressions for $x$ and $y$ given in \cref{prop:linear}.
Since $\beta=\gamma$, we have
	\begin{displaymath}
	    y(\lambda) = \frac 1 {2n} \frac{N(\lambda)-2\alpha(r+1)}{\lambda(f-r)-\alpha r},
	\end{displaymath}
	where $N(\lambda)$ is the numerator of $x(\lambda)$ as in \cref{prop:ic-slack-regime}.

    We start with the case $\eta<1/2$, or equivalently $\alpha>0$. In this case $\lambda_{\rm min}=\frac{\alpha r}{f-r}$ and $|y(\lambda)|\to\infty$ when $\lambda\to\lambda_{\rm min}^+$. Thus, the stationarity curve originates outside the obedience ellipse \ref{eq:ellipse} and we have a boundary solution. Moreover, by \cref{prop:ic-slack-regime}:
    \begin{itemize}
	\item if $N(\lambda_{\rm min})\geq 0$, the stationarity curve originates in the half-plane $x\geq 0$, thus the boundary solution satisfies $\sigma_{a_i\theta_j}>0$ and the IC monotonicity condition is slack;
	\item if $N(\lambda_{\rm min})<0$, the stationarity curve originates in the half-plane $x<0$ and crosses the vertical axis $x=0$ for $\lambda=\lambda_0>\lambda_{\rm min}$ such that $N(\lambda_0)=0$. At $\lambda_0$, we have
	    \begin{equation}\label{eq:foo-y}
		y(\lambda_0) = \frac  1{2n} \frac{-2\alpha(r+1)}{\lambda_0(f-r)-\alpha r}.
	    \end{equation}
	    The denominator is positive since $\lambda_0>\lambda_{\rm min}$ and the numerator is negative since $\alpha>0$. Hence $(0,y(\lambda_0))$ lies strictly below $\varnothing=(0,0)$ on the axis $x=0$, hence outside $\mathcal{E}$. Thus, the stationarity curve intersects $\mathcal{E}$ for some $\lambda>\lambda_0$ at which $x(\lambda)>0$. The IC monotonicity condition is therefore slack at this intersection.
    \end{itemize}

    The case $\eta=\frac{1}{2}$ ($\alpha=0$) follows by continuity: as $\alpha\to 0^+$, the crossing point $(0,y(\lambda_0))$ satisfies $y(\lambda_0)\to 0$, so the optimal intersection on $\partial\mathcal{E}$ converges to $\varnothing$, and the IC condition $x\geq 0$ binds in the limit.

    We now consider $\eta> 1/2$, or equivalently $\alpha< 0$. In this case, $\lambda_{\rm min}=0$ and 
    \begin{displaymath}
	N(\lambda_{\rm min}) = N(0) = \alpha(r+2) +\gamma(r+1).
    \end{displaymath}
    Since $\gamma=2\alpha-1<0$ and $r\geq -1$, we have $N(0)<0$. Thus, as above, the stationarity curve originates in the half-plane $x<0$ and intersects the $x=0$ axis for $\lambda_0>0$ such that $N(\lambda_0)=0$. The expression for $y(\lambda_0)$ obtained in \eqref{eq:foo-y} remains valid, but this time $y(\lambda_0)>0$ since $\alpha<0$, and $y(\lambda_0)<y_0<2y_0$ by monotonicity, so $\big(0,y(\lambda_0)\big)$ lies inside $\mathcal{E}$. The obedient optimal mechanism is obtained for some $\lambda<\lambda_0$ at which $x(\lambda)<0$, hence the IC condition is violated and must bind at the optimum. Since the stationarity condition for $y$ is unaffected by the IC constraint $x\geq 0$, the IC-constrained optimal mechanism lies on the curve $\big(0,y(\lambda)\big)$. Indeed,
    \begin{displaymath}
	y(0) = \frac 1{2n} \frac{N(0)-2\alpha(r+1)}{-\alpha r}
	= \frac 1{2n}\frac{\alpha(r+2)-(r+1)}{-\alpha r}<0,
    \end{displaymath}
    where the inequality uses $\alpha<0$ and $r>0$. Hence the curve $\big(0,y(\lambda)\big)$ originates outside $\mathcal{E}$ and, since $y(\lambda)$ is increasing, first enters $\mathcal{E}$ at $y=0$, i.e., at $\varnothing=(0,0)\in\partial\mathcal{E}$. The IC-constrained optimal mechanism is therefore $\varnothing$.
\end{proof}

\newpage

\bibliographystyle{ecta-fullname}
\bibliography{main}
\newpage
\begin{center}
    \huge Supplemental Appendix
\end{center}

\section{Endogenous Participation}\label{sec:participation}
\subsection{Mechanism Design}

We now extend our model to the case of a pure information intermediary, i.e., we allow the players to play in the downstream game regardless of whether they participate in the designer's mechanism. Thus, a mechanism must now specify a distribution of action recommendations for each subset of participating players. Such a mechanism consists of two functions. 
\begin{itemize}
	\item An information policy that maps each participating player's type report and state of nature to a distribution over action profiles. Formally, for 
    each subset $S\subseteq[n]$ of participating players, the designer commits to an information policy $\tau^S:\R^S\times \R\to\Delta(\R^S)$ that maps a
		vector of type reports $\theta\in\R^S$ and a state of nature $\omega\in
		\R$ to a distribution $\tau^S(\theta,\omega)$ over action recommendations for the participating players.
	
    \item For each participating player $i\in S$, a function $p_i:\R\to\R$ mapping their reported type to the payment they are being charged in exchange for an action recommendation.
	A non-participating player pays no price and receives no information.
\end{itemize}

In this section, we focus on Perfect Bayesian Equilibria where every player participates in the designer's mechanism. Because we have restricted attention to mechanisms that are  Gaussian  when all players participate, the designer is entirely unconstrained in their choice of information policy off the equilibrium path, when a player does not. In particular, because players do not observe each other's participation decisions and every action $a_i\in \mathbb{R}$ recommendation is in the support of the equilibrium distribution, the designer can perfectly control the actions of the participating players when one player deviates. Intuitively, the designer will do so in order to maximally punish any nonparticipating players.

\subsection{Participation Constraints}

When a player does not  participate in the mechanism, their utility is still affected by the actions of the remaining players in the downstream game. In order to complete the description of a mechanism, we must then  specify the \emph{off-path} action recommendations---those  used when a player chooses not to participate. This determines the reservation utility of non-participating players and  constrains the payments that can be collected from participating players.

In this section, we characterize the set of payments and off-path action recommendations that incentivize full participation.
Though the complete description of a mechanism is a priori exponential in the numbers of players, the following lemma shows that, in order to  characterize full-participation equilibria, the choice of an off-path mechanism can be parametrized by a single number: the mean aggregate action that a single non-participating player faces.

\begin{lemma}[Reservation utility]\label{lemma:reserv}
    For $i\in[n]$ and $j\neq i$, let $a_j^i$ be a random variable distributed according to $\tau^{[n]\setminus\set{i}}(\theta_{-i},\omega)$ and define $M_i\eqdef \sum_{j\neq i}\E[a_j^i]$, the \emph{mean aggregate action} faced by non-participating player $i$. Then, type $\theta_i$'s optimal action and  reservation utility are given by
    \begin{displaymath}
	a_i^o(\theta_i)=rM_i+s\mu_\omega+t\theta_i\quad \text{and}\quad u_i^o(\theta_i) = \frac 1 2 a_i^o(\theta_i)^2.
    \end{displaymath}
\end{lemma}

From the perspective of the non-participating player $i\in[n]$, all off-path  mechanisms that  induce a given mean aggregate action $M_i$ are equivalent to one in which the actions of all participating players are constant and equal to $M_i/(n-1)$. 
Furthermore, recall that we do not require the off-path recommendations $\set{a_j^i}_{j\neq i}$ to be obedient for the participating players, i.e., the designer is not bound to choose $M_i/(n-1)=\mu_{a_i}$ as in \cref{prop:obedience-char-sym}.  
In fact, revenue extraction will require us to consider different—and generically non-obedient—choices of $M_i$.

Our first result establishes conditions under which a Gaussian recommendation rule can be supported by on-path payments and off-path recommendations that induce \emph{truthful participation} by all players.

\begin{proposition}\label{prop:max-payment}
    Consider a Gaussian recommendation rule $(\mu, \K)$ and a player $i\in [n]$.
    \begin{enumerate}
        \item     There exist a payment function $p_i$ and mean off-path aggregate action $M_i$ that incentivize truth-telling and participation of player $i$ iff $t\sigma_{a_i\theta_i}\geq t^2\sigma_{\theta_i}^2$.
    
   \item For such a recommendation rule and $r\neq 0$, the pair $(p_i, M_i)$ that maximizes player $i$'s interim payment $p_i(\theta_i)$ pointwise is given by
    \begin{equation}\label{eq:part-mech}
	p_i(\theta_i) = \tu_i(\theta_i)
	- \frac 1 2 \frac{t\sigma_{\theta_i}^2}{\sigma_{a_i\theta_i}} \E[a_i\given \theta_i]^2
	\quad\text{and}\quad
	M_i= \frac{\mu_{a_i}t\sigma_{\theta_i}^2}{r\sigma_{a_i\theta_i}}
	-\frac{s\mu_\omega+t\mu_{\theta_i}}{r}.
    \end{equation}
At $r=0$, the reservation utility $u_i^o$ is independent of $M_i$ (\cref{lemma:reserv}), and the same payment formula maximizes $p_i(\theta_i)$ for any choice of $M_i$.
	\item If the recommendation rule $(\mu, \K)$ is also obedient, then the maximal payment $p_i(\theta_i)$ is nonnegative for all $\theta_i$ and can be written
	    \begin{displaymath}
		p_i(\theta_i)
		= \frac 1 2\bigg(1-\frac{t\sigma_{\theta_i}^2}{\sigma_{a_i\theta_i}}\bigg)\E[a_i\given \theta_i]^2
		+ \frac 1 2\var(a_i\given\theta_i)
		= \frac 1 2\bigg(1-\frac{t\sigma_{\theta_i}^2}{\sigma_{a_i\theta_i}}\bigg)\E[a_i\given \theta_i]^2
		+ \frac 1 2\bigg(\sigma_{a_i}^2-\frac{\sigma_{a_i\theta_i}^2}{\sigma_{\theta_i}^2}\bigg).
	    \end{displaymath}
    \end{enumerate}
\end{proposition}


Part (1) of \cref{prop:max-payment} shows that the same covariance condition, $t\sigma_{a_i\theta_i}\geq t^2\sigma_{\theta_i}^2$, that characterizes  incentive compatibility (\cref{prop:ic-char}) is also necessary for participation. Thus, for a mechanism that incentivizes participation, requiring incentive compatibility is equivalent to imposing truthfulness and obedience separately. In other words, if a mechanism  satisfies obedience, any transfer that guarantees participation and induces truth-telling also prevents double deviations.

That individual rationality and obedience (with the right transfers) suffice to discourage double deviations is a bit surprising but ex post intuitive: each player's participation decision  involves a specific double deviation, in which the player  deviates at the first stage (by choosing not to participate instead of reporting their type) and then reacts optimally to the actions recommended to the other players, on the basis of the information they have at that point.

Part (2) elucidates the role of the mean off-path aggregate action $M_i$. When designing the actions of the $n-1$ participating players, the information seller performs a horizontal shift of a deviating player's action and reservation utility. Thus, the choice of $M_i$  controls the type $\theta^o_{\star}$ that takes a nil action $a_i^o=0$ and hence obtains a reservation utility $u_i^o$ of zero. The seller-optimal $M_i$ is then obtained when the minima of the reservation and (gross) on-path utilities are aligned. In this case, the same critical type $\theta_{\star}$ takes a nil (expected) action whether or not they participate. The seller can then maximize  revenue by adding a constant to the payment function $p$ until the critical type's net utility when participating reaches zero. This is illustrated in \cref{fig:participation}.

\begin{figure}
    \hfill{}
    \includegraphics[width=0.48\textwidth]{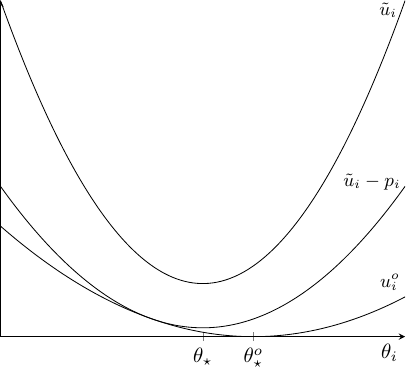}
    \hfill{}
    \includegraphics[width=0.48\textwidth]{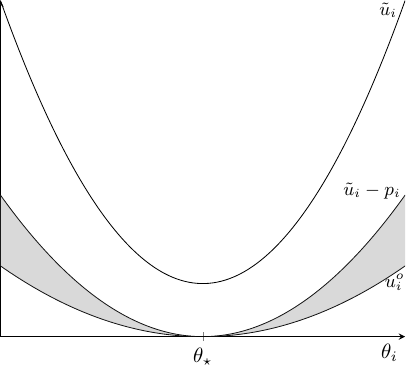}
    \hfill{}
\caption{Reservation utility $u_i^o$, utility $\tilde u_i$ and net utility $\tilde u_i - p_i$ of player $i$ as a function of $\theta_i$ for an incentive compatible and individually rational mechanism. The utility and net utility are minimal for $\theta_\star=\mu_{\theta_i}-\mu_{a_i}\sigma_{\theta_i}^2/\sigma_{a_i\theta_i}$, and the reservation utility reaches its minimum, zero, at $\theta_\star^o=-(rM_i+s\mu_\omega)/t$. By choosing the mean off-path aggregate action $M_i$, the information seller controls $\theta_\star^o$, which in turns constrains the maximum collectible payment. The left plot shows a generic, suboptimal choice of $M_i$, while the right plot shows the optimal setting of $M_i$, for which $\theta_\star^o=\theta_\star$. The grayed area shows the loss in revenue due to the elicitation constraint.}
\label{fig:participation}
\end{figure}

Thus, in order to maximize revenue, the off-path mechanism needs to be tuned with respect to the chosen on-path recommendation rule. In other words, there is no off-path recommendation rule that maximally penalizes non-participation \emph{uniformly} across all on-path recommendation rules, and the on- and off-path mechanisms must thus be designed in concert. This goes against the intuition that there should exist a uniform maximal threat (e.g.\ by asking the remaining players to flood the market in the Cournot example).\footnote{In particular, this contrasts with the findings of \cite{BDHN22} for a binary state-matching game, in which such a threat exists: disclosing all the information to all participating players.}

Finally, all our findings in this Section can be combined into the following corollary, which formalizes how our results on endogenous participation do not affect the optimal \emph{obedient} mechanisms in \cref{sec:opt}. In particular, for any incentive compatible mechanism, we can find an off-path mechanism which allows the designer to charge the same transfers.

\begin{corollary}
    Consider an obedient, symmetric Gaussian recommendation rule $(\mu, \K)$. If $rt\sigma_{a_i\theta_j}\geq 0$, there exist a  payment function $p_i\geq0$ and off-path action $M_i$ such that the mechanism $(\mu,\K,p,M)$ is individually rational and incentive compatible.
\end{corollary}

\subsection{Proofs}

\begin{proof}[Proof of \cref{lemma:reserv}]
    Using the notations of the lemma, let us write the utility of a non-participating player $i\in[n]$ with type $\theta_i$ and taking action $a_i$ in the downstream game:
    \begin{align*}
	\E[u_i(a_i,a_{-i}^i;\theta,\omega)\given \theta_i]
	&=-\frac 1 2 a_i^2 + ra_i\sum_{j\neq i} \E[a_j^i\given\theta_i] +
	s a_i\E[\omega\given\theta_i] + t a_i\theta_i\\
	&=-\frac 1 2 a_i^2 + r a_i M_i + s a_i\mu_\omega +ta_i\theta_i,
    \end{align*}
    where the second equality uses that $\omega, a_j^i\perp\theta_i$ due to $a_j^i$ being $(\theta_{-i},\omega)$-measurable. Consequently, the optimal action of player $i$ in the downstream game is
    \begin{displaymath}
	a_i^o(\theta_i)=rM_i+s\mu_\omega+t\theta_i,
    \end{displaymath}
    for which we obtain
    \begin{displaymath}
	u_i^o(\theta_i) = \frac 1 2 (rM_i+s\mu_\omega+t\theta_i)^2.\qedhere
    \end{displaymath}
\end{proof}

\begin{proof}[Proof of \cref{prop:max-payment}]
    1.\ \& 2. By the on-path FOC derived in the proof of \cref{prop:ic-char} in \cref{sec:app-tic}, any payment supporting truthful reporting on path must satisfy, for all $\theta_i\in\R$,
    \begin{equation}\label{eq:truth-der}
	\tu_i'(\theta_i)-p_i'(\theta_i) = t\E[a_i\given\theta_i].
    \end{equation}
    Recall that for a Gaussian mechanism we have $\E[a_i\given\theta_i] = \mu_{a_i} + \sigma_{a_i\theta_i}/\sigma_{\theta_i}^2(\theta_i-\mu_{\theta_i})$. Hence, by integrating \eqref{eq:truth-der} we obtain for all $\theta_i$
    \begin{displaymath}
	\tu_i(\theta_i)-p_i(\theta_i) =
	\frac{t\sigma_{\theta_i}^2}{2\sigma_{a_i\theta_i}}\E[a_i\given\theta_i]^2 + C,
    \end{displaymath}
    where $C$ is an integration constant. Using the expression for the reservation utility in \cref{lemma:reserv}, participation is equivalent to, for all $\theta_i\in\R$,
    \begin{equation}\label{eq:foo-par}
	\frac{t\sigma_{\theta_i}^2}{2\sigma_{a_i\theta_i}}\E[a_i\given\theta_i]^2 + C
	\geq \frac 1 2 (rM_i+s\mu_\omega+t\theta_i)^2.
    \end{equation}
    This is an inequality between two quadratics (since $\E[a_i\given\theta_i]$ is affine) for which a necessary condition is that the leading coefficient on the left be at least the one on the right. That is,
    \begin{displaymath}
	    \frac{t\sigma_{a_i\theta_i}}{2\sigma_{\theta_i}^2}\geq \frac{t^2} 2,
    \end{displaymath}
    and we thus obtain the condition in the proposition's statement.

    Conversely, assume that this condition is satisfied. Plugging in \eqref{eq:foo-par} the value of $\theta_i$ for which $\E[a_i\given\theta_i]=0$, we see that we must have $C\geq 0$. But for the value of $M_i$ given in \eqref{eq:part-mech}, the right hand side of \eqref{eq:foo-par} becomes
    \begin{displaymath}
	\frac 1 2
	\prn[\bigg]{ \frac{\mu_{a_i}t\sigma_{\theta_i}^2}{\sigma_{a_i\theta_i}}
	+t(\theta_i-\mu_{\theta_i})}^2
	= \frac 1 2
	\prn[\bigg]{\frac{t\sigma_{\theta_i}^2}{\sigma_{a_i\theta_i}}}^2
	\E[a_i\given\theta_i]^2.
    \end{displaymath}
    Hence \eqref{eq:foo-par} is satisfied with $C=0$, since $t\sigma_{\theta_i}^2/\sigma_{a_i\theta_i}\leq 1$.

    In summary, we proved that truthfulness and participation imply the covariance condition $t\sigma_{a_i\theta_i}\geq t^2\sigma_{\theta_i}^2$ and that the integration constant $C$ is non-negative. Conversely, when the covariance condition is satisfied, the on-path FOC of \cref{prop:ic-char} also holds and the payment from \eqref{eq:part-mech} satisfies \eqref{eq:truth-der}, hence truthfulness is satisfied. Furthermore, with the setting of $M_i$ from \eqref{eq:part-mech}, we saw that \eqref{eq:foo-par} is satisfied with $C=0$, hence the resulting mechanism incentivizes participation. Finally, since the payment is non-increasing in $C$, being able to achieve $C=0$ results in the maximal payment.

3. Recall the expression \eqref{eq:interim-utility-ob} for the interim utility of an obedient mechanism:
    \begin{displaymath}
	\tu_i(\theta_i) =
	\frac 1 2 \E[a_i\given\theta_i]^2 + \frac 1 2\var(a_i\given\theta_i).
    \end{displaymath}
    Moreover, for a Gaussian mechanism, the conditional variance $\var(a_i\given\theta_i)$ does not depend on $\theta_i$ and is given by
    \begin{displaymath}
	\var(a_i\given\theta_i)
	= \sigma_{a_i}^2 - \frac{\sigma_{a_i\theta_i}^2}{\sigma_{\theta_i}^2}.
    \end{displaymath}
    The stated formula for the payment follows by combining the previous expressions with \eqref{eq:part-mech}. The non-negativity of the payment then follows from the non-negativity of the conditional variance (Cauchy–Schwarz inequality) and the covariance constraint $\sigma_{a_i\theta_i}/t\sigma_{\theta_i}^2\geq 1$.
\end{proof}

\section{Contractable Actions}\label{sec:delegation}

\looseness=-1
We now consider a situation in which the designer can contract on the players' actions. This effectively relaxes the obedience constraint, but we still assume that the players strategically report their types, hence the recommendation rule needs to guarantee truthful participation.
Using \eqref{eq:collusive_a}, we see that the first-best action satisfies $t\sigma_{a_i\theta_i}^{\rm FB}\geq t^2\sigma_{\theta_i}^2$ whenever $r\in\big(-\frac 1 2,\frac 1{2(n-1)}\big)$. This implies the existence of (discriminatory) transfers that elicit the players' types and allow the coordinator to achieve the efficient outcome.

We now turn to revenue maximization. In this case, \cref{prop:max-payment} and \eqref{eq:utility} imply
\begin{align*}
    \E[p_i(\theta_i)] = & \Big ((n-1)r-\frac{1}{2} \Big)\mu_{a_i}^2
    +(t\mu_{\theta_i}+s\mu_{\omega})\mu_{a_i}
    -\frac{t\mu_{a_i}^2\sigma_{\theta_{i}}^2}{2\sigma_{a_i\theta_i}}\\
    &-\frac 1 2\sigma_{a_i}^2+(n-1)r\sigma_{a_i a_j}
    +\frac t 2\sigma_{a_i \theta_i}+s\sigma_{a_i \omega}\,.
\end{align*}
The total payment $\sum_i \E[p_i(\theta_i)]$ is symmetric and we can thus restrict to symmetric mechanisms without loss of generality. The constraints are the two inequalities characterizing positive semi-definiteness (\cref{lemma:psd}) and the strong monotonicity constraint from \cref{prop:max-payment} for participation: $t\sigma_{a_i\theta_i}\geq t^2\sigma_{\theta_i}^2$.

\begin{proposition}\label{prop:opt:TwO}
    For $r\in\big(-\frac 1 2,\frac 1 {2(n-1)}\big)$, there exists a unique symmetric revenue-maximizing mechanism that guarantees truthful participation. In this mechanism, recommendations are deterministic conditioned on $(\theta,\omega)$: they take the form \eqref{eq:params-sym} with $\delta=0$ and
    \begin{displaymath}
	\mu_{a_i} =
	\frac{f(t\mu_{\theta_i}+s\mu_{\omega})}
	{ft\sigma_{\theta_i}^2/\sigma_{a_i\theta_i}+f-2r},
	\;\; \sigma_{a_i\omega}=\frac {fs\sigma_\omega^2} {f-2r},
	\;\; \sigma_{a_i\theta_j} = \frac{2rf\sigma_{a_i\theta_i}}{f-2r(1-f)},
	\;\; \sigma_{a_i\theta_i} = \max\set{1,x^\star}t\sigma_{\theta_i}^2.
	\end{displaymath}
	where $x^\star$ is the unique real solution in $(-f/(f-2r),+\infty)$ to
	\begin{equation}\label{eq:main-bench}
	\frac{f^2(s\mu_\omega+t\mu_{\theta_i})^2}
	    {2t^2\sigma_{\theta_i}^2\big[(f-2r)x+f\big]^2}
	    =\frac{(f-2r)(1+2r)}{f-2r(1-f)} x - \frac 1 2.
	\end{equation}
\end{proposition}

\begin{proof}
    Since $\mu_{a_i}$ is unconstrained, we first optimize over it keeping the other variables fixed.  Note that the objective is concave in $\mu_{a_i}$ for $2r\in(-1,f)$, so we simply solve $\partial F/\partial \mu_{a_i} = 0$:
    \begin{align*}
        \mu_{a_i} =
	\frac{f(t\mu_{\theta_i}+s\mu_{\omega})}
	{f(1+t\sigma_{\theta_i}^2/\sigma_{a_i\theta_i})-2r}\,.
    \end{align*}
    For this optimal value of $\mu_{a_i}$ the objective function reduces (after multiplying by 2) to
    \begin{displaymath}
	\frac{f(t\mu_{\theta_{i}}+s\mu_{\omega})^2{\sigma_{a_i\theta_{i}}}/{t\sigma_{\theta_{i}}^2}}{2\left((f-2r)\sigma_{a_i\theta_{i}}/t\sigma_{\theta_i}^2+f\right)}\\
        -\frac 1 2\sigma_{a_i}^2
	+(n-1)r\sigma_{a_i a_j}
	+\frac t 2\sigma_{a_i \theta_i}+s\sigma_{a_i \omega}
    \end{displaymath}
    which is concave in the covariance parameters.

    Define $x \eqdef \sigma_{a_i\theta_i}/t\sigma_{\theta_i}^2$, we introduce non-negative multipliers $\lambda$ and $\nu$ for the two PSD constraints and $\xi$ for the constraint $x\geq 1$ that guarantees truthful participation. Taking the partial derivatives of the Lagrangian with respect to $\sigma_{a_i}^2$ and $\sigma_{a_ia_j}$ and equating them to zero allows us to solve for $\lambda$ and $\nu$
\begin{displaymath}
	\lambda = \frac{2r+1}{2nf}
	\quad\quad
	\nu=\frac{f-2r}{2nf}.
\end{displaymath}
    We see that $\lambda$ and $\nu$ are positive, hence the PSD constraints are binding. Thus the mechanism is deterministic conditioned on $(\theta,\omega)$ and we can compute $\sigma_{a_i}^2$ and $\sigma_{a_i a_j}$ from $\sigma_{a_i\theta_i}$, $\sigma_{a_i \theta_j}$, and $\sigma_{a_i \omega}$. The stationarity conditions with respect to $\sigma_{a_i\omega}$ and $\sigma_{a_i\theta_j}$ yield
\begin{displaymath}
	\frac{\sigma_{a_i\omega}}{s\sigma_\omega^2}=\frac 1 {2n\nu} = \frac f {f-2r}
	\quad\text{and}\quad
	\sigma_{a_i\theta_j} = \frac{2rf\sigma_{a_i\theta_i}}{f-2r(1-f)}.
\end{displaymath}

Finally, the expression for $\sigma_{a_i\theta_j}$ combined with the stationarity condition with respect to $\sigma_{a_i\theta_i}$ yields the following equation in the variable $x$:
\begin{equation}\label{eq:foo-bench}
	\frac{f^2(s\mu_\omega+t\mu_{\theta_i})^2}{2t^2\sigma_{\theta_i}^2\big[(f-2r)x+f\big]^2}
	=\frac{(f-2r)(1+2r)}{f-2r(1-f)} x - \frac 1 2-\xi.
\end{equation}
By \cref{lemma:implicit-eq}, when $\xi=0$, \eqref{eq:foo-bench} has a unique solution $x^\star$ in the interval $(x_\infty,+\infty)$ with $x_\infty\eqdef -f/(f-2r)$. Complementary slackness implies $\xi(x-1)=0$, thus:
\begin{itemize}
    \item either $x^*\geq 1$, in which case $x=x^\star$, $\xi=0$ are feasible variables satisfying \eqref{eq:foo-bench}.
    \item or $x^\star<1$. In this case a primal-dual feasible pair satisfying \eqref{eq:foo-bench} is obtained by setting $x=1$ and $\xi$ solution to \eqref{eq:foo-bench}.
\end{itemize}
The inequality determining whether the optimal $x$ is $1$ or $x^\star$ is
thus:
\begin{align*}
	\frac{(s\mu_\omega+t\mu_{\theta_i})^2}{2t^2\sigma_{\theta_i}^2}
	&< \frac{4(f-r)^2}{f^2}\left(\frac 1 2 - \frac{4r^2}{f-2r(1-f)}\right).\qedhere
\end{align*}
\end{proof}

Similar to the mechanism that achieves the efficient outcome, the mechanism of \cref{prop:opt:TwO} is always deterministic conditioned on $(\theta,\omega)$. Intuitively this is because relaxing obedience allows the recommendation to reveal more information about the fundamentals, even in the strong substitutes regime. Moreover, the mean action $\mu_{a_i}$ in \cref{prop:opt:TwO} is, in absolute value, smaller than the first-best action: $\abs{\mu_{a_i}}< \abs{\mu_{a_i}^{\rm FB}}$.

Contrary to the first-best actions, we see that the strong monotonicity condition ($t\sigma_{a_i\theta_i}\geq t^2\sigma_{\theta_i}^2$) can bind in \cref{prop:opt:TwO}. In other words, guaranteeing truthful participation sometimes requires the designer to distort recommendations. The last inequality in the proof of \cref{prop:opt:TwO} shows that this happens when $\sigma_{\theta_i}^2$ is large relative to $(s\mu_\omega+t\mu_{\theta_i})^2$.

In the absence of the truthfulness constraint, the designer would instead maximize revenue by maximizing the players' welfare, minimizing the players' reservation utility, and extracting the entire surplus. The optimal mechanism thus recommends the first-best actions and minimizes $\E[u_i^o(\theta_i)]$ by choosing off-path mean aggregate action $M_i = -(t\mu_{\theta_i}+s\mu_{\omega})/r$, resulting in $\E[u_i^o(\theta_i)] = t^2\sigma_{\theta_i}^2/2$.

\section{Cournot Competition}\label{app-consumer}

We now revisit the consumer-optimal mechanism, this time through the lens of a Cournot competition (\cref{example:cournot}). As in the proof of \cref{prop:bertrand-opt}, a simple derivation shows that for a linear-quadratic valuation inducing the demand curve of \cref{example:cournot}, the consumer surplus for a vector of good quantities $q\in\R^n$ at the market-clearing price is given by
\begin{equation}\label{eq:cournot}
    W_C(q) =\frac 1 4\sum_{i=1}^n{q_i^2} -\frac r 2\sum_{i\neq j} q_iq_j.
\end{equation}
Since this objective function is symmetric, we can restrict to symmetric mechanisms without loss of generality. Maximizing the consumer's expected surplus is thus equivalent to maximizing
\begin{displaymath}
    \frac 4 n \E[W_C(q)] = \sigma_{a_i}^2 - 2(n-1)r\sigma_{a_ia_j}.
\end{displaymath}
Using the obedience constraints to rewrite this objective in terms of $\sigma_{a_i\theta_j}$, $\sigma_{a_ia_j}$ and $\sigma_{a_i\omega}$, the resulting optimization problem takes the form \eqref{eq:reduced} with
\begin{displaymath}
	\tilde F(\sigma_{a_ia_j}, \sigma_{a_i\omega},\sigma_{a_i\theta_j})
   =-(n-1)r\sigma_{a_ia_j} +s\sigma_{a_i\omega} + (n-1)rt\sigma_{a_i\theta_j}.
\end{displaymath}

\begin{proposition}\label{prop:consumer-opt}
    For each $r$ with $2r\in(-1, \frac 1 {n-1})$, there exists a unique symmetric mechanism maximizing the consumer's expected surplus subject to incentive compatibility. Furthermore,
    \begin{enumerate}
	\item The action recommendations take the form \eqref{eq:params-sym} with noise correlation $\rho=1$ when $r<0$ and $\rho=-\frac 1 {n-1}$ when $r>0$.
	\item For $r>0$, the optimal values for $\sigma_{a_i\theta_j}$ and $\sigma_{a_i\omega}$ are given by
	    \begin{displaymath}
		\sigma_{a_i\theta_j} = 0,
		\quad
		\sigma_{a_i\omega} = s\sigma_{\omega}^2\cdot\min\set*{\frac{1+2r}{2nr}, \frac f{f-r}}.
	    \end{displaymath}
	\item For $r<0$, the optimal values for $\sigma_{a_i\theta_j}$ and $\sigma_{a_i\omega}$ are given by
    \begin{gather*}
	\sigma_{a_i\theta_j} = x(\lambda^\star)rt\sigma_{\theta_i}^2,
	\quad
	\sigma_{a_i\omega} = y(\lambda^\star)s\sigma_\omega^2,\\
	\text{with}\quad
	x(\lambda)=
	\frac{f}{2(1+r)}\cdot \frac{\lambda nf -1}{\lambda nf(f-r) +(1+r)r},
	\quad y(\lambda) = \frac{1}{2n}\cdot\frac{\lambda nf+1+2r}{\lambda(f-r)+r},
    \end{gather*}
	    and $\lambda^\star$ the unique scalar in $(-\frac r {f-r},+\infty)$ such that $\big(x(\lambda^*), y(\lambda^*)\big)$ lies on \ref{eq:ellipse}'s boundary.
    \end{enumerate}
\end{proposition}

\begin{proof}[Proof of \cref{prop:consumer-opt}]
    As in the proof of \cref{prop:bertrand-opt}, we start from a consumer's symmetric linear-quadratic valuation  in goods' quantities \eqref{eq:cons-gen}. When the firms choose quantities $q\in\R^n$, the market-clearing price is given by
\begin{equation}\label{eq:id}
    p = \nabla W_C(q) = Aq+c 1_n.
\end{equation}
    Comparing this equation with the demand curve in \cref{example:cournot} we see that we must have $A=-\frac 1 2 I_n + r(J_n-I_n)$ and $c=s\omega$, which implies \eqref{eq:cournot} under \eqref{eq:id}.

    We first solve for the consumer-optimal mechanism subject to the obedience constraint only. We can directly apply \cref{prop:linear} with $\alpha=-1$, $\beta=\gamma=1$. The resulting mechanism is as described in \cref{prop:linear} with $\sigma_{a_i\omega}^{\rm C}(\lambda)\eqdef y^{\rm C}(\lambda)\,s\sigma_{\omega}^2$, where
\begin{displaymath}
    y^{\rm C}(\lambda)
	    = \frac{1}{2n}\frac{\lambda nf+1+2r}{\lambda(f-r)+r},
	    \quad
	x^{\rm C}(\lambda)=
	    \frac{f}{2(1+r)}
	    \frac{\lambda nf -1}
	    {\big[\lambda nf(f-r)+(1+r)r\big]}.
\end{displaymath}
    This time, dual feasibility implies that $\lambda\geq \lambda_{\rm min}\eqdef\max\set{0, -\frac r {f-r}}$. The monotonicity of $y^{\rm C}(\lambda)$ is determined by the sign of the determinant
    \begin{displaymath}
	(1+f)r - (f-r)(1+2r) = (2r-f)(1+r),
    \end{displaymath}
    which is negative for $2r\in(-1, f)$, hence $y^{\rm C}$ is decreasing. Similarly, the monotonicity of $x^{\rm C}(\lambda)$ is determined by the sign of $r(1+r) +(f-r) = r^2+f$, which is positive.

    Observe that $\lambda_0=1/nf$ is solution to $x^{\rm C}(\lambda_0)=0$ and satisfies $\lambda_0> \lambda_{\rm min}$. This implies that the stationarity curve originates in the half plane $x< 0$. We distinguish two cases:
    \begin{enumerate}
	\item $r<0$. In this case, we verify that $y^{\rm C}(\lambda_0)> 2y_0$. That is, the point $\big(x^{\rm C}(\lambda_0), y^{\rm C}(\lambda_0)\big)$ is “above” the \emph{State-only} mechanism $\text{SO}$ and in particular lies outside \ref{eq:ellipse}. The optimal mechanism is thus obtained at the intersection of the stationarity curve with \ref{eq:ellipse}'s boundary. At this intersection we have $x^{\rm C}(\lambda^*)>0$ and $y^{\rm C}(\lambda^*)\geq 2y_0$. In particular, the second-order incentive compatibility condition $\sigma_{a_i\theta_j}/rt\sigma_{\theta_i}^2\geq 0$ is slack.
	\item $r>0$. In this case, we verify that $y^{\rm C}(\lambda_0)< 2y_0$. It is easy to see that this implies the optimal mechanism subject to obedience only always violates the incentive compatibility condition. We therefore need to explicitly add this constraint to the optimization problem, and it will be binding at the optimum. Hence, the only parameter left to determine is $y^{\rm C}$, which is either  $y^{\rm C}(0)$ when $y^{\rm C}(0)<2y_0$ (in which case the stationarity curve originates inside \ref{eq:ellipse}) or $2y_0$ when the curve originates outside \ref{eq:ellipse}.\qedhere
    \end{enumerate}
\end{proof}

These results highlight a notable structural difference between the consumer-optimal mechanism and the welfare- or revenue-optimal mechanisms: the direction of correlation in action recommendations is reversed. This reversal arises due to the change in the objective function's sign on the cross-action covariance, which alters the designer’s preference over action coordination.

When $r < 0$ (strategic substitutes), the consumer-optimal mechanism is always \emph{deterministic}, and the optimal mechanism subject only to obedience is incentive-compatible. Transfers are positive, and the action recommendations place more weight on the common state and less on the private types, relative to the welfare-optimal, revenue-optimal, and complete-information benchmarks.

When $r > 0$ (strategic complements), the IC constraint is \emph{always binding}, meaning the mechanism issues  recommendations that  do not depend on the competitors' types ($\sigma_{a_i\theta_j}= 0$) and can be implemented without discriminatory transfers.\footnote{The designer can extract a constant payment bounded by  $(\sigma_{a_i\omega}^2/\sigma_{\omega}^2+\delta^2)/2$ per \cref{prop:max-payment}.} Furthermore, we have two cases based on the comparison between $r$ and the complementarity threshold $r^{\star}\eqdef(-3+\sqrt{9+8f})/4$:
\begin{itemize}
    \item For $0 < r < r^{\star}$, the mechanism coincides with SO, i.e., it fully reveals the state. It is deterministic, and the action recommendations put more weight on the state than in the welfare- and revenue-optimal mechanisms.
    \item When $r^{\star} < r < f/2$, the optimal action recommendation is randomized, and hence does not fully reveal the state.
\end{itemize}

\begin{figure}[!t]
	\hfill{}
	\includegraphics[scale=1.1]{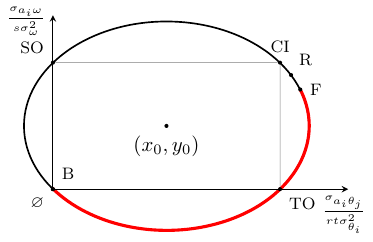}
	\hfill{}
	\includegraphics[scale=1.1]{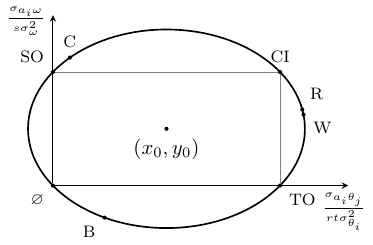}
	\hfill{}
    \caption{Optimal mechanisms for $r>0$ (left) and $r<0$ (right). The consumer-optimal mechanisms are denoted respectively by B and C, when our model is interpreted respectively as a Bertrand or Cournot competition.}
    \label{fig:consumer-opt-app}
\end{figure}

\section{Linear algebra}\label{sec:app-la}

\begin{proposition}\label{prop:matrix}
    For $n\geq 2$ and $(a,b)\in\R^2$, define $J_n(a,b)\eqdef aI_n + b (J_n-I_n)$,
	the matrix in $\M_n(\R)$ whose diagonal entries are all equal to $a$ and
	off-diagonal entries equal to $b$.
	\begin{enumerate}
		\item The determinant of $J_n(a,b)$ is
		    $\det J_n(a,b) =(a-b)^{n-1}\big(a+(n-1)b)$ and when $a\notin\set{-(n-1)b, b}$ its inverse is
			\begin{displaymath}
				J_n^{-1}(a,b)
				= \frac{J_n\big(a+(n-2)b, -b\big)}{(a-b)\big(a+(n-1)b\big)}.
			\end{displaymath}
		\item For a matrix $A\in\M_n(\R)$, $P_\pi A = AP_\pi$ for each
			$\pi\in\sg_n$ iff $A=J_n(a,b)$ for some $(a,b)\in\R^2$. In other
			words, $\spn(I_n,J_n)$ is the commutant of $\set{P_\pi\given
			\pi\in\sg_n}$.
		\item $\spn(I_n,J_n)$ is a commutative algebra.
		\item The matrix $J_n(a,b)$ is positive semidefinite iff $-a/(n-1)\leq
			b\leq a$.
	\end{enumerate}
\end{proposition}

\begin{proof}
The matrix $J_n$ is symmetric, hence we can diagonalize it as $J_n = UD\tr{U}$
where $U\in\M_n(\R)$ is orthogonal and $D\in\M_n(\R)$ is diagonal. Furthermore,
$J_n$ has rank $1$, hence its kernel has dimension $n-1$ and $D$ has $n-1$
zeros on its diagonal. Finally, $J_n 1_n = n 1_n$ shows that the remaining
eigenvalue of $J_n$ is $n$, with an associated eigenspace of dimension~$1$. We
can therefore write $D=\diag(0,\dots,0, n)$. Hence, for each $(a,b)\in\R^2$
\begin{equation}\label{eq:diag}
	J_n(a,b) = U\big((a-b)I_n + bD\big)\tr{U}
	=U\begin{bmatrix}
		a-b&0&\cdots&0\\
		0&\ddots&\ddots&\vdots\\
		\vdots&\ddots&a-b&0\\
		0&\cdots&0&a+(n-1)b\\
	\end{bmatrix}\tr{U}
\end{equation}
\begin{enumerate}
	\item We directly compute from \eqref{eq:diag}, $\det
		J_n(a,b)=(a-b)^{n-1}\big(a+(n-1)b\big)$, showing that $J_n(a,b)$ is 
		nonsingular iff $a\neq b$ and $a\neq -(n-1)b$. For such a pair $(a,b)$, we
		look for an inverse of $J_n(a,b)$ of the form $cI_n +d J_n$ for some
		$(c,d)\in\R^2$. Using the identity $J_n^2 = nJ_n$, we get
		\begin{displaymath}
			\big((a-b)I_n + bJ_n\big)(cI_n + dJ_n) = (a-b)c I_n
			+ \big(bc+d(a+(n-1)b)\big) J_n.
		\end{displaymath}
		Hence, $c I_n +  dJ_n$ is an inverse of $J_n(a,b)$ iff
		\begin{displaymath}
			c = \frac{1}{a-b}\quad\text{and}\quad
			d = -\frac{b}{(a-b)\big(a+(n-1)b\big)}
		\end{displaymath}
		yielding the stated expression for $J_n^{-1}(a,b)$.
	\item Let $A$ be a matrix commuting with all permutation matrices. For
		$i\in[n]$ and $j\in[n]\setminus\set{i}$, let $\tau$ be the
		transposition that swaps $i$ and $j$, and let $e_i$ be the $i$th
		standard basis vector. The condition $P_\tau A e_i = AP_\tau e_i$
		implies $a_{\tau(k)i} = a_{kj}$ for all $k\in[n]$. Writing this for
		$k\in\set{i,j}$ and $k\in[n]\setminus\set{i,j}$ (when $n\geq 3$) yields
		\begin{equation}\label{eq:foo-trans}
			a_{ii}=a_{jj},\qquad a_{ij} = a_{ji},\qquad a_{ki} = a_{kj}.
		\end{equation}
		The first equality shows that all diagonal entries are equal and the
		last equality shows that for each row $k$, all off-diagonal entries in
		row $k$ are equal. To compare off-diagonal entries in different rows,
		consider $i'\neq i$ and $j'\neq i'$, then
		\begin{displaymath}
			a_{ij} = a_{ii'} = a_{i'i} = a_{i'j'}
		\end{displaymath}
		where the first and last equalities used the third equality in
		\eqref{eq:foo-trans} and the middle equality used the second equality
		in \eqref{eq:foo-trans}.
		\looseness=-1
	\item By definition, $\spn(I_n, J_n)$ is a subspace of $\M_n(\R)$. The
		fact that it is a commutative algebra follows from the commutativity
		of $I_n$ and $J_n$ combined with the identity $J_n^2 = n J_n$.
	\item The eigenvalues of $J_n(a,b)$ can be read directly from
		\eqref{eq:diag}. The matrix $J_n(a,b)$ is positive
		semidefinite iff all its eigenvalues are nonnegative, that is, $a\geq
		b$ and $a\geq -(n-1)b$. \qedhere
\end{enumerate}
\end{proof}

\section{Symmetry}\label{sec:symmetry-supp}

\begin{definition}\label{def:invariance}
	Let $G$ be a group acting on a set $S$. We say that $x\in S$ is
	\emph{$G$-invariant} if $g\cdot x=x$ for all $g\in G$. Similarly,
	a function $f$ defined on $S$ is \emph{$G$-invariant} if $f(g\cdot x)
	= f(x)$ for all $(x,g)\in S\times G$. Finally, a subset $T\subseteq S$ is
	\emph{$G$-stable} if $g\cdot x\in T$ for all $(x,g)\in T\times G$.
\end{definition}

For $n\geq 1$, the symmetric group $\sg_n$ acts linearly on $\R^n$ by defining for $x\in\R^n$ and $\pi\in\sg_n$, $\pi\cdot x\eqdef \big(x_{\pi^{-1}(1)},\dots, x_{\pi^{-1}(n)}\big)$ . For $\pi\in\sg_n$, define the permutation matrix $P_\pi\in\M_n(\R)$ whose $(i,j)$ entry is $(P_\pi)_{ij} = \ind\set{\pi(j) = i}$, then the group action $(\pi,x)\mapsto \pi\cdot x$ can equivalently be defined as the (left) multiplication by the matrix $P_\pi$. Similarly the symmetric group $\sg_n$ acts linearly on $\M_n(\R)$ with $\pi\cdot A\eqdef P_\pi A \tr{P_\pi}$ for $\pi\in\sg_n$ and $A\in\M_n(\R)$.

The following lemma captures the core of our symmetrization argument below,
showing that restricting our study to symmetric mechanisms is without loss of
generality. Although the lemma is elementary, we were not able to find a
suitable formulation in the literature.

\begin{lemma}[Symmetrization]\label{lemma:symmetrization}
	Let $G$ be a finite group acting linearly on a real vector space $V$, and
	let $\cC$ be a convex and $G$-stable subset of $V$. Consider $x\in\cC$ and
	define
	\begin{displaymath}
		x_G \eqdef \frac{1}{|G|}\sum_{g\in G} g\cdot x.
	\end{displaymath}
	Then $x_G$ is a $G$-invariant element of $\cC$, and $f(x_G)\geq f(x)$ for
	every concave and $G$-invariant function $f:\cC\to\R$. Consequently, for
	such a function $f$, $\sup f(\cC) = \sup f(\cC_G)$, where $\cC_G$ denotes
	the $G$-invariant elements of $\cC$. If $f$ is furthermore affine, we have
	$f(\cC) = f(\cC_G)$.
\end{lemma}

\begin{proof}
	Because $\cC$ is $G$-stable, $g\cdot x\in\cC$ for all $g\in G$, hence
	$x_G\in\cC$ by convexity of $\cC$. Furthermore, $x_G$ is
	$G$-invariant: indeed, for each $h\in G$
	\begin{displaymath}
		h\cdot x_G = \frac{1}{|G|}\sum_{g\in G} (hg)\cdot
		x = \frac{1}{|G|}\sum_{g\in G} g\cdot x = x_G
	\end{displaymath}
	where the first equality uses that $G$ acts linearly on $V$ and the second
	equality uses that $g\mapsto hg$ is a permutation of $G$. Finally,
	\begin{equation}\label{eq:invariant-convex}
		f(x_G) = f\left(\frac{1}{|G|}\sum_{g\in G} g\cdot x\right)
		\geq \frac{1}{|G|}\sum_{g\in G} f(g\cdot x) = f(x)
	\end{equation}
	where the inequality is by concavity of $f$ and the second equality uses
	that $f$ is $G$-invariant.

	The claim $\sup f(\cC) = \sup f(\cC_G)$ then follows immediately when $f$
	is concave and $G$-invariant, and when $f$ is affine,
	\eqref{eq:invariant-convex} becomes an equality, implying that $f(\cC)
	= f(\cC_G)$.
\end{proof}

The following lemma shows that the set of obedient mechanisms is $\sg_n$-stable.

\begin{lemma}\label{lemma:obedience-symmetry}
	Assume that $(u_i)_{i\in[n]}$ defines a symmetric game and let
	$(a,\theta,\omega)$ be an obedient mechanism. Then, for all 
	$\pi\in\sg_n$, the permuted mechanism $(P_\pi a, P_\pi\theta, \omega)$ is
	also obedient.
\end{lemma}

\begin{proof}
	Consider a permutation $\pi\in\sg_n$, and an obedient mechanism
	$(a,\theta,\omega)$: for all $i\in[n]$, and
	$a'\in\R$
	\begin{displaymath}
	\E[u_i(a_i, a_{-i};\theta_i,\omega)\given a_i, \theta_i]
	\geq \E[u_i(a', a_{-i};\theta_i,\omega)\given a_i, \theta_i]\,.
	\end{displaymath}
	By symmetry of the game, the previous inequality is equivalent to
	\begin{displaymath}
		\E[u_{\pi(i)}(a_{j}, a_{-j};\theta_j,\omega)\given a_j, \theta_j]
		\geq \E[u_{\pi(i)}(a_j', a_{-j};\theta_j,\omega)\given a_j, \theta_j]\,
	\end{displaymath}
	where $j=\pi^{-1}(i)$. This inequality is exactly the obedience constraint
	of player $\pi(i)$ for the permuted mechanism $(P_\pi a,P_\pi \theta,
	\omega)$. Since this is true for all $i$, we conclude that the permuted
	mechanism is also obedient.
\end{proof}

\looseness=-1
Let $\cO$ be the set of all pairs $(\mu,\K)$ where $\mu$ and $\K$ denote
respectively the mean vector and covariance matrix of an obedient mechanism. Recall that by \cref{prop:obedience-char} of \cref{sec:char-app} (stated for non-symmetric covariance matrices), $\mu$'s value is fixed by \eqref{eq:obedience-mean}, and the covariance matrix must satisfy the $2n$ linear obedience constraints \eqref{eq:obedience-var} and positive semi-definiteness: $\K\in\S_{2n+1}^+(\R)$. Since $\mu$ is pinned down by obedience, it is effectively not a design parameter. We henceforth drop it from our
notations and simply use $\cO$ to denote the set of covariance matrices of
obedient mechanisms. A general optimization problem over obedient mechanisms
can therefore be written as
\begin{equation}\label{eq:opt-gen}
    \sup_{\K\in\cO} F(\K),
\end{equation}
where $F:\M_{2n+1}(\R)\to\R$. We focus on the case where $F$ is
a concave function, in which case \eqref{eq:opt-gen} defines a convex
optimization problem due to $\S_n^+(\R)$ being convex.

\looseness=-1
A \emph{symmetric} mechanism is one for which the distribution of $(a,\theta,\omega)$ is invariant under a relabeling of the players (cf.\ \cref{def:symmetric}). Let $\cO_s$ be the subset of $\cO$ consisting of covariance matrices of \emph{obedient and symmetric} mechanisms, that is the set of $\sg_n$-invariant elements of $\cO$. 

The following proposition shows that when both the prior distribution and the objective function $F$ in \eqref{eq:opt-gen} are symmetric (that is, $\sg_n$-invariant), then we can restrict to symmetric mechanisms in \eqref{eq:opt-gen} without loss of generality.

\begin{proposition}\label{prop:opt-symmetry}
    Assume that the prior distribution on players' types $\theta$ is symmetric, and let $F:\cO\to\R$ be a concave and symmetric function. Then
    \begin{displaymath}
	\sup_{\K\in\cO} F(\K) =\sup_{\K\in\cO_s} F(\K).
    \end{displaymath}
\end{proposition}

\begin{proof}
    As already observed above, the set $\cO$
    is convex. For $\K\in\M_{2n+1}(\R)$ and $\pi\in\sg_n$, let $\K^\pi$ be the
    matrix obtained by relabeling the players according to $\pi$, as in
    \eqref{eq:perm-params}. It is easy to see that $(\pi,\K)\mapsto \pi\cdot\K$
    defines a linear action of $\sg_n$ on $\M_{2n+1}(\R)$. Furthermore, $\cO$
    is stable under this action: indeed, if $\K$ is the covariance matrix of a
    Gaussian and obedient mechanism $(a,\theta,\omega)$, then $\K^\pi$ is covariance matrix of the permuted mechanism $(P_\pi a,
    P_\pi\theta,\omega)$, which is obedient by \cref{lemma:obedience-symmetry}.
    Applying \cref{lemma:symmetrization} thus concludes the proof.
\end{proof}


Note that the averaged mechanism resulting from \cref{lemma:symmetrization} is \emph{different} from the usual lottery that first draws a permutation of the players uniformly at random and then applies the original mechanism to this permutation. While a lottery would certainly preserve obedience (a mixture of obedient mechanisms is still obedient), the resulting mechanism would no longer be Gaussian, because a nontrivial mixture of normal distributions is not normal. In contrast, our symmetrization argument is expressed directly at the level of the moments of the mechanism and exploits the convexity of the constraints on the covariance matrix $\K$.

\begin{remark}\label{rem:compactness}
    \Cref{cor:obedience-char-bis} shows that the set $\cO_s$ is compact whenever $r\in\big(-1, \frac{1}{n-1}\big)$ which guarantees that the supremum in \cref{prop:opt-symmetry} is reached when $F$ is continuous. In contrast, when $r\notin\big(-1,\frac 1 {n-1}\big)$, the set $\cO_s$ is unbounded as we now show by explicitly constructing unbounded rays. Consider a random perturbation of the prior's Bayes–Nash equilibrium \eqref{eq:bne}:
    \begin{displaymath}
	a_i = \mu_{a_i} + t(\theta_i-\mu_{\theta_i}) + \delta \eps_i,
    \end{displaymath}
    where $\mu_{a_i}$ is given by \eqref{eq:obedience-mean} as in any obedient mechanism\footnote{At the endpoints $r\in\set{-1,\frac 1{n-1}}$, the linear system defining $\mu_{a_i}$ fails to be invertible. This is the limit case where obedient mechanisms, including the complete information Nash equilibrium no longer generically exist.} and with correlated random perturbations $\eps_i\sim\mathcal{N}(0,1)$. Writing $\rho\eqdef\cov(\eps_i,\eps_j)$, by \cref{prop:matrix}, we must have $-\frac 1 {n-1}\leq \rho\leq 1$ for the covariance matrix of $\eps$ to be PSD. Since $\sigma_{a_i\theta_j} = 0$ and $\sigma_{a_i\theta_i}=t\sigma_{\theta_i}^2$, the second covariance constraint is satisfied, and the first obedience constraint immediately implies that
    \begin{displaymath}
	\rho = \frac 1{(n-1)r},
    \end{displaymath}
    which satisfies $-\frac 1{n-1}\leq\rho\leq 1$ when $r\geq \frac 1 {n-1}$ or $r\leq -1$. In other words, the correlation coefficient $\rho$ is determined by obedience, but the variance $\delta^2$ of the noise is a free parameter that we can let grow to infinity. This shows that adding carefully correlated noise to the Bayes–Nash equilibrium lets us create obedient action recommendations with arbitrarily large variance. In this case, the mechanism is not providing the players with any information about the fundamentals but serves purely as a correlation device. Of course, this is only possible in “extreme” complements or substitutes situations. Finally, note that there is nothing specific about the Bayes–Nash equilibrium in the previous construction: the same random perturbation method works the same starting from any obedient mechanism.
\end{remark}

\end{document}